\documentclass[11pt]{article}

\usepackage{amsmath}
\usepackage{amssymb}
\usepackage{amsthm}
\usepackage{enumitem}
\usepackage{float}
\usepackage[margin=1in]{geometry}
\usepackage{hyperref}
\usepackage{multirow}
\usepackage{tikz}

\theoremstyle{plain}
\newtheorem{theorem}{Theorem}[section]

\newtheorem{corollary}[theorem]{Corollary}
\newtheorem{lemma}[theorem]{Lemma}
\newtheorem{proposition}[theorem]{Proposition}

\theoremstyle{definition}
\newtheorem{algorithm}[theorem]{Algorithm}
\newtheorem{definition}[theorem]{Definition}

\newcommand{\F}{\mathbb{F}}
\newcommand{\N}{\mathbb{N}}
\newcommand{\R}{\mathbb{R}}
\newcommand{\Z}{\mathbb{Z}}

\newcommand{\cD}{\mathcal{D}}
\newcommand{\cG}{\mathcal{G}}
\newcommand{\cH}{\mathcal{H}}
\newcommand{\cM}{\mathcal{M}}
\newcommand{\cQ}{\mathcal{Q}}
\newcommand{\cR}{\mathcal{R}}
\newcommand{\cS}{\mathcal{S}}
\newcommand{\cX}{\mathcal{X}}
\newcommand{\cY}{\mathcal{Y}}
\newcommand{\cZ}{\mathcal{Z}}

\newcommand{\tc}{\tilde{c}}

\newcommand{\teps}{\tilde{\varepsilon}}

\newcommand{\Geo}{\mathrm{Geo}}
\newcommand{\Lap}{\mathrm{Lap}}
\newcommand{\poly}{\mathrm{poly}}
\newcommand{\supp}{\mathrm{supp}}
\newcommand{\ctime}{\mathrm{Time}}
\newcommand{\unif}{\mathrm{Unif}}

\newcommand{\eps}{\varepsilon}
\newcommand{\zo}{\{0,1\}}

\newcommand{\vpref}[1]{\mbox{\protect\vphantom{hy}\ref{#1}}}

\newcommand{\overfullhypen}{-}
\title{Differential Privacy on Finite Computers\footnote{A condensed version of this paper appeared in ITCS 2018 \cite{BaVa18}}}
\author{Victor Balcer\thanks{
Supported by NSF grant CNS-1237235 and CNS-1565387.}
\hspace{3em}
Salil Vadhan\thanks{
\protect\url{salil.seas.harvard.edu}.
Supported by NSF grant CNS-1237235, a Simons Investigator Award, and a grant from the Sloan Foundation.}\\
Center for Research on Computation \& Society\\
School of Engineering \& Applied Sciences\\
Harvard University\\
\texttt{vbalcer@g.harvard.edu, salil\textunderscore vadhan@harvard.edu}}
\date{December 24, 2018}

\begin{document}

\maketitle
\begin{abstract}

We consider the problem of designing and analyzing differentially private algorithms that can be implemented on {\em discrete} models of computation in {\em strict} polynomial time, motivated by known attacks on floating point implementations of real-arithmetic differentially private algorithms (Mironov, CCS 2012) and the potential for timing attacks on expected polynomial-time algorithms.
As a case study, we examine the basic problem of approximating the histogram of a categorical dataset over a possibly large data universe $\mathcal{X}$.
The classic Laplace Mechanism (Dwork, McSherry, Nissim, Smith, TCC 2006 and J. Privacy \& Confidentiality 2017) does not satisfy our requirements, as it is based on real arithmetic, and natural discrete analogues, such as the Geometric Mechanism (Ghosh, Roughgarden, Sundarajan, STOC 2009 and SICOMP 2012), take time at least linear in $|\mathcal{X}|$, which can be exponential in the bit length of the input.

In this paper, we provide strict polynomial-time discrete algorithms for approximate histograms whose simultaneous accuracy (the maximum error over all bins) matches that of the Laplace Mechanism up to constant factors, while retaining the same (pure) differential privacy guarantee.
One of our algorithms produces a sparse histogram as output.
Its  ``per-bin accuracy'' (the error on individual bins) is worse than that of the Laplace Mechanism by a factor of $\log|\mathcal{X}|$, but we prove a lower bound showing that this is necessary for any algorithm that produces a sparse histogram.
A second algorithm avoids this lower bound, and matches the per-bin accuracy of the Laplace Mechanism, by producing a compact and efficiently computable representation of a dense histogram; it is based on an $(n+1)$-wise independent implementation of an appropriately clamped version of the Discrete Geometric Mechanism.
\end{abstract}

\section{Introduction}\label{sec:intro}

{\em Differential Privacy} \cite{DwMcNiSm06} is by now a well-established framework for privacy-protective statistical analysis
of sensitive datasets.
Much work on differential privacy involves an interplay between statistics and computer science.
Statistics provides many of the (non-private) analyses that we wish to approximate with differentially private algorithms, as well as probabilistic tools that are useful in analyzing such algorithms, which are necessarily randomized.
From computer science, differential privacy draws upon a tradition of adversarial modeling and strong security definitions, techniques for designing and analyzing randomized algorithms, and considerations of algorithmic resource constraints (such as time and memory).

Because of its connection to statistics, it is very natural that much of the literature on differential privacy considers the estimation of real-valued functions on real-valued data (e.g. the sample mean) and introduces noise from continuous probability distributions (e.g. the Laplace distribution) to obtain privacy.
However, these choices are incompatible with standard computer science models for algorithms (like the Turing machine or RAM model) as well as implementation on physical computers (which use only finite approximations to real arithmetic, e.g. via floating point numbers).
This discrepancy is not just a theoretical concern;
Mironov \cite{Mironov12} strikingly demonstrated that common floating-point implementations of the most basic differentially private algorithm (the Laplace Mechanism) are vulnerable to real attacks.
Mironov shows how to prevent his attack with a simple modification to the implementation, but this solution is specific to a single differentially private mechanism and particular floating-point arithmetic standard.
His solution increases the error by a constant factor and seems likely to be quite efficient in practice.
However, he provides no bounds on asymptotic running time.
Gazeau, Miller and Palamidessi \cite{GaMiPa13} provide more general conditions under which an implementation of real numbers and a mechanism that perturbs the correct answer with noise maintains differential privacy.
However, they do not provide an explicit construction with bounds on accuracy and running time.

From a theoretical point of view, a more appealing approach to resolving these issues is to avoid real or floating-point arithmetic entirely and only consider differentially private computations that involve discrete inputs and outputs, and rational probabilities, as first done in \cite{DwKeMcMiNa06}.
Such algorithms are realizable in standard discrete models of computation.
However, some such algorithms have running times that are only bounded in expectation (e.g. due to sampling from an exponential distribution supported on the natural numbers), and this raises a potential vulnerability to timing attacks.
If an adversary can observe the running time of the algorithm, it learns something about the algorithm's coin tosses, which are assumed to be secret in the definition of differential privacy.
(Even if the time cannot be directly observed, in practice an adversary can determine an upper bound on the running time, which again is information that is implicitly assumed to be secret in the privacy definition.)

Because of these considerations, we advocate the following principle:
\begin{quote}
\begin{center}
{\em \textbf{Differential Privacy for Finite Computers:}}
\end{center}
{\em We should describe how to implement differentially private algorithms on \textbf{discrete} models of computation with
\textbf{strict} bounds on running time (ideally polynomial in the bit length of the input) and \textbf{analyze} the effects of those constraints on both privacy and accuracy.}
\end{quote}
Note that a strict {\em bound} on running time does not in itself prevent timing attacks, but once we have such a bound, we can pad all executions to take the same amount of time.
Also, while standard discrete models of computation (e.g. randomized Turing machines) are defined in terms of countable rather than finite resources (e.g. the infinite tape), if we have a strict bound on running time, then once we fix an upper bound on input length, they can indeed be implemented on a truly finite computer (e.g. like a randomized Boolean circuit).

In many cases, the above goal can be achieved by appropriate discretizations and truncations applied to a standard, real-arithmetic differentially private algorithm.
However, such modifications can have a nontrivial price in accuracy or privacy, and thus we also call for a rigorous analysis of these effects.

In this paper, we carry out a case study of achieving ``differential privacy for finite computers'' for one of the first tasks studied in differential privacy, namely approximating a histogram of a categorical dataset.
Even this basic problem turns out to require some nontrivial effort, particularly to maintain strict polynomial time, optimal accuracy and pure differential privacy when the data universe is large.

We recall the definition of differential privacy.

\begin{definition}[\cite{DwMcNiSm06}]\label{def:intro dp}
Let $\cM: \cX^n \rightarrow \cR$ be a randomized algorithm.
We say $\mathcal{M}$ is {\bf $(\varepsilon,\delta)$-differentially private} if for every pair of
datasets $D$ and $D'$ that differ on one row
and every subset $S \subseteq \cR$
\begin{align*}
\Pr[\cM(D) \in S]
\le e^\varepsilon\cdot \Pr[\cM(D') \in S] + \delta
\end{align*}
\end{definition}

We say an $(\varepsilon, \delta)$-differentially private algorithm satisfies {\bf pure differential privacy} when $\delta = 0$ and say it satisfies {\bf approximate differential privacy} when $\delta > 0$.

In this paper, we study the problem of estimating the {\em histogram} of a dataset $D\in \cX^n$, which is the vector $c=c(D)\in \N^\cX$, where $c_x$ is the number of rows in $D$ that have value $x$.
Histograms can be approximated while satisfying differential privacy using the {\em Laplace Mechanism}, introduced in the original paper of Dwork, McSherry, Nissim and Smith \cite{DwMcNiSm06}.
Specifically, to obtain $(\eps,0)$-differential privacy, we can add independent noise distributed according to a Laplace distribution, specifically $\Lap(2/\eps)$, to each component of $c$ and output the resulting vector $\tc$.
Here $\Lap(2/\eps)$ is the {\em continuous}, real-valued random variable with probability density function $f(z)$ that is proportional to $\exp(-\eps\cdot|z|/2)$.
The Laplace Mechanism also achieves very high accuracy in two respects:

\begin{description}
\item[Per-Query Error:]
	For each bin $x \in \cX$, with high probability we have $|\tc_x-c_x| \leq O(1/\eps)$.
\item[Simultaneous Error:]
	With high probability, we have $\max_x |\tc_x-c_x| \leq O(\log (|\cX|)/\eps)$.
\end{description}

\noindent Note that both of the bounds are independent of the number $n$ of rows in the dataset, and so the fractional error vanishes linearly as $n$ grows.

Simultaneous error is the more well-studied notion in the differential privacy literature, but we consider per-query error to be an equally natural concept:
if we think of the approximate histogram $\tc$ as containing approximate answers to the $|\cX|$ different counting queries corresponding to the bins of $\cX$, then per-query error captures the error as experienced by an analyst who may be only interested in one or a few of the bins of $\tc$.
The advantage of considering per-query error is that it can be significantly smaller than the simultaneous error, as is the case in the Laplace Mechanism when the data universe $\cX$ is very large.
It is known that both of the error bounds achieved by the Laplace Mechanism are optimal up to constant factors;
no $(\eps,0)$-differentially private algorithm for histograms can achieve smaller per-query error or simultaneous error \cite{HaTa10,BeBrKaNi14}.

Unfortunately, the Laplace Mechanism uses real arithmetic and thus cannot be implemented on a finite computer.
To avoid real arithmetic, we could use the Geometric Mechanism \cite{GhRoSu12}, which adds noise to each component of $c$ according to the 2-sided geometric distribution, $\Geo(2/\eps)$, which is supported on the integers and has probability mass function $f(z)\propto \exp(-\varepsilon\cdot|z|/2)$.
However, this mechanism uses integers of unbounded size and thus cannot be implemented on a finite computer.
Indeed, while the algorithm can be implemented with a running time that is bounded in expectation (after reducing $\varepsilon$ so that $e^{\varepsilon/2}$ and hence all the probabilities are rational numbers), truncating long executions or allowing an adversary to observe the actual running time can lead to a violation of differential privacy.
Thus, as first described by Dwork, Kenthapadi, McSherry, Mironov and Naor \cite{DwKeMcMiNa06}, it is better to restrict the output of the mechanism to a binary representation of fixed length in order to avoid small tail probabilities.
Similarly, we work with the Truncated Geometric Mechanism of Ghosh, Roughgarden and Sundararajan \cite{GhRoSu12}, where we clamp each noisy count $\tc_x$ to the interval $[0,n]$.
We observe that the resulting probability distribution of $\tc_x$, supported on $\{0,1,\ldots,n\}$, can be described explicitly in terms of $c_x$, $\eps$ and $n$, and it can be sampled in polynomial time using only integer arithmetic (after ensuring $e^{\varepsilon/2}$ is rational).
Thus, we obtain:

\begin{theorem}[Bounded Geometric Mechanism, informal statement of Thm. \vpref{thm:explicit basic}]\label{thm:intro basic}
For every finite $\cX$, $n$ and $\eps \in (0,1]$, there is an $(\eps,0)$-differentially private algorithm $\cM : \cX^n\rightarrow \{0,1,\ldots,n\}^\cX$ for histograms achieving:

\begin{itemize}
\item
	Per-query error $O(1/\eps)$.
\item
	Simultaneous error $O(\log (|\cX|)/\eps)$.
\item
	Strict running time $\tilde{O}(|\cX|/\varepsilon)\cdot \log^2 n) + O(n \log n \cdot \log |\cX|)$.
\end{itemize}
\end{theorem}

We note that while we only consider our particular definition of per-query accuracy, namely that with high probability $|\tc_x-c_x|\leq O(1/\eps)$, Ghosh, Roughgarden and Sundararajan \cite{GhRoSu12} proved that the output of the Bounded Geometric Mechanism can be used (with post-processing) to get optimal expected loss with respect to an extremely general class of loss functions and arbitrary priors.
The same result applies to each individual noisy count
$\tc_x$ output by our mechanism, since each bin is distributed according to the Bounded Geometric Mechanism (up to a modification of $\eps$ to ensure rational probabilities).

The Bounded Geometric Mechanism is not polynomial time for large data universes $\cX$.
Indeed, its running time (and output length) is linear in $|\cX|$, rather than polynomial in the bit length of data elements, which is $\log |\cX|$.
To achieve truly polynomial time, we can similarly discretize and truncate a variant of the Stability-Based Histogram that was introduced by Korolova, Kenthapadi, Mishra and Ntoulas \cite{KoKeMiNt09}, and explicitly described by Bun, Nissim and Stemmer \cite{BuNiSt16}.
This mechanism only adds $\Lap(2/\eps)$ noise to the {\em nonzero} components of $c_x$ and then retains only the noisy values $\tc_x$ that are larger than a threshold $t=\Theta(\log(1/\delta)/\eps)$.
Thus, the algorithm only outputs a partial histogram, i.e. counts $\tc_x$ for a subset of the bins $x$, with the rest of the counts being treated as zero.
By replacing the use of the Laplace Mechanism with the (rational) Bounded Geometric Mechanism as above, we can implement this algorithm in strict polynomial time:

\begin{theorem}[Stability-Based Histogram,
informal statement of Thm. \vpref{thm:explicit stability}]
\label{thm:intro stability}
For every finite $\cX$, $n$, $\eps \in (0,1]$ and $\delta\in (0,1/n)$, there is an $(\eps,\delta)$-differentially private algorithm $\cM : \cX^n\rightarrow \{0,1,\ldots,n\}^{\subseteq \cX}$ for histograms achieving:

\begin{itemize}
\item
	Per-query error $O(1/\eps)$ on bins with true count at least $O(\log(1/\delta)/\eps)$.
\item
	Simultaneous error $O(\log(1/\delta)/\eps)$.
\item
	Strict running time $\tilde{O}((n/\varepsilon)\cdot \log (1/\delta)) + O(n \log n \cdot \log |\cX|)$.
\end{itemize}
\end{theorem}

Notice that the simultaneous error bound of $O(\log(1/\delta)/\eps)$ is better than what is achieved by the Laplace Mechanism when $\delta>1/|\cX|$, and is known to be optimal up to constant factors in this range of
parameters (see Theorem \vpref{thm:lower old}).
The fact that this error bound is independent of the data universe size $|\cX|$ makes it tempting to apply even for infinite data domains $\cX$.
However, we note that when $\cX$ is infinite, it is impossible for the algorithm to have a strict bound on running time (as it needs time to read arbitrarily long data elements) and thus is vulnerable to timing attacks and is not implementable on a finite computer.
Note also that the per-query error bound only holds on bins with large enough true count (namely, those larger than our threshold $t$);
we will discuss this point further below.

A disadvantage of the Stability-based Histogram is that it sacrifices pure differential privacy.
It is natural to ask whether we can achieve polynomial running time while retaining pure differential privacy.
A step in this direction was made by Cormode, Procopiuc, Srivastava and Tran \cite{CoPrSrTr12}.
They observe that for an appropriate threshold $t=\Theta(\log (|\cX|)/\eps)$, if we run the Bounded Geometric Mechanism and only retain the noisy counts $\tc_x$ that are larger than $t$, then the expected number of bins that remain is less than $n+1$.
Indeed, the expected number of bins we retain whose true count is zero (``empty bins'') is less than 1.
They describe a method to directly sample the distribution of the empty bins that are retained, without actually adding noise to all $|\cX|$ bins.
This yields an algorithm whose output length is polynomial in expectation.
However, the output length is not strictly polynomial, as there is a nonzero probability of outputting all $|\cX|$ bins.
And it is not clear how to implement the algorithm even in expected polynomial time, because even after making the probabilities rational, they have denominators of bit length linear in $|\cX|$.

To address these issues, we consider a slightly different algorithm.
Instead of trying to retain all noisy counts $\tc_x$ that are larger than some fixed threshold $t$, we retain the $n$ largest noisy counts (since there are at most $n$ nonzero true counts).
This results in a mechanism whose output length is always polynomial, rather than only in expectation.
However, the probabilities still have denominators of bit length linear in $|\cX|$.
Thus, we show how to approximately sample from this distribution, to within an arbitrarily small statistical distance $\delta$, at the price of a $\poly(\log(1/\delta))$ increase in running time.
Naively, this would result only in $(\eps,O(\delta))$-differential privacy.
However, when $\delta$ is significantly smaller than $1/|\cR|$, where $\cR$ is the range of the mechanism, we can convert an $(\eps,\delta)$-differentially private mechanism to an $(\eps,0)$-differentially private mechanism by simply outputting a uniformly random element of $\cR$ with small probability.
(A similar idea for the case that $|\cR|=2$ has been used in \cite{KaLeNiRaSm11,CaDaKa17}.)
Since our range is of at most exponential size (indeed at most polynomial in bit length), the cost in our runtime for taking $\delta\ll 1/|\cR|$ is at most polynomial.
With these ideas we obtain:
\begin{theorem}[Pure DP Histogram in Polynomial Time, informal statement of Thm. \vpref{thm:explicit pure}]\label{thm:intro puresparse}
For every finite $\cX$, $n$ and $\eps\in(0,1]$, there is an $(\eps,0)$-differentially private algorithm $\cM : \cX^n\rightarrow \{0,1,\ldots,n\}^{\subseteq \cX}$ for histograms achieving:

\begin{itemize}
\item
	Per-query error $O(1/\eps)$ on bins with true count at least $O(\log(|\cX|)/\eps)$.
\item
	Simultaneous error $O(\log(|\cX|)/\eps)$.
\item
	Strict running time $\tilde{O}\left(n^2 \cdot \log^2 |\cX| + n^2\cdot \log ({1}/{\eps}) + n \cdot \log |\cX|\cdot \log(1/\eps)\right)$.
\end{itemize}
\end{theorem}
It is an open problem as to whether or not one can improve the nearly quadratic dependence in running time on $n$ to nearly linear while maintaining the sparsity, privacy and accuracy guarantees achieved in Theorem \vpref{thm:intro puresparse}.

Both Theorems \vpref{thm:intro stability} and \vpref{thm:intro puresparse} only retain per-query error $O(1/\eps)$ on bins with a large enough true count.
We
also prove a lower bound showing that this limitation is inherent in any algorithm that outputs a sparse histogram (as both of these algorithms do).
\begin{theorem}[Lower Bound on Per-Query Error for Sparse
Histograms, Theorem \vpref{thm:lower pq}]
\label{thm:intro lower}
Suppose that there is an $(\eps,\delta)$-differentially private algorithm $\cM : \cX^n\rightarrow \{0,1,\ldots,n\}^{\cX}$ for histograms that always outputs histograms with at most $n'$ nonempty bins and has per-query error at most $E$ on all bins.
Then
\begin{align*}
E \geq \Omega\left(\frac{\min\{\log |\cX|,\log(1/\delta)\}}{\eps}\right)
\end{align*}
provided that $\eps>0$, $\eps^2>\delta>0$ and $|\cX|\geq (n')^2$.
\end{theorem}

This lower bound is similar in spirit to a lower bound of \cite{BeBrKaNi14}, which shows that no $(\eps,0)$-differentially private PAC learner for ``point functions'' (functions that are 1 on exactly one element of the domain) can produce sparse functions as hypotheses.

To bypass this lower bound, we can consider algorithms that produce succinct descriptions of dense histograms.
That is, the algorithm can output a polynomial-length description of a function $\tc : \cX \rightarrow [0,n]$ that can be evaluated in polynomial time, even though $\cX$ may be of exponential size.
We show that this relaxation allows us to regain per-query error $O(1/\eps)$.
\begin{theorem}[Polynomial-Time DP Histograms with Optimal Per-Query Accuracy, informal statement of Thm. \vpref{thm:explicit compact}]\label{thm:intro compact}
For every finite $\cX$, $n$ and $\eps\in(0,1]$, there is an $(\eps,0)$-differentially private algorithm $\cM : \cX^n\rightarrow \cH$ for histograms (where $\cH$ is an appropriate class of succinct descriptions of histograms) achieving:

\begin{itemize}
\item
	Per-query error $O(1/\eps)$.
\item
	Simultaneous error $O(\log (|\cX|)/\eps)$.
\item
	Strict running time $\tilde{O}\left((n/\eps)\cdot\log|\cX|\right)$.
\item
	Evaluating a count takes time $\tilde{O}\left((n/\eps)\cdot\log|\cX|\right)$.\footnote{In the original version of our paper \cite{BaVa18}, both the running time and evaluation time were missing logarithmic factors.}
\end{itemize}
\end{theorem}

The algorithm is essentially an $(n+1)$-wise independent instantiation of the Bounded Geometric Mechanism.
Specifically, we release a function $h : \cX\rightarrow \zo^r$ selected from an $(n+1)$-wise independent family of hash functions, and for each $x\in \cX$, we view $h(x)$ as coin tosses specifying a sample from the Bounded Geometric Distribution.
That is, we let $S : \zo^r\rightarrow [0,n]$ be an efficient sampling algorithm for the Bounded Geometric Distribution, and then $\tc_x=S(h(x))$ is our noisy count for $x$.
The hash function is chosen randomly from the family conditioned on  values $\tc_x$ for the nonempty bins $x$, which we obtain by running the actual Bounded Geometric Mechanism on those bins.
The $(n+1)$-wise independence ensures that the behavior on any two neighboring datasets (which together involve at most $n+1$ distinct elements of $\cX$) are indistinguishable in the same way as in the ordinary Bounded Geometric Mechanism.
The per-query accuracy comes from the fact that the marginal distributions of each of the noisy counts are the same as in the Bounded Geometric Mechanism.\footnote{Actually, we incur a small approximation error in matching the domain of the sampling procedure to the range of a family of hash functions.}

As far as we know, the only other use of limited independence in constructing differentially private algorithms is a use of pairwise independence by \cite{BeBrKaNi14} in differentially private PAC learning algorithms for the class of point functions.
Although that problem is related to the one we consider (releasing a histogram amounts to doing ``query release'' for the class of point functions, as discussed below), the design and analysis of our algorithm appears quite different.
(In particular, our analysis seems to rely on $(n+1)$-wise independence in an essential way.)

Another potential interest in our technique is as another method for bypassing limitations of {\em synthetic data} for {\em query release}.
Here, we have a large family of predicates $\cQ = \{q : \cX\rightarrow \zo\}$, and are interested in differentially private algorithms that, given a dataset $D=(x_1,\ldots,x_n)\in \cX^n$, output a ``summary'' $\cM(D)$ that allows one to approximate the answers to all of the {\em counting queries} $q(D) = \sum_i q(x_i)$ associated with predicates $q\in \cQ$.
For example, if $\cQ$ is the family of {\em point functions} consisting of all predicates that evaluate to 1 on exactly one point in the data universe $\cX$, then this query release problem amounts to approximating the histogram of $D$.
The fundamental result of Blum, Ligett, and Roth \cite{BlLiRo13} and successors show that this is possible even for families $\cQ$ and data universes $\cX$ that are of size exponential in $n$.
Moreover, the summaries produced by these algorithms has the form of a synthetic dataset --- a dataset $\hat{D}\in \cX^{\hat{n}}$ such that for every query $q\in \cQ$, we have $q(\hat{D})\approx q(D)$.
Unfortunately, it was shown in \cite{UlVa11} that even for very simple families $\cQ$ of queries, such correlations between pairs of binary attributes, constructing such a differentially private synthetic dataset requires time exponential in the bit length $\log |\cX|$ of data universe elements.
Thus, it is important to find other ways of representing approximate answers to natural families $\cQ$ of counting queries, which can bypass the inherent limitations of synthetic data, and progress along these lines was made in a variety of works \cite{GuRoUl12,ChKlKoLe12,HaRoSe12,
ThUlVa12,ChThUlWa14,DwNiTa15}.
Our algorithm, and its use of $(n+1)$-wise independence, can be seen as yet another representation that bypasses a limitation of synthetic data (albeit a statistical rather than computational one).
Indeed, a sparse histogram is simply a synthetic dataset that approximates answers to all point functions, and by Theorem \vpref{thm:intro lower}, our algorithm achieves provably better per-query accuracy than is possible with synthetic datasets.
This raises the question of whether similar ideas can also be useful in bypassing the computational limitations of synthetic data for more complex families of counting queries.

\section{Preliminaries}\label{sec:pre}

Throughout this paper, let $\N$ be the set $\{0, 1, \ldots\}$, $\N_+$ be the set $\{1, 2, \ldots\}$ and $\N^{-1}$ be the set $\{1/n : n \in \N_+\}$.
For $n \in \N_+$, let $[n]$ denote the set $\{0, \ldots, n\}$ and $[n]_+$ denote the set $\{1, \ldots, n\}$.
(Notice that $|[n]| = n+1$ while $|[n]_+| = n$.)
Given a set $A$ and finite set $B$, we define $A^B$ to be the set of length $|B|$ vectors over $A$ indexed by the elements of $B$.

\subsection{Differential Privacy}

We define a {\bf dataset} $D \in \cX^n$ to be an ordered tuple of $n \ge 1$ rows where each row is drawn from a discrete {\bf data universe} $\cX$ with each row corresponding to an individual.
Two datasets $D, D' \in \cX^n$ are considered {\bf neighbors} if they differ in exactly one row.

\begin{definition}[\cite{DwMcNiSm06}]
Let $\cM: \cX^n \rightarrow \cR$ be a randomized algorithm.
We say $\cM$ is {\bf$(\varepsilon,\delta)$-differentially private} if for every pair of neighboring datasets $D$ and $D'$ and every subset $S \subseteq \cR$
\begin{align*}
\Pr[\cM(D) \in S] \le e^\varepsilon\cdot \Pr[\cM(D') \in S] + \delta
\end{align*}
\end{definition}

We say an $(\varepsilon, \delta)$-differentially private algorithm satisfies {\bf pure differential privacy} when $\delta = 0$ and say it satisfies {\bf approximate differential privacy} when $\delta > 0$.
Intuitively, the $\varepsilon$ captures an upper bound on an adversary's ability to determine whether a particular individual is in the dataset.
And the $\delta$ parameter represents an upper bound of the probability of a catastrophic privacy breach (e.g. the entire dataset is released).
The common setting of parameters takes $\varepsilon \in (0,1]$ to be a small constant and $\delta$ to be negligible in $n$.

The following properties of differentially private algorithms will be used in some of our proofs.

\begin{lemma}[post-processing \cite{DwMcNiSm06}]\label{lem:dp postprocess}
Let $\cM: \cX^n \rightarrow \cY$ be $(\varepsilon,\delta)$-differentially private and $T: \cY \rightarrow \cZ$ be any randomized function.
Then $T \circ \cM: \cX^n \rightarrow \cZ$ is $(\varepsilon, \delta)$-differentially private.
\end{lemma}

\begin{lemma}[group privacy \cite{DwMcNiSm06}]\label{lem:dp group}
Let $\cM: \cX^n \rightarrow \cY$ be $(\varepsilon,\delta)$-differentially private.
Let $D_1, D_2 \subseteq \cX^n$ be datasets such that $D_2$ can be obtained by changing at most $m$ rows of $D_1$.
Then for all $S \subseteq \cY$
\begin{align*}
\Pr[\cM(D_1) \in S]
&\le e^{m\varepsilon} \cdot \Pr[\cM(D_2) \in S] + e^{m\varepsilon}\cdot\delta/\varepsilon
\end{align*}
\end{lemma}

\begin{lemma}[composition \cite{DwLe09}]\label{lem:dp composition}
Let $\cM_1: \cX^n \rightarrow \cY_1$ be $(\varepsilon_1,\delta_1)$-differentially private and $\cM_2: \cX^n \rightarrow \cY_2$ be $(\varepsilon_2,\delta_2)$-differentially private.
Define $\cM: \cX^n \rightarrow \cY_1 \times \cY_2$ as $\cM(x) = (\cM_1(x), \cM_2(x))$. Then $\cM$ is $(\varepsilon_1 + \varepsilon_2, \delta_1 + \delta_2)$-differentially private.
\end{lemma}

\subsection{Histograms}

For $x \in \cX$, the {\bf point function} $c_x : \cX^n \rightarrow \N$ is defined to count the number of occurrences of $x$ in a given dataset, i.e.
\begin{align*}
c_x(D) = |\{i \in [n]_+ : D_i = x\}|
\end{align*}
In this paper we focus on algorithms for privately releasing approximations to the values of all point functions, also known as a {\bf histogram}.
A histogram is a collection of {\bf bins}, one for each element $x$ in the data universe, with the $x^\text{th}$ bin consisting of its {\bf label} $x$ and a {\bf count} $c_x \in \N$.

\subsubsection{Representations}

The input to our algorithms is always a dataset (i.e. an element $D \in \cX^n$) and the outputs represent approximate histograms.
We consider the following histogram representations as our algorithms' outputs:

\begin{itemize}
\item
	A vector in $\N^{\cX}$.
	We use $\{\tc_x\}_{x \in \cX}$ to denote a histogram where $\tc_x \in \N$ is the approximate count for the element $x$.
\item
	A partial vector $h \in (\cX \times \N)^*$ such that each element $x \in \cX$ appears at most once in $h$ with each pair $(x, \tc_x) \in \cX \times \N$ interpreted as element $x$ having approximate count $\tc_x$.
	Elements $x$ not listed in the partial vector are assumed to have count $\tc_x = 0$.
	Implicitly, an algorithm can return a partial vector by releasing bins for a subset of $\cX$.\footnote{Note that the order in which bins are released can result in a breach of privacy (e.g. releasing the bins of elements in the dataset before the bins of elements not in the dataset).
	As a result, our algorithms always sort the released bins according to a predefined ordering based only on $\cX$.}
\item
	A data structure, encoded as a string, which defines a function $h: \cX \rightarrow \N$ where $h(x)$, denoted $h_x$, is the approximate count for $x \in \cX$ and $h_x$ is efficiently computable given this data structure (e.g. time polynomial in the length of the data structure).
	In Section \vpref{sec:compact}, this data structure consists of the coefficients of a polynomial, along with some parameters.
\end{itemize}

Each representation is able to express any histogram over $\cX$.
The difference between them is the memory used and the efficiency of computing a count.
For example, computing the approximate count for $x \in \cX$, when using the data structure representation is bounded by the time it takes to compute the associated function.
But when using partial vectors, one only needs to iterate through the vector to determine the approximate count.

We define the following class of histograms.
Let $\cH_{n,n'}(\cX) \subseteq \N^{\cX}$ be the set of all histograms over $\cX$ with integer counts in $[0,n]$ (or $\N$ when $n = \infty$) and at most $n'$ of them nonzero.
By using partial vectors each element of $\cH_{n,n'}(\cX)$ can be stored in $O(n'\cdot(\log n+\log|\cX|))$ bits, which is shorter than the vector representation when $n' = o(|\cX|/\log|\cX|)$.

\subsubsection{Accuracy}

In order to preserve privacy, our algorithms return histograms with noise added to the counts.
Therefore, it is crucial to understand their accuracy guarantees.
So given a dataset $D \in \cX^n$ we compare the {\bf noisy count} $\tc_x = \cM(D)_x$ of $x \in \cX$ (the count released by algorithm $\cM$) to its {\bf true count}, $c_x(D)$.
We focus on the following two metrics:

\begin{definition}
A histogram algorithm $\cM: \cX^n \rightarrow \N^{\cX}$ has \textbf{$(a, \beta)$-per-query accuracy} if
\begin{align*}
\forall D \in \cX^n ~~ \forall x \in \cX ~~~~\Pr[|\cM(D)_x - c_x(D)| \le a]
\ge 1-\beta
\end{align*}
\end{definition}

\begin{definition}
A histogram algorithm $\cM: \cX^n \rightarrow \N^{\cX}$ has \textbf{$(a, \beta)$-simultaneous accuracy} if
\begin{align*}
\forall D \in \cX^n~~~~\Pr[\forall x \in \cX ~~ |\cM(D)_x - c_x(D)| \le a]
\ge 1-\beta
\end{align*}
\end{definition}

Respectively, these metrics capture the maximum error for any one bin and the maximum error simultaneously over all bins.
Even though simultaneous accuracy is commonly used in differential privacy, per-query accuracy has several advantages:

\begin{itemize}
\item
	For histograms, one can achieve a smaller per-query error than is possible for simultaneous error.
	Indeed, the optimal simultaneous error for $(\varepsilon,0)$-differentially private histograms is $a = \Theta\left(\log(|\cX|/\beta)/\varepsilon\right)$ whereas the optimal per-query error is $a = \Theta\left(\log(1/\beta)/\varepsilon\right)$, which is independent of $|\cX|$ \cite{HaTa10,BeBrKaNi14}.
\item
	Per-query accuracy may be easier to convey to an end user of differential privacy.
	For example, it is the common interpretation of error bars shown on a graphical depiction of a histogram.
	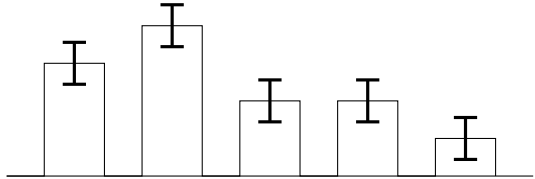
\begin{figure}[H]
	\centering
	\begin{tikzpicture}
	\draw (0,0) -- (7,0);
	\draw (0,0)
	-- (0.5,0) -- (0.5,1.5) -- (1.3,1.5) -- (1.3,0)
	-- (1.8,0) -- (1.8,2.0) -- (2.6,2.0) -- (2.6,0)
	-- (3.1,0) -- (3.1,1.0) -- (3.9,1.0) -- (3.9,0)
	-- (4.4,0) -- (4.4,1.0) -- (5.2,1.0) -- (5.2,0)
	-- (5.7,0) -- (5.7,0.5) -- (6.5,0.5) -- (6.5,0);

	\draw[|-|,very thick] (0.9,1.2) -- (0.9,1.8);
	\draw[|-|,very thick] (2.2,1.7) -- (2.2,2.3);
	\draw[|-|,very thick] (3.5,0.7) -- (3.5,1.3);
	\draw[|-|,very thick] (4.8,0.7) -- (4.8,1.3);
	\draw[|-|,very thick] (6.1,0.2) -- (6.1,0.8);
	\end{tikzpicture}
	\caption{A histogram with error bars}
	\end{figure}
\item
	For many algorithms (such as ours), per-query accuracy is good enough to imply optimal simultaneous accuracy.
	Indeed, an algorithm with $(a,\beta)$-per-query accuracy also achieves $(a, \beta \cdot |\cX|)$-simultaneous accuracy (by a union bound).
\end{itemize}

However, we may not always be able to achieve as good per-query accuracy as we want.
So we will also use the following relaxation which bounds the error only on bins with large enough true count.

\begin{definition}
A histogram algorithm $\cM: \cX^n \rightarrow \N^{\cX}$ has \textbf{$(a, \beta)$-per-query accuracy on counts larger than $t$} if
\begin{align*}
\forall D \in \cX^n ~~ \forall x \in \cX \text{ s.t. $c_x(D) > t$}~~~~\Pr[|\cM(D)_x - c_x(D)| \le a]
\ge 1-\beta
\end{align*}
\end{definition}

\subsection{Probability Terminology}

\begin{definition}\label{def:prob1}
Let $Z$ be an integer-valued random variable.
The {\bf probability mass function of $Z$}, denoted $f_Z$, is the function $f_Z(z) = \Pr[Z = z]$ for all $z \in \Z$.
The {\bf cumulative distribution function of $Z$}, denoted $F_Z$, is the function $F_Z(z) = \Pr[Z \le z]$ for all $z \in \Z$.
The {\bf support} of $Z$, denoted $\supp(Z)$, is the set of elements for which $f(z) \neq 0$.
\end{definition}

\begin{definition}\label{def:tvd}
Let $Y$ and $Z$ be random variables taking values in discrete range $\cR$. The {\bf statistical between $Y$ and $Z$} (a.k.a. total variation distance) is defined as
\begin{align*}
\Delta(Y, Z)
&
= \max_{A \subseteq \cR}\big|\Pr[Y \in A] - \Pr[Z \in A]\big|
\\ &
= \frac{1}{2}\cdot\sum_{a \in \cR}\big|\Pr[Z = a] - \Pr[Y = a]\big|
\end{align*}
\end{definition}

\begin{lemma}\label{lem:tvd}
Let $Y$ and $Z$ be random variables over discrete range $\cR$.
Statistical distance has the following properties:
\begin{enumerate}
\item
	$Y$ and $Z$ are identically distributed, denoted $Y \sim Z$, if and only if $\Delta(Y, Z) = 0$ (equivalently, $F_Y(z) = F_Z(z)$ for all $z \in \cR$).
\item
	Let $T: \cR \rightarrow \cR'$ be a randomized mapping with $\cR'$ discrete.
	Then
	\begin{align*}
	\Delta(T(Y),T(Z)) \le \Delta(Y,Z)
	\end{align*}
\item
	For $i \in \{1,2\}$, let $Y_i$ and $Z_i$ be random variables over discrete range $\cR_i$.
	Then
	\begin{align*}
	\Delta((Y_1,Y_2), (Z_1,Z_2))
	\le \Delta(Y_1,Z_1) + \max_{a\in \cR_1}\Delta(Y_2 | \{Y_1 = a\}, Z_2 | \{Z_1 = a\})
	\end{align*}
\end{enumerate}
\end{lemma}

\begin{definition}\label{def:orderstats}
Let $Z_1, \ldots, Z_\ell$ be integer-valued random variables.
The \textbf{$i$-th order statistic of $Z_1, \ldots, Z_\ell$} denoted $Z_{(i)}$ is the $i$-th smallest value among $Z_1, \ldots, Z_\ell$.
\end{definition}

\subsubsection{Sampling}

Because we are interested in the computational efficiency of our algorithms we need to consider the efficiency of sampling from various distributions.

A standard method for sampling a random variable is via {\bf inverse transform sampling}.
Let $\unif(A)$ denote the uniform distribution over the set $A$.

\begin{lemma}\label{lem:sampling continuous}
Let $U \sim \unif((0,1])$.
Then for any integer-valued random variable $Z$ we have $F^{-1}_Z(U) \sim Z$ where $F^{-1}_Z(u)$ is defined as $\min\{z \in \text{\rm supp}(Z) : F_Z(z) \ge u\}$.
\end{lemma}

If $Z$, the random variable we wish to sample, has finite support we can compute the inverse cumulative distribution by performing binary search on $\text{supp}(Z)$ to find the minimum.
This method removes the need to compute the inverse function of the cumulative distribution function.
If in addition, the cumulative distribution function of $Z$ can be represented by rational numbers, then we only need to sample from a discrete distribution instead of $(0,1]$.

\begin{lemma}\label{lem:sampling discrete}
Let $Z$ be an integer-valued random variable with finite support and has all probabilities of its cumulative distribution function expressible as rational numbers with denominator $d$.
Then $F^{-1}_Z(U) \sim Z$ where $U \sim (1/d) \cdot \unif([d]_+)$ and $F^{-1}_Z(u)$ is defined as $\min\{z \in \text{\rm supp}(Z) : F_Z(z) \ge u\}$.
\end{lemma}

\subsection{Model of Computation}\label{ssec:model}

We analyze the running time of our algorithms with respect to the {\bf $w$-bit word RAM model} taking $w$ logarithmic in our input length, namely $w = O(\log n + \log\log|\cX|)$.
In this model, memory accesses and basic operations (arithmetic, comparisons and logical) on $w$-bit words are constant time.
In addition, we assume the data universe is indexed so we can view $\cX = [m]_+$ for some $m \in \N$.
Some parameters to our algorithms are rational and we represent rationals by pairs of integers.
Some of our algorithms will use numbers that span many words.
For ease of notation, we will assume multiplication of two $x$-bit number is $\tilde{O}(x)$ (\cite{vzGaGe13} Theorem 8.24).

Our algorithms require randomness so we assume that they have access to an oracle that when given a number $d \in \N_+$ returns a uniformly random integer between $1$ and $d$ inclusive.

Finally, for representing histograms as partial vectors, we will assume internally to the algorithms that they are stored as red-black trees. This will allow us to insert and search for elements in $O(\log n \cdot \log |\cX|)$ time (\cite{CoLeRiSt09} Chapter 13).
When releasing a partial vector we don't return the tree itself (which may violate privacy), but instead return a list of bins using an in-order traversal of tree.

\section{A General Framework for Implementing Differential Privacy}

In this section, we outline a basic framework for implementing a pure differentially private algorithm $\cM$ on a finite computer with only a small loss in privacy and possibly a small loss in accuracy. 
It can be broken down into the following steps:

\begin{enumerate}
\item
	Start by discretizing the input and output of $\cM$ so that they can only take on a finite number of values (e.g. rounding a real-valued number to the nearest integer in some finite set). Depending on how utility is measured, the loss in accuracy by discretizing may be acceptable.
\item
	Then find an algorithm $\cM'$ that runs on a finite computer and approximates the output distribution of the discretized version of $\cM$ to within ``small'' statistical distance.
	Notice that $\cM'$ is only guaranteed to satisfy approximate differential privacy and may not satisfy pure differentially privacy.
	(This step may require a non-trivial amount of work.
	For one example, see Theorem $\vpref{thm:histogram sparse}$.)
\item
	Finally, provided that the statistical distance of the previous step is small enough, by mixing $\cM'$ with uniformly random output (from the discretized and finite output space), the resulting algorithm satisfies pure differential privacy.
\end{enumerate}

We will use this framework several times in designing our algorithms. Here we start by formalizing Step 3.
That is, for algorithms whose output distribution is close in statistical distance to that of a pure differentially private algorithm, we construct an algorithm satisfying pure differentially privacy by mixing it with random output inspired by similar techniques in \cite{KaLeNiRaSm11,CaDaKa17}.

\begin{figure}[H]
\rule{\textwidth}{.5pt}
\begin{algorithm}\label{alg:mixture}
$\cM^*_{\cM',\cD,\gamma}(D)$
for $D \in \cX^n$ where $\cR$ is discrete and finite, an algorithm $\cM': \cX^n \rightarrow \hphantom{\hspace{15em}} \cR$, a distribution $\cD$ over $\cR$ and $\gamma \in \N^{-1}$

\begin{enumerate}[leftmargin=.3in]
	\item
	With probability $1 - \gamma$ release $\cM'(D)$.
	\item
	Otherwise release an element sampled from the distribution $\cD$.
\end{enumerate}
\end{algorithm}
\vspace{-.1in}
\rule{\textwidth}{.5pt}
\end{figure}

\begin{lemma}\label{lem:mixture}
Suppose that there is an $(\varepsilon, 0)$-differentially private algorithm $\cM: \cX^n \rightarrow \cR$ such that $\Delta\left(\cM(D),\cM'(D)\right) \le \delta$ for all input datasets $D \in \cX^n$ with parameter $\delta \in [0,1)$.
Then the algorithm $\cM^*_{\cM',\cD,\gamma}: \cX^n \rightarrow \cR$ has the following properties:

\begin{enumerate}[label=\roman*.]
\item
	$(\varepsilon, 0)$-differential privacy whenever
	\begin{equation}\label{eq:mix inequality}
	\delta \le \frac{e^\varepsilon - 1}{e^\varepsilon + 1}\cdot\frac{\gamma}{1-\gamma}\cdot \min_{r \in \cR}\left\{\Pr_{Z \sim \cD}[Z = r]\right\}
	\end{equation}
\item
	Running time $O(\log(1/\gamma)) + \ctime(\cM') + \ctime(\cD)$ where $\ctime(\cD)$ is the time to sample from the distribution $\cD$.
\end{enumerate}
\end{lemma}

By taking $\gamma$ and $\delta$ small enough and satisfying \eqref{eq:mix inequality}, the algorithm $\cM^*_{\cM',\cD,\gamma}$ satisfies pure differential privacy and has nearly the same utility as $\cM$ (due to having a statistical distance at most $\gamma + \delta$ from $\cM$) while allowing for a possibly more efficient implementation since we only need to approximately sample from the output distribution of $\cM$.

To maximize the minimum in \eqref{eq:mix inequality}, one can take $\cD \sim \unif(\cR)$. However, it may the case that sampling this distribution exactly is inefficient and we are willing to trade needing a smaller $\delta$ to maintain pure differentially privacy for a faster sampling algorithm.

\begin{proof}[Proof of {\it i}]
Let $r \in \cR$ and $p = \Pr_{Z \sim \cD}[Z = r]$. Then for neighboring datasets $D, D' \in \cX^n$
\begin{align*}
\Pr[\cM^*_{\cM',\cD,\gamma}(D) = r]
&= \gamma\cdot p +\left(1-\gamma\right)\cdot \Pr[\cM'(D)=r]\\
&\le \gamma\cdot p +\left(1-\gamma\right)\cdot\left( \Pr[\cM(D)=r] + \delta\right)\\
&\le \gamma\cdot p +\left(1-\gamma\right)\left(e^\varepsilon\cdot\Pr[\cM(D')=r] + \delta \right)\\
&\le \gamma\cdot p +\left(1-\gamma\right)\left(e^\varepsilon\cdot\left(\Pr[\cM'(D')=r]+\delta\right) + \delta \right)
\end{align*}
Rearranging terms and using the upper bound on $\delta$ yields
\begin{align*}
\Pr[\cM^*_{\cM', \cD, \gamma}(D) = r]
&\le e^\varepsilon\left(1-\gamma\right)\cdot\Pr[\cM'(D')=r] + \gamma\cdot p + (e^\varepsilon + 1)\left(1-\gamma\right)\cdot\delta \\
&\le e^\varepsilon\left(1-\gamma\right)\cdot\Pr[\cM'(D')=r] + \gamma\cdot p + (e^\varepsilon - 1)\cdot\gamma\cdot \min_{r' \in \cR}\left\{\Pr_{Z \sim \cD}[Z = r']\right\}\\
&\le e^\varepsilon\left(\left(1-\gamma\right)\cdot\Pr[\cM'(D')=r] +\gamma\cdot p\right)\\
&= e^\varepsilon\cdot \Pr[\cM^*_{\cM',\cD,\gamma}(D')=r]\qedhere
\end{align*}
\end{proof}
\begin{proof}[Proof of {\it ii}]
This follows directly from the construction of $\cM^*$.
\end{proof}

\section{Counting Queries}\label{sec:count}

Before discussing algorithms for privately releasing histograms, we show how to privately answer a single counting query using only integers of bounded length.
While there exist known algorithms for this problem \cite{DwKeMcMiNa06,Mironov12}, our algorithms have additional properties that will be used to construct histogram algorithms in later sections.
In general, counting queries have as input the dataset $D \in \cX^n$ and the bin $x$ to query.
However, we will take the true count, $c_x(D)$, as the input to our counting query algorithms.
When constructing histogram algorithms in later sections, this will allow us to improve the running time as we will only need to iterate through the dataset once to determine all true counts prior to answering any counting query.
In addition, we would like to keep track of the randomness used by most of our algorithms so we write that as an explicit second input.
As a result, we have the following definitions:

\begin{definition}\label{def:count privacy}
Let $n, d \in \N_+$.
We say a (deterministic or randomized) algorithm $\cM: [n] \times [d]_+ \rightarrow [n]$ is {\bf
$(\varepsilon,\delta)$-differentially
private for counting queries} if the randomized algorithm $\cM: \zo^n \rightarrow [n]$ defined as $\cM(D) = \cM(c,U)$ where $c = \sum_{i=1}^{n}D_i$ and $U \sim \unif([d]_+)$ is $(\varepsilon,\delta)$-differentially private.
\end{definition}

\begin{definition}\label{def:count accuracy}
Let $n, d \in \N_+$.
We say $\cM: [n] \times [d]_+ \rightarrow [n]$ has {\bf $(a,\beta)$-accuracy} if for all $c \in [n]$
\begin{align*}
\Pr[|\cM(c,U) - c| \le a] \ge 1-\beta
\end{align*}
where $U \sim \unif([d]_+)$.
\end{definition}

\begin{definition}\label{def:count cdf}
Let $n, d \in \N_+$ and $\cM: [n] \times [d]_+ \rightarrow [n]$ be deterministic.
Let the {\bf scaled cumulative distribution function of $\cM$ at $0$} denoted $F_{\cM}$ be the function $F_{\cM}: [n] \rightarrow [d]_+$ defined as $F_{\cM}(z) = d\cdot F_{\cM(0,U)}(z)$ where $U \sim \unif([d]_+)$ for all $z \in [n]$.
\end{definition}

Being able to efficiently compute the scaled cumulative distribution function of $\cM$ at 0 will be a necessary property for constructing efficient histogram algorithms later (see Proposition \vpref{prop:approxordsample} and Lemma \vpref{lem:count conversion}).
Definition \vpref{def:count privacy} can easily be extended to handle point queries over data universes with more than two elements.

\begin{lemma}\label{lem:count privacy}
Let $\cM: [n] \times [d]_+ \rightarrow [n]$ be $(\varepsilon,\delta)$-differentially private for counting queries.
Let $D, D' \in \cX^n$ be neighboring datasets.
Then for all $x \in \cX$ and $c \in [n]$
\begin{align*}
\Pr[\cM(c_x(D), U) = c] \le e^{\varepsilon}\cdot\Pr[\cM(c_x(D'),U) = c] + \delta
\end{align*}
where $U \sim \unif([d]_+)$.
\end{lemma}

\begin{proof}
Define the dataset $D^{(x)} \in \zo^n$ as
\begin{align*}
D^{(x)}_i = \begin{cases}
1 & \text{if $D_i = x$}\\
0 & \text{otherwise}
\end{cases}
\end{align*}
Notice that $c_x(D) = \sum_{i=1}^{n}D^{(x)}_i$.
Similarly, we define the dataset $D'^{(x)}$ for $D'$.
Now, $D^{(x)}$ and $D'^{(x)}$ are neighboring datasets.
The lemma follows by $(\varepsilon,\delta)$-differential privacy for counting queries.
\end{proof}

\subsection{The Geometric Mechanism}

As shown by Dwork, McSherry, Nissim and Smith \cite{DwMcNiSm06}, we can privately release a counting query by adding appropriately scaled Laplace noise to the count.
Because our algorithm's outputs are counts, we do not need to use continuous noise and instead use a discrete analogue, as in \cite{DwKeMcMiNa06,GhRoSu12}.

We say an integer-valued random variable $Z$ follows a {\bf two-sided geometric distribution with scale parameter $s$ centered at $c \in \Z$} (denoted $Z\sim c + \Geo(s)$) if its probability mass function $f_Z(z)$ is proportional to $e^{-|z-c|/s}$.
It can be verified that $f_Z$ and its cumulative distribution function $F_Z$ are
\begin{align*}
f_Z(z)
= \left(\frac{e^{1/s}-1}{e^{1/s}+1}\right)\cdot e^{-|z-c|/s}
\hspace{.5in} F_Z(z)
= \begin{cases}
\frac{e^{1/s}}{e^{1/s} + 1}\cdot e^{-(c-z)/s} & \text{if $z \le c$} \\
1 - \frac{1}{e^{1/s} + 1} \cdot e^{-(z-c)/s} & \text{otherwise}
\end{cases}
\end{align*}
for all $z \in \Z$.
When $c$ is not specified, it is assumed to be 0.
The inverse cumulative distribution of $Z$ is
\begin{align*}
F^{-1}_Z(u)
&= c + \begin{cases}
\left\lceil s\ln\left(u\right) + s\ln\left(e^{1/s}+1\right)\right\rceil - 1 & \text{if $u \le 1/2$} \\
\left\lceil -s\ln\left(1-u\right) - s\ln\left(e^{1/s}+1\right)\right\rceil & \text{otherwise}
\end{cases}\\
\text{or, equivalently,} & \notag\\
F^{-1}_Z(u)
&= c + \left\lceil s \cdot \text{sign}(1/2 - u)\left(\ln(1 - |2u - 1|)+ \ln(e^{1/s} + 1) - \ln{2}\right)\right\rceil + \lfloor 2u \rfloor - 1
\end{align*}

Now, we state the counting query algorithm using discrete noise formally studied in \cite{GhRoSu12}.
We will not keep track of the randomness used by this algorithm, but to match our syntax for counting query algorithms we use the dummy parameter 1 as the second argument.

\begin{figure}[H]
\rule{\textwidth}{.5pt}
\begin{algorithm}\label{alg:geometricmechanism}
${\tt GeometricMechanism}_{n,\varepsilon}(c, 1)$
for $c \in [n]$ where $n \in \N_+$ and $\varepsilon > 0$

\begin{enumerate}[leftmargin=.3in]
\item
	Return $\tc$ set to $c + \Geo(2/\varepsilon)$ clamped to the interval $[0,n]$. i.e.
	\begin{align*}
	\hspace{50pt}
	\tc
	= \begin{cases}
	0 &\text{if $Z \le 0$}\\
	n &\text{if $Z \ge n$}\\
	Z &\text{otherwise}
	\end{cases} \text{~~~~where $Z = c + \Geo(2/\varepsilon)$.}
	\end{align*}
\end{enumerate}
\end{algorithm}
\vspace{-.1in}
\rule{\textwidth}{.5pt}
\end{figure}

\begin{theorem}\label{thm:geometricmechanism}
Let $n \in \N_+$ and $\varepsilon > 0$.
Then ${\tt GeometricMechanism}_{n,\varepsilon}: [n] \times [1]_+ \rightarrow [n]$ has the following properties:

\begin{enumerate}[label=\roman*.]
\item
	${\tt GeometricMechanism}_{n,\varepsilon}$ is $(\varepsilon/2,0)$-differentially private for counting queries \cite{GhRoSu12}.
\item
	${\tt GeometricMechanism}_{n,\varepsilon}$ has $(a, \beta)$-accuracy for $\beta \in (0,1]$ and
	\begin{align*}
	a = \left\lceil\frac{2}{\varepsilon}\cdot\ln\frac{1}{\beta}\right\rceil
	\end{align*}
\end{enumerate}
\end{theorem}

\begin{proof}[Proof of {\it ii}]
Let $Z \sim \Geo(2/\varepsilon)$.
Then for $c \in [n]$,
\begin{align*}
\Pr[|{\tt GeometricMechanism}_{n,\varepsilon}(c) - c | \le a]
&\ge \Pr\left[|Z| \le \lfloor a\rfloor\right]\\
&= 1 - 2\cdot\Pr[Z \le -\lfloor a \rfloor - 1]\\
&= 1 - 2\cdot \frac{e^{-\lfloor a \rfloor \cdot \varepsilon/2}}{e^{\varepsilon/2}+1}\\
&\ge 1 - \frac{2\cdot\beta}{e^{\varepsilon/2}+1}\\
&\ge 1 - \beta\qedhere
\end{align*}
\end{proof}

As presented above, this algorithm needs to store integers of unbounded size since $\Geo(2/\varepsilon)$ is unbounded in magnitude.
As noted in \cite{GhRoSu12}, by restricting the generated noise to a fixed range we can avoid this problem.
 However, even when the generated noise is restricted to a fixed range, generating this noise via inverse transform sampling may require infinite precision.
By appropriately choosing $\varepsilon$, the probabilities of this noise's cumulative distribution function can be represented with finite precision, and therefore generating this noise via inverse transform sampling only requires finite precision.

\begin{theorem}\label{thm:count geo}
Let $n \in \N_+$, $\varepsilon \in \N^{-1}$ and $\teps = 2\cdot \ln\left(1 + 2^{-\left\lceil \log(2/\varepsilon)\right\rceil}\right) \in (4/9 \cdot \varepsilon, \varepsilon]$.
Then there is a deterministic algorithm ${\tt GeoSample}_{n,\varepsilon}: [n] \times [d]_+ \rightarrow [n]$ where $\log d = O(n\cdot \log(1/\varepsilon))$ with the following properties:

\begin{enumerate}[label=\roman*.]
\item
	${\tt GeoSample}_{n,\varepsilon}(c, U) \sim {\tt GeometricMechanism}_{n,\teps}(c,1)$ for all $c \in [n]$ where $U \sim \unif([d]_+)$.
	Thus, ${\tt GeoSample}_{n,\varepsilon}$ is $(\teps/2, 0)$-differentially
	private for counting queries and has $(a, \beta)$-ac\overfullhypen curacy for $\beta \in (0,1]$ and
	\begin{align*}
	a = \left\lceil\frac{2}{\teps}\cdot\ln\frac{1}{\beta}\right\rceil
	\end{align*}
\item
	${\tt GeoSample}_{n,\varepsilon}$ has running time $\tilde{O}(n\cdot\log(1/\varepsilon))$.
\item
	For all $z \in [n]$, $F_{{\tt GeoSample}_{n,\varepsilon}}(z)$ can be computed in time $\tilde{O}(n\cdot\log(1/\varepsilon))$.
\item
	${\tt GeoSample}_{n,\varepsilon}(c,u)$ is a non-decreasing function in $u$.
\end{enumerate}
\end{theorem}

We have chosen $\teps$ so that the cumulative distribution function of a two-sided geometric random variable with scale parameter $2/\teps$ clamped to $[0,n]$ takes on only rational values with a common denominator $d$.
Therefore, to implement inverse transform sampling on this distribution we only need to choose a uniformly random integer from $[d]_+$ rather than a uniformly random variable over $(0,1]$ which allows us to provide a strict bound on the running time.

\begin{figure}[H]
\vspace{-.1in}
\rule{\textwidth}{.5pt}
\begin{algorithm}\label{alg:count geo}
${\tt GeoSample}_{n,\varepsilon}(c, u)$
for $c \in [n]$ and
$u \in [d]_+$ where
$n \in \N_{+}$ and
$\varepsilon \in \N^{-1}$

\begin{enumerate}[leftmargin=.3in]
\item
	Let $k = \left\lceil \log(2/\varepsilon)\right\rceil$ and $d = (2^{{k} +1} + 1)(2^{k} + 1)^{n-1}$.
\item
	For $z \in \Z$, define the function
	\begin{align*}
	F(z) = \begin{cases}
	0 & \text{if $z < 0$}\\
	2^{{k} (c-z)} \left(2^{{k} }+1\right)^{n-(c-z)} & \text{if $z \in [0,c)$}\\
	d- 2^{{k} (z-c+1)}\left(2^{k} + 1\right)^{n-1-(z-c)} & \text{if $z \in [c,n)$}\\
	d&\text{if $z \ge n$}
	\end{cases}
	\end{align*}
\item
	Using binary search find the smallest $z \in [n]$ such that $F(z) \ge u$.
\item
	Return $z$.
\end{enumerate}
\end{algorithm}
\vspace{-.1in}
\rule{\textwidth}{.5pt}
\end{figure}

The function $F$ is obtained by clearing denominators in the cumulative distribution function of $c + \Geo(2/	\teps)$ clamped to $[0,n]$.

\begin{lemma}\label{lem:count geo distribution}
Let $\teps$, $c$, $k$, $d$ and $F$ be defined as in Theorem \vpref{thm:count geo} and Algorithm \vpref{alg:count geo}.
Then $F(z) \in [d]$ and $F(z)/d$ equals the cumulative distribution function of $c + \Geo(2/\teps)$ clamped to $[0,n]$.
\end{lemma}

We prove this lemma after seeing how it implies Theorem \vpref{thm:count geo}.

\begin{proof}[Proof of Theorem \vpref{thm:count geo} Part {\it i}]
By construction, for all $z \in [n]$
\begin{align*}
\Pr[{\tt GeoSample}_{n,\varepsilon}(c, U) \le z] = \Pr[U \le F(z)] = F(z) / d
\end{align*}
implying ${\tt GeoSample}_{n,\varepsilon}(c, U) \sim {\tt GeometricMechanism}_{n,\teps}(c,1)$ by Lemma \vpref{lem:count geo distribution}.
\end{proof}

\begin{proof}[Proof of Theorem \vpref{thm:count geo} Part {\it ii}-{\it iv}]
Notice that integers used do not exceed $d$ whose bit length is $O(n \cdot \log(1/\varepsilon))$.
Thus, $F(z)$ can be computed in $\tilde{O}(n \cdot \log(1/\varepsilon))$ time using exponentiation by repeated squaring.
By construction, $F(z) = F_{{\tt GeoSample}_{n,\varepsilon}}(z)$, when $c=0$ implying Part {\it iii}.

Notice that the binary search of ${\tt GeoSample}_{n,\varepsilon}$ has at most $O(\log n)$ rounds each with an evaluation of $F$.
Thus, ${\tt GeoSample}_{n,\varepsilon}$ has the desired running time and by construction ${\tt GeoSample}_{n,\varepsilon}(c,u)$ is a non-decreasing function in $u$.
\end{proof}

\begin{proof}[Proof of Lemma \vpref{lem:count geo distribution}]
The cumulative distribution function of $Z \sim c + \Geo(2/\teps)$ is
\begin{align*}
F_Z(z)
&= \begin{cases}
0 &\text{if $z < 0$}\\
\frac{e^{\teps/2}}{e^{\teps/2} + 1}\cdot e^{-(c-z)\cdot\teps/2} & \text{if $z \in [0,c)$}\\
1 - \frac{1}{e^{\teps/2} + 1} \cdot e^{-(z-c)\cdot\teps/2} & \text{if $z \in [c, n)$}\\
1&\text{if $z \ge n$}\end{cases}
\end{align*}
Consider the case when $z \in [0,c)$.
\begin{align*}
F_Z(z)
= \frac{e^{\teps/2}}{e^{\teps/2} + 1}\cdot e^{-(c-z)\teps/2}
&= \frac{1 + 2^{-k}}{2 + 2^{-k}}\cdot \left(1 + 2^{-k}\right)^{-(c-z)}\\
&=\frac{2^{k} + 1}{2^{k + 1} + 1}\cdot \left(\frac{2^{k}}{2^{k} + 1}\right)^{c-z}\\
&=\frac{2^{k (c-z)}}{\left(2^{k + 1} + 1\right)\left(2^{k} + 1\right)^{c-z-1}}\cdot\left(\frac{2^k + 1}{2^k+1}\right)^{n-(c-z)}\\
&=\frac{{2^{k (c-z)}\left(2^{k} + 1\right)^{n + z - c}}}{d}\\
&= \frac{F(z)}{d}
\end{align*}
A similar argument holds for $z \in [c,n)$.
The remaining cases are trivial.
So $F_Z(z) = F(z) / d$ for all $z \in \Z$.
\end{proof}

\subsection{Approximating Geometric Noise to Release Counting Queries Faster}

Notice that ${\tt GeoSample}_{n,\varepsilon}$ has running time at least linear in $n$.
This is due to evaluating a (scaled) cumulative distribution function operating on integers with bit length $\Omega(n)$.
We can improve the running time by approximately sampling from a two-sided geometric distribution.
Small tail probabilities are dropped to reduce the number of required bits to represent probabilities to logarithmic in $n$.
And then to recover pure differential privacy, following Lemma \vpref{lem:mixture}, we mix with uniformly random output.

\begin{theorem}\label{thm:count fast}
Let $n \in \N_+$, $\varepsilon, \gamma \in \N^{-1}$ and $\teps = 2\cdot \ln\left(1 + 2^{-\left\lceil \log(2/\varepsilon)\right\rceil}\right) \in (4/9 \cdot \varepsilon, \varepsilon]$.
Then there is a deterministic algorithm ${\tt FastSample}_{n,\varepsilon,\gamma}: [n] \times [d]_+ \rightarrow [n]$ where $\log d = \tilde{O}(1/\varepsilon) \cdot \log(n/\gamma)$ with the following properties:

\begin{enumerate}[label=\roman*.]
\item
	${\tt FastSample}_{n,\varepsilon,\gamma}$ is $(\varepsilon/2, 0)$-differentially private for counting queries.
\item
	For every $\beta > \gamma$, ${\tt FastSample}_{n,\varepsilon,\gamma}$ has $(a, \beta)$-accuracy for
	\begin{align*}
	a = \left\lceil\frac{2}{\teps}\ln\frac{1}{\beta - \gamma}\right\rceil
	\end{align*}
\item
	${\tt FastSample}_{n,\varepsilon,\gamma}$ has running time
	\begin{align*}
	\tilde{O}\left(\frac{1}{\varepsilon}\cdot\log^2 n + \frac{1}{\varepsilon}\cdot\log n\cdot\log \frac{1}{\gamma}\right)
	\end{align*}
\item
	For all $z \in [n]$, $F_{{\tt FastSample}_{n,\varepsilon,\gamma}}(z)$ can be computed in time
	\begin{align*}
	\tilde{O}\left(\frac{1}{\varepsilon}\cdot\log \frac{n}{\gamma}\right)
	\end{align*}
\item
	${\tt FastSample}_{n,\varepsilon,\gamma}(c, u)$ is a non-decreasing function in $u$.
\end{enumerate}
\end{theorem}

\begin{figure}[H]
\rule{\textwidth}{.5pt}
\begin{algorithm}\label{alg:count fast}
${\tt FastSample}_{n,\varepsilon,\gamma}(c,u)$
for $c \in [n]$ and $u \in [d]_+$ where $n \in \N_+$ and $\varepsilon,\gamma \in \N^{-1}$

\begin{enumerate}[leftmargin=.3in]
\item
	Let $k = \lceil\log(2/\varepsilon)\rceil$ and $t = \left\lceil\frac{9}{2\varepsilon}\cdot\left\lceil\log\left(\frac{8(n+1)(1-\gamma)}{\varepsilon\gamma}\right)\right\rceil\right\rceil - 1$.
\item
	Let $d' = (2^{k+1}+1)(2^k+1)^{t}$ and $d = (n+1) \cdot d' / \gamma$.
\item
	For $z \in [n]$, define the functions
	\begin{align*}
	F'(z)
	&= \begin{cases}
	0 & \text{if $z < \max\{0, c - t\}$}\\
	2^{k(c-z)}(2^k+1)^{t+1-(c-z)} - 2^{k(t+1)} & \text{if $z \in [\max\{0, c - t\},c)$}\\
	d' - 2^{k(z-c+1)}(2^k + 1)^{t-(z-c)} + 2^{k(t+1)} & \text{if $z \in [c, \min\{c+t, n\})$}\\
	d' & \text{if $z \ge \min\{c+t, n\}$}
	\end{cases}
	\end{align*}
	and
	\begin{align*}
	F(z)
	&= (z+1) \cdot d' + (1/\gamma-1) \cdot (n+1) \cdot F'(z)
	\end{align*}
\item
	Using binary search find the smallest $z \in [n]$ such that $F(z) \ge u$.
\item
	Release $z$.
\end{enumerate}
\end{algorithm}
\vspace{-.1in}
\rule{\textwidth}{.5pt}
\end{figure}

As presented, it is clear how to compute the cumulative distribution function of this algorithm's output distribution, a necessary property for Section \vpref{sec:speed}.
And the algorithm, as a function of $u$, is non-decreasing, a property that will be used in Section \vpref{sec:compact}.
However, as stated, the interpretation of $F'$ and $F$ may not be clear. This is clarified by the following lemma and Figure \vpref{fig:count fast example}.

\begin{lemma}\label{lem:count fast distribution}
Let $k$, $t$, $d$, $d'$, $F'(z)$, $F(z)$ be defined as in Algorithm \vpref{alg:count fast}.
Then

\begin{enumerate}[label=\roman*.]
\item
	$F'(z) \in [d']$ and $F'(z)/d'$ is the cumulative distribution function of a random variable $Z'$ clamped to $[0,n]$ where $Z'$ has probability mass function
	\begin{align*}
	f_{Z'}(z)
	& =	\begin{cases}
	\Pr[Z = z] & \text{if $z \in [c-t, c+t]$ and $z \neq c$}\\
	\Pr[Z = c] + \Pr[|Z - c| > t] &\text{if $ z = c$}\\
	0 & \text{if $z \notin [c-t,c+t]$}
	\end{cases}
\end{align*}
for all $z \in \Z$ where $Z \sim c + \Geo(2/\teps)$.
\item
	$F(z) \in [d]$ and $F(z)/d$ is the cumulative distribution function of the random variable that with probability $1-\gamma$ is distributed as $Z'$ (defined in Part {\it i}) clamped to $[0,n]$ and with probability $\gamma$ is uniform over $[n]$.
\end{enumerate}
\end{lemma}

We prove this lemma after using it to prove Theorem \vpref{thm:count fast}.

\begin{figure}[H]
\centering
\begin{tikzpicture}
\tikzstyle{point}=[inner sep=.3ex, circle, fill=white, draw=black,thick]
\begin{scope}[xscale=.4]
\draw (-0.5,0) -- (11,0);
\node[anchor=south] at (5,-1.35) {(a)};
\draw (0,.1) -- (0,-.1) node[anchor=north] {\footnotesize $0$};
\draw (8,.1) -- (8,-.1) node[anchor=north] {\footnotesize $c$};
\draw (10,.1) -- (10,-.1) node[anchor=north] {\footnotesize $n$};
\node[point] (f0) at (0,0.039) {};
\node[point] (f1) at (1,0.039) {};
\node[point] (f2) at (2,0.078) {};
\node[point] (f3) at (3,0.156) {};
\node[point] (f4) at (4,0.312) {};
\node[point] (f5) at (5,0.625) {};
\node[point] (f6) at (6,1.25) {};
\node[point] (f7) at (7,2.5) {};
\node[point] (f8) at (8,5.0) {};
\node[point] (f9) at (9,2.5) {};
\node[point] (f10) at (10,2.5) {};
\draw[dashed] (f0) -- (f1) -- (f2) -- (f3) -- (f4) -- (f5) -- (f6) -- (f7) -- (f8) -- (f9) -- (f10);
\end{scope}
\begin{scope}[xscale=.4,xshift=5.2in]
\draw (-0.5,0) -- (11,0);
\node[anchor=south] at (5,-1.35) {(b)};
\draw (0,.1) -- (0,-.1) node[anchor=north] {\footnotesize $0$};
\draw (4,.1) -- (4,-.1) node[anchor=north] {\footnotesize $c-t$};
\draw (8,.1) -- (8,-.1) node[anchor=north] {\footnotesize $c$};
\draw (10,.1) -- (10,-.1) node[anchor=north] {\footnotesize $n$};
\node[point] (f0) at (0,0.0) {};
\node[point] (f1) at (1,0.0) {};
\node[point] (f2) at (2,0.0) {};
\node[point] (f3) at (3,0.0) {};
\node[point] (f4) at (4,0.312) {};
\node[point] (f5) at (5,0.625) {};
\node[point] (f6) at (6,1.25) {};
\node[point] (f7) at (7,2.5) {};
\node[point] (f8) at (8,5.625) {};
\node[point] (f9) at (9,2.5) {};
\node[point] (f10) at (10,2.188) {};
\draw[dashed] (f0) -- (f1) -- (f2) -- (f3) -- (f4) -- (f5) -- (f6) -- (f7) -- (f8) -- (f9) -- (f10);
\end{scope}
\begin{scope}[xscale=.4,xshift=10.4in]
\draw (-0.5,0) -- (11,0);
\node[anchor=south] at (5,-1.35) {(c)};
\draw (0,.1) -- (0,-.1) node[anchor=north] {\footnotesize $0$};
\draw (4,.1) -- (4,-.1) node[anchor=north] {\footnotesize $c-t$};
\draw (8,.1) -- (8,-.1) node[anchor=north] {\footnotesize $c$};
\draw (10,.1) -- (10,-.1) node[anchor=north] {\footnotesize $n$};
\node[point] (f0) at (0,0.273) {};
\node[point] (f1) at (1,0.273) {};
\node[point] (f2) at (2,0.273) {};
\node[point] (f3) at (3,0.273) {};
\node[point] (f4) at (4,0.523) {};
\node[point] (f5) at (5,0.773) {};
\node[point] (f6) at (6,1.273) {};
\node[point] (f7) at (7,2.273) {};
\node[point] (f8) at (8,4.773) {};
\node[point] (f9) at (9,2.273) {};
\node[point] (f10) at (10,2.023) {};
\draw[dashed] (f0) -- (f1) -- (f2) -- (f3) -- (f4) -- (f5) -- (f6) -- (f7) -- (f8) -- (f9) -- (f10);
\end{scope}
\end{tikzpicture}
\caption{Example probability mass functions of (a) {\tt GeoSample}, (b) $Z'$ clamped to $[0,n]$ as defined in Lemma \vpref{lem:count fast distribution} and (c) {\tt FastSample}. Notice that in this example $c+t > n$.}
\label{fig:count fast example}
\end{figure}
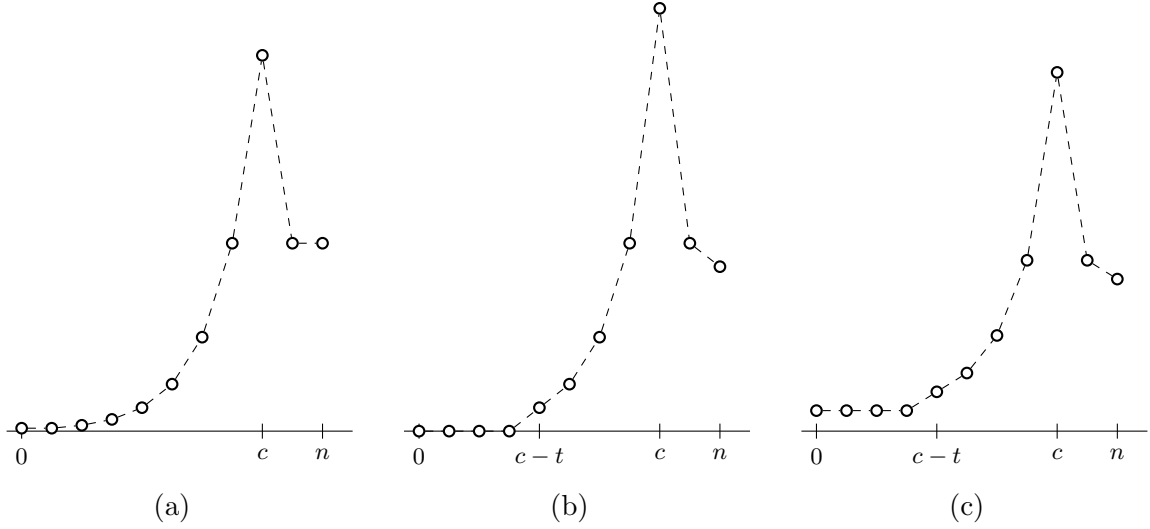

\begin{proof}[Proof of Theorem \vpref{thm:count fast} Part {\it i}]
By construction, for all $z \in [n]$
\begin{align*}
\Pr[{\tt FastSample}_{n,\varepsilon,\gamma}(c, U) \le z]
&= \Pr[U \le F(z)]
\end{align*}
where $U \sim \unif([d]_+)$. So ${\tt FastSample}_{n,\varepsilon,\gamma}(c, U)$ has cumulative distribution function $F(z)/d$. Then by Lemma \vpref{lem:count fast distribution}, ${\tt FastSample}_{n,\varepsilon,\gamma}(c, U)$ is a mixture that with probability $1-\gamma$ is distributed as $Z'$ clamped to $[0,n]$ as defined in Lemma \vpref{lem:count fast distribution} and otherwise is distributed as $\unif([n])$.

Since ${\tt GeoSample}_{n,\varepsilon}$ is $(\varepsilon/2,0)$-differentially private for counting queries (Theorem \vpref{thm:count geo}), we can use Lemma \vpref{lem:mixture} to prove that ${\tt FastSample}_{n,\varepsilon,\gamma}$ is also $(\varepsilon/2,0)$-differentially private for counting queries provided we can show the statistical distance between $Z'$ clamped to $[0,n]$ and ${\tt GeoSample}_{n,\varepsilon}$ is small enough.
\begin{align*}
\Delta(Z'\text{ clamped to }[0,n], {\tt GeoSample}_{n,\varepsilon}(c, U))
&\le \Delta(Z', Z)\\
&= \Pr[|Z-c| > t]\\
&= 2\cdot F_Z(c-t-1)\\
&= 2\cdot\frac{e^{\teps/2}}{e^{\teps/2} + 1}\cdot e^{-\teps \cdot (t+1)/2}\\
&\le \frac{4}{3} \cdot e^{-\frac{\teps}{2} \cdot \frac{9}{2\varepsilon} \cdot \ln\left(\frac{8(n+1)(1-\gamma)}{\varepsilon\gamma}\right)}\\
&\le \frac{\varepsilon}{6} \cdot \frac{\gamma}{1-\gamma}\cdot \frac{1}{n+1}\\
&\le \frac{e^{\varepsilon/2} - 1}{e^{\varepsilon/2} + 1} \cdot \frac{\gamma}{1-\gamma}\cdot \frac{1}{n+1}
\end{align*}
where $Z \sim \Geo(2/\teps)$.
\end{proof}

\begin{proof}[Proof of Theorem \vpref{thm:count fast} Part {\it ii}]
Let $U \sim \unif([d]_+)$. Let $Z'$ be defined as in Lemma \vpref{lem:count fast distribution} and $Z \sim c + \Geo(2/\teps)$.
Then
\begin{align*}
\Pr[|{\tt FastSample}_{n,\varepsilon,\gamma}(c, U) - c| \le a]
&\ge (1 - \gamma) \cdot \Pr[|Z' - c| \le a]\\
&\ge (1 - \gamma) \cdot \Pr[|Z - c| \le a]\\
&\ge (1 - \gamma) \cdot (1 - (\beta - \gamma))\\
&\ge 1 - \beta\qedhere
\end{align*}
\end{proof}

\begin{proof}[Proof of Theorem \vpref{thm:count fast} Part {\it iii}-{\it v}]
$d'$ can be computed in $\tilde{O}(t \cdot \log(1/\varepsilon))$ time using exponentiation by repeated squaring. Notice that integers used in computing $F'$ and $F$ do not exceed $d$ whose bit length is $O\left(t\cdot\log(1/\varepsilon) + \log(n/\gamma)\right)$.
Thus, $F(z)$ can be computed in time
\begin{align*}
\tilde{O}\left(t\cdot \log\frac{1}{\varepsilon} + \log \frac{n}{\gamma}\right)
\end{align*}
Part {\it iv} follows after observing $t = \tilde{O}(1/\varepsilon) \cdot \log(n/\gamma)$.

Now, the binary search of ${\tt FastSample}_{n,\varepsilon,\gamma}$ has at most $O(\log n)$ rounds each with an evaluation of $F(z)$ for some $z \in [n]$. Thus, we obtain the desired running time and by construction ${\tt FastSample}_{n,\varepsilon,\gamma}(c,u)$ is a non-decreasing function in $u$.
\end{proof}

\begin{proof}[Proof of Lemma \vpref{lem:count fast distribution} Part {\it i}]
Let $Z''$ be defined as $Z'$ clamped to $[0,n]$.
The cumulative distribution function of $Z''$ is
\begin{align*}
F_{Z''}(z)
&= \begin{cases}
0 & \text{if $z < \max\{0, c - t\}$}\\
F_Z(z) - F_Z(c-t-1) & \text{if $z \in [\max\{0, c - t\},c)$}\\
F_Z(z) + (1-F_Z(c+t)) & \text{if $z \in [c, \min\{c+t, n\})$}\\
1 & \text{if $z \ge \min\{c+t, n\}$}
\end{cases}
\end{align*}
where $Z \sim c + \Geo(2/\teps)$. Consider the case $z \in [\max\{0, c - t\},c)$.
\begin{align*}
F_{Z''}(z)
&= F_Z(z) - F_Z(c-t-1)\\
&= \left(\frac{2^k+1}{2^{k+1}+1}\right)\left(\frac{2^k}{2^k+1}\right)^{c-z} - \left(\frac{2^k+1}{2^{k+1}+1}\right)\left(\frac{2^k}{2^k+1}\right)^{t+1}\\
&= \frac{1}{d'}\cdot\left(2^{k(c-z)}(2^k+1)^{t+1-(c-z)} - 2^{k(t+1)}\right)\\
&= \frac{F'(z)}{d'}
\end{align*}
and $F'(z) \in [d']$ since $c-z < t+1$. A similar argument holds for $z \in [c, \min\{c+t, n\})$. The remaining cases are trivial. So $F_{Z''}(z) = F'(z)/d'$ for all $z \in \Z$.
\end{proof}
\begin{proof}[Proof of Lemma \vpref{lem:count fast distribution} Part {\it ii}]
Following from Part {\it i}, notice that for $z \in [n]$
\begin{align*}
\frac{F(z)}{d}
&= \gamma\cdot\frac{z+1}{n+1} + (1-\gamma)\cdot \frac{F'(z)}{d'}
= \gamma\cdot F_{U_{[n]}}(z) + (1-\gamma)\cdot \frac{F'(z)}{d'}
\end{align*}
where $U_{[n]} \sim \unif([n])$ which implies the desired result.
\end{proof}

\section{Generalizations of Known Histogram Algorithms}\label{sec:general}

In this section we show how to construct differentially private histograms within our finite model of computation given a private algorithm for releasing a single counting query.

\subsection{The Laplace Mechanism}\label{ssec:general laplace}

As shown by Dwork, McSherry, Nissim and Smith \cite{DwMcNiSm06}, we can privately release a histogram by adding independent and appropriately scaled Laplace noise to each bin.
Below we state a generalization guaranteeing privacy provided the counting query algorithm used is private and the released counts are independent.

\begin{figure}[H]
\rule{\textwidth}{.5pt}
\begin{algorithm}\label{alg:histogram basic}
${\tt BasicHistogram}_{\cM,A}(D)$
for $D \in \cX^n$ where
$\cM: [n] \times [d]_+ \rightarrow [n]$ and $A \subseteq \cX$

\begin{enumerate}[leftmargin=.3in]
\item
	Compute $c_x(D)$ for all $x \in A$.
\item
	For each $x \in A$, do the following:
	\begin{enumerate}[leftmargin=.3in]
	\item
		Sample $u_x$ uniformly at random from $[d]_+$.
	\item
		Let $\tc_x = \cM(c_x(D), u_x)$.
	\item
		Release $(x, \tc_x)$.
	\end{enumerate}
\end{enumerate}
\end{algorithm}
\vspace{-.1in}
\rule{\textwidth}{.5pt}
\end{figure}

Note that the
output of this algorithm is a collection of bins $(x, \tc_x)$ representing a partial vector.

\begin{theorem}\label{thm:histogram basic}
Let $\cM: [n] \times [d]_+ \rightarrow [n]$ be $(\varepsilon/2, 0)$-differentially private for counting queries and have $(a, \beta)$-accuracy. And let $A \subseteq \cX$.
Then ${\tt BasicHistogram}_{\cM, A}: \cX^n \rightarrow \N^{A}$ has the following properties:

\begin{enumerate}[label=\roman*.]
\item
	${\tt BasicHistogram}_{\cM,A}$ is $(\varepsilon,0)$-differentially private.
\item
	For all $D \in \cX^n$, we have
	\begin{align*}
	\forall x \in A ~~ \Pr[|({\tt BasicHistogram}_{\cM,A}(D))_x - c_x(D)| \le a]
	\ge 1 - \beta
	\end{align*}
	In particular, ${\tt BasicHistogram}_{\cM,\cX}(D)$ has $(a, \beta)$-per-query accuracy.
\item
	For all $D \in \cX^n$, we have
	\begin{align*}
	\Pr[\forall x \in A ~~ |({\tt BasicHistogram}_{\cM,A}(D))_x - c_x(D)| \le a]
	\ge 1 - \beta'
	\end{align*}
	where $\beta' = 1-(1-\beta)^{|A|} \le \beta \cdot |A|$.
	In particular, ${\tt BasicHistogram}_{\cM,\cX}$ has
	$(a, \beta')$-simultaneous accuracy where $\beta' = 1-(1-\beta)^{|\cX|} \le \beta \cdot |\cX|$.
\item
	${\tt BasicHistogram}_{\cM,A}$ has running time
	\begin{align*}
	O(n\log n \cdot \log |\cX|) + |A| \cdot O(\log n \cdot \log |\cX| + \log d + \ctime(\cM))
	\end{align*}
\end{enumerate}
\end{theorem}

It is important to note that the privacy guarantee only holds when $A$ is fixed and does not depend on the dataset $D$.
The choice of parameterizing by $A$ will be convenient in defining more complex histogram algorithms later.

\begin{proof}[Proof of {\it i}]
This proof follows similarly to the proof of privacy for the Laplace Mechanism.
Let $D, D' \in \cX^n$ be neighboring datasets and let $h \in \N^A$.
Then
\begin{align*}
\Pr[\forall x \in A ~~ \tc_x = h_x]
&= \prod_{x \in A} \Pr[\cM(c_x(D), u_x) = h_x ]&\text{by independence}
\end{align*}
Because there are at most two $x \in A$ for which $c_x(D) \neq c_x(D')$, by the $(\varepsilon/2,0)$-differential privacy for counting queries of $\cM$ with Lemma \vpref{lem:count privacy} and composition (Lemma \vpref{lem:dp composition}),
\begin{align*}
\prod_{x \in A} \Pr[\cM(c_x(D), u_x) = h_x ]
&\le e^{\varepsilon}\cdot\prod_{x \in A} \Pr[\cM(c_x(D'), u_x) = h_x ]\\
&\le e^{\varepsilon}\cdot\Pr[\forall x \in A ~~ \tc_x(D') = h_x]\qedhere
\end{align*}
\end{proof}

\begin{proof}[Proof of {\it ii}]
For all $x \in \cX$, $\tc_x$ is distributed as $\cM(c_x(D), u_x)$ where $u_x \sim \unif([d]_+)$.
The result follows from $\cM$ having $(a,\beta)$-accuracy.
\end{proof}

\begin{proof}[Proof of {\it iii}]
Simultaneous accuracy follows similarly as the bins' counts are independent.
So
\begin{align*}
\Pr[\forall x \in A ~~ |\tc_x - c_x(D)| \le a]
&= \prod_{x \in A}\Pr[|\tc_x - c_x(D)| \le a]
\ge (1-\beta)^{|A|}\qedhere
\end{align*}
\end{proof}

\begin{proof}[Proof of {\it iv}]
Computing $c_x(D)$ for all $x \in \cX$ can be accomplished by iterating through the dataset once and maintaining a partial vector with counts for the observed data elements.
This can be done in $O(n\log n\cdot\log|\cX|)$ time.
Each of the $|A|$ bin releases takes time $O(\log n \cdot \log |\cX| + \log d + \ctime(\cM))$ in order to get the true count, generate the randomness and then compute the noisy count.
\end{proof}

Now, we can use the counting query algorithms of Section \vpref{sec:count} to get an explicit instantiation of this algorithm.
If we take $\cM = {\tt GeometricMechanism}$, then ${\tt BasicHistogram}_{\cM,\cX}$ is identically distributed to the Truncated Geometric Mechanism of Ghosh, Roughgarden and Sundararajan \cite{GhRoSu12} which achieves per-query and simultaneous accuracy with error up to constant factors matching known lower bounds for releasing a private histogram \cite{HaTa10,BeBrKaNi14}.

\begin{theorem}\label{thm:explicit basic}
Let $\varepsilon, \beta_0 \in \N^{-1}$ and $\cM = {\tt FastSample}_{n,\varepsilon,\gamma}$ where $\gamma =  \beta_0/(2|\cX|)$.
Then
\begin{enumerate}[label=\roman*.]
\item
	${\tt BasicHistogram}_{\cM,\cX}$ is $(\varepsilon,0)$-differentially private.
\item
	For every $\beta \ge 2\gamma$, ${\tt BasicHistogram}_{\cM,\cX}$ has $(a,\beta)$-per-query accuracy for
	\begin{align*}
	a = \left\lceil\frac{9}{2\varepsilon}\ln\left(\frac{2}{\beta}\right)\right\rceil
	\end{align*}
\item
	For every $\beta \ge \beta_0$, ${\tt BasicHistogram}_{\cM,\cX}$ has $(a,\beta)$-simultaneous accuracy for
	\begin{align*}
	a = \left\lceil\frac{9}{2\varepsilon}\ln\left(\frac{2\cdot |\cX|}{\beta}\right)\right\rceil
	\end{align*}
\item
	${\tt BasicHistogram}_{\cM,\cX}$ has running time
	\begin{align*}
	\tilde{O}\left(\frac{|\cX|}{\varepsilon}\cdot \left(\log^2 n +\log n \cdot \log \frac{1}{\beta_0}\right)\right) + O(n \log n \cdot \log |\cX|)
	\end{align*}
\end{enumerate}
\end{theorem}

\begin{figure}[H]
\centering
\renewcommand{\arraystretch}{1.8}
\begin{tabular}{cccc}
$\cM$ & Running Time & $(a,\beta)$-Per-Query & $(a,\beta)$-Simul. \\
\hline
{\tt GeometricMechanism} &
n/a &
$\left\lceil\frac{2}{\varepsilon}\ln\frac{1}{\beta}\right\rceil$ &
$\left\lceil\frac{2}{\varepsilon}\ln\frac{|\cX|}{\beta}\right\rceil$\\
{\tt GeoSample} &
$\tilde{O}(|\cX|\cdot n\cdot \log(1/\varepsilon))$ &
$\left\lceil\frac{9}{2\varepsilon}\ln\frac{1}{\beta}\right\rceil$ &
$\left\lceil\frac{9}{2\varepsilon}\ln\frac{|\cX|}{\beta}\right\rceil$\\
{\tt FastSample} &
$\tilde{O}\left((|\cX|/\varepsilon)\cdot\log^2 n\right) + \tilde{O}(n) \cdot \log|\cX|$ &
$\left\lceil\frac{9}{2\varepsilon}\ln\frac{2}{\beta}\right\rceil$ &
$\left\lceil\frac{9}{2\varepsilon}\ln\frac{2|\cX|}{\beta}\right\rceil$
\end{tabular}
\caption{The running time and errors of ${\tt BasicHistogram}_{\cM,\cX}$ for the counting query algorithms of Section \vpref{sec:count}.
Values shown are for a $(\varepsilon,0)$-differentially private release where $\varepsilon \in \N^{-1}$.
For ${\tt FastSample}$, we assume $\beta_0 \ge 1/n^{O(1)}$ and $\gamma$ is specified in Theorem \vpref{thm:explicit basic}.}
\end{figure}

\subsection{Stability-Based Histogram}\label{ssec:general stability}

For a large data universe $\cX$, the at least linear in $|\cX|$ running time of ${\tt BasicHistogram}_{\cM,\cX}$ can be prohibitive.
By using approximate differential privacy, we can release counts for a smaller number of bins (at most $n$) based on stability techniques \cite{KoKeMiNt09, BuNiSt16}. 
We present a generalization of the algorithm from \cite{BuNiSt16}.

\begin{figure}[H]
\rule{\textwidth}{.5pt}
\begin{algorithm}\label{alg:histogram stability}
${\tt StabilityHistogram}_{\cM,b}(D)$
for $D \in \cX^n$
where $\cM: [n] \times [d]_+ \rightarrow [n]$ and $b \in [n]$

\begin{enumerate}[leftmargin=.3in]
\item
	Let $A = \{x \in \cX : c_x(D) > 0\}$.
\item
	Let $\{(x, \tc_x)\}_{x \in A} = {\tt BasicHistogram}_{\cM,A}(D)$.
\item
	Release $h = \left\{(x, \tc_x) : x \in A \text{ and } \tc_x > b\right\} \in \cH_{n,n}(\cX)$.
\end{enumerate}
\end{algorithm}
\vspace{-.1in}
\rule{\textwidth}{.5pt}
\end{figure}

Note that we only release counts for $x \in \cX$ whose true count is nonzero, namely elements in the set $A$.
Thus, the output length is $O(n \cdot (\log |\cX| + \log n))$.
However, releasing the set $A$ does not satisfy pure differential privacy because this would distinguish between neighboring datasets: one with a count of 0 and the other with a count of 1 for some element $x \in \cX$.
Thus, we only release noisy counts $\tc_x$ that exceed a threshold $b$.
If $b$ is large enough, then a count of 1 will only be kept with small probability, yielding approximate differential privacy.

\begin{theorem}\label{thm:histogram stability}
Let $\cM: [n] \times [d]_+ \rightarrow [n]$ be $(\varepsilon/2, 0)$-differentially private for counting queries and have $(a, \beta)$-accuracy.
And let $b \in [n]$.
Then ${\tt StabilityHistogram}_{\cM, b}: \cX^n \rightarrow \cH_{n,n}(\cX)$ has the following properties:
\begin{enumerate}[label=\roman*.]
\item
	${\tt StabilityHistogram}_{\cM, b}$ is $(\varepsilon, \delta)$-differentially private whenever
	\begin{align*}
	\delta \ge 2\cdot\Pr[\cM(1,U) > b]\text{ for $U \sim \unif([d]_+)$}
	\end{align*}
\item
	${\tt StabilityHistogram}_{\cM, b}$ has $(a,\beta)$-per-query accuracy on counts larger than $a + b$.
\item
	${\tt StabilityHistogram}_{\cM, b}$ has $(a+b,\beta')$-simultaneous accuracy where
	\begin{align*}
	\beta' = 1 - (1-\beta)^{n} \le n \cdot \beta
	\end{align*}
\item
	${\tt StabilityHistogram}_{\cM, b}$ has running time
	\begin{align*}
	O\left(n\log n \cdot \log |\cX| + n \cdot  \log d + n \cdot\ctime(\cM)\right)
	\end{align*}
\end{enumerate}
\end{theorem} 

\begin{proof}[Proof of {\it i}]
Let $D, D' \in \cX^n$ be neighboring datasets, $h \sim {\tt StabilityHistogram}_{\cM,b}(D)$ and $h' \sim {\tt StabilityHistogram}_{\cM,b}(D')$.
Let $U \sim \unif([d]_+)$.
Let $x \in \cX$ such that $c_x(D) \neq c_x(D')$ and let $S \subseteq [n]$.
There are 3 cases to consider:
\begin{itemize}
\item
	$c_x(D) \ge 1$ and $c_x(D') \ge 1$.
	By the privacy guarantee on $\cM$, we have $\Pr[\cM(c_x(D), U) \in S] \le e^{\varepsilon/2} \cdot \Pr[\cM(c_x(D'), U) \in S]$. Thus, by differential privacy's closure under post-processing (Lemma \vpref{lem:dp postprocess}),
	\begin{align*}
	\Pr[h_x \in S] \le e^{\varepsilon/2}\cdot\Pr[h_x' \in S]
	\end{align*}
\item
	$c_x(D) = 1$ and $c_x(D') = 0$.
	Notice that $\Pr[h_x' = 0] = 1$.
	So if $0 \in S$, then	$\Pr[h_x \in S] \le \Pr[h_x' \in S]$.
	If $0 \notin S$, then
	\begin{align*}
	\Pr[h_x \in S]
	\le \Pr[h_x > 0]
	&= \Pr[h_x > b]\\
	&= \Pr[\cM(1,U) > b]\\
	&\le \delta/2
	= \Pr[h_x' \in S] + \delta/2
	\end{align*}
\item
	$c_x(D) = 0$ and $c_x(D') = 1$.
	This case follows similarly to the previous one.
\end{itemize}
Then overall
\begin{align*}
\Pr[h_x \in S] \le e^{\varepsilon/2}\cdot\Pr[h_x' \in S] + \delta/2
\end{align*}
Because there are at most two bins on which $D$ and $D'$ have differing counts and each count $\tc_x$ is computed independently, by Lemma \vpref{lem:dp composition}, this algorithm is $(\varepsilon, \delta)$-differentially private.
\end{proof}

\begin{proof}[Proof of {\it ii}]
Let $U \sim \unif([d]_+)$ and $x \in \cX$ such that $c_x(D) > a+b$.
Notice that $|\tc_x - c_x(D)| \le a$ implies $\tc_x \ge c_x(D) - a > b$. Thus,
\begin{align*}
\Pr[|h_x - c_x(D)| \le a]
&= \Pr[|\tc_x - c_x(D)| \le a]\\
&= \Pr[|\cM(c_x(D), U) - c_x(D)| \le a]\\
&\ge 1-\beta\qedhere
\end{align*}
\end{proof}

\begin{proof}[Proof of {\it iii}]
Notice that the counts of elements not in $A$ are trivially accurate.
Therefore, we only need to consider the counts of elements in $A$.
By Theorem \vpref{thm:histogram basic} Part {\it iii}, ${\tt BasicHistogram}_{\cM,A}(D)$ has $(a, \beta')$-simultaneous accuracy where $\beta' = 1 - (1-\beta)^n$ as $|A| \le n$.

The final step of ${\tt StabilityHistogram}_{\cM,b}$ can increase the error on any count additively by at most $b$.
Therefore, ${\tt StabilityHistogram}_{\cM,b}$ has $(a + b, \beta')$-simultaneous accuracy.
\end{proof}

\begin{proof}[Proof of {\it iv}]
The running time follows from Theorem \vpref{thm:histogram basic} Part {\it iv} where $|A| \le n$ after noting the final step has running time $O(n\log n\cdot \log |\cX|)$.
\end{proof}

By taking $\cM = {\tt GeoSample}$, the simultaneous accuracy of ${\tt StabilityHistogram}_{\cM,b}$ matches known lower bounds.
However, we only achieve optimal per-query accuracy on sufficiently large counts.
This constraint is necessary for any algorithm outputting sparse histograms as we will show in Theorem \vpref{thm:lower pq}.
The following theorem shows the accuracies achieved by taking $\cM = {\tt FastSample}$.

\begin{theorem}\label{thm:explicit stability}
Let $\varepsilon, \delta, \beta_0 \in \N^{-1}$, $\cM = {\tt FastSample}_{n,\varepsilon,\gamma}$ where $\gamma = \min\{\beta_0 / (2n), \delta/4\}$ and $b = 1 + \left\lceil9/(2\varepsilon)\cdot\ln(4/\delta)\right\rceil$.
Then
\begin{enumerate}[label=\roman*.]
\item
	${\tt StabilityHistogram}_{\cM,b}$ is $(\varepsilon,\delta)$-differentially private.
\item
	For every $\beta \ge 2\gamma$, ${\tt StabilityHistogram}_{\cM,b}$ has $(a,\beta)$-per-query accuracy on counts larger than $t$ for
	\begin{align*}
	a = \left\lceil\frac{9}{2\varepsilon}\ln\left(\frac{2}{\beta}\right)\right\rceil
	\hspace{1em}\text{and}\hspace{1.25em}
	t = 2+\left\lceil\frac{9}{2\varepsilon}\ln\left(\frac{8}{\beta \cdot \delta}\right)\right\rceil
	\end{align*}
\item
	For every $\beta \ge \beta_0$, ${\tt StabilityHistogram}_{\cM,b}$ has $(a,\beta)$-simultaneous accuracy for
	\begin{align*}
	a = 2 + \left\lceil\frac{9}{2\varepsilon}\ln\left(\frac{8n}{\beta \cdot \delta}\right)\right\rceil
	\end{align*}
\item
	For $b \in [n]$, ${\tt StabilityHistogram}_{\cM,b}$ has running time
	\begin{align*}
	\tilde{O}\left(\frac{n}{\varepsilon}\cdot \log \frac{1}{\beta_0 \cdot \delta}\right) + O(n \log n \cdot \log |\cX|)
	\end{align*}
\end{enumerate}
\end{theorem}

\begin{figure}[H]
\centering
\renewcommand{\arraystretch}{1.8}
\begin{tabular}{ccc}
$\cM$ & $b$ & Running Time \\
\hline
{\tt GeometricMechanism} &
$1 + \left\lceil\frac{2}{\varepsilon}\ln\frac{2}{\delta}\right\rceil$ &
n/a \\
{\tt GeoSample} &
$1 + \left\lceil\frac{9}{2\varepsilon}\ln\frac{2}{\delta}\right\rceil$ &
$\tilde{O}(n^2\cdot\log(1/\varepsilon)) + O(n\log n \cdot \log |\cX|)$ \\
{\tt FastSample} &
$1 + \left\lceil\frac{9}{2\varepsilon}\ln\frac{4}{\delta}\right\rceil$ &
$\tilde{O}\left((n/\varepsilon)\cdot\log(1/(\beta_0\cdot\delta))\right) + O(n\log n \cdot \log |\cX|)$
\vspace{3ex}
\end{tabular}
\begin{tabular}{cccc}
\vspace{-.8em} & \multicolumn{2}{c}{$(a,\beta)$-Per-Query on $c_x(D) > t$} & \\
$\cM$ & $a$ & $t$ &$(a,\beta)$-Simultaneous\\
\hline
{\tt GeometricMechanism} &
$\left\lceil\frac{2}{\varepsilon}\ln\frac{1}{\beta}\right\rceil$ &
$2 + \left\lceil\frac{2}{\varepsilon}\ln\frac{2}{\beta\cdot \delta}\right\rceil$ &
$2 + \left\lceil\frac{2}{\varepsilon}\ln\frac{2n}{\beta\cdot \delta}\right\rceil$\\
{\tt GeoSample} &
$\left\lceil\frac{9}{2\varepsilon}\ln\frac{1}{\beta}\right\rceil$ &
$2 + \left\lceil\frac{9}{2\varepsilon}\ln\frac{2}{\beta\cdot \delta}\right\rceil$ &
$2 + \left\lceil\frac{9}{2\varepsilon}\ln\frac{2n}{\beta\cdot \delta}\right\rceil$\\
{\tt FastSample} &
$\left\lceil\frac{9}{2\varepsilon}\ln\frac{2}{\beta}\right\rceil$ &
$2 + \left\lceil\frac{9}{2\varepsilon}\ln\frac{8}{\beta\cdot \delta}\right\rceil$ &
$2 + \left\lceil\frac{9}{2\varepsilon}\ln\frac{8n}{\beta\cdot \delta}\right\rceil$
\end{tabular}
\caption{The running time and errors of ${\tt StabilityHistogram}_{\cM,b}$ for the counting query algorithms of Section \vpref{sec:count}.
Values shown are for a $(\varepsilon,\delta)$-differentially private release where $\varepsilon, \delta \in \N^{-1}$.
We assume $b \in [n]$ as otherwise the algorithm always outputs an empty histogram.
For ${\tt FastSample}$, the choice of parameters is specified in Theorem \vpref{thm:explicit stability}.}
\end{figure}

\section{Improving the Running Time}\label{sec:speed}

While {\tt StabilityHistogram} has running time logarithmic in the universe size, it can only guarantee approximate differential privacy.
In this section, we present an algorithm whose running time depends only poly-logarithmically on the universe size while maintaining pure differential privacy based on the observation that most counts are 0 when $n \ll |\cal{X}|$;
this is the same observation made by Cormode, Procopiuc, Srivastava and Tran \cite{CoPrSrTr12} to release private histograms that are sparse in expectation.

\subsection{Sparse Histograms}

We start by reducing the output length of ${\tt BasicHistogram}_{\cM,\cX}$ to release only the bins with the heaviest (or largest) counts (interpreted as a partial vector).

\noindent\rule{\textwidth}{.5pt}
\begin{algorithm}\label{alg:histogram heavy}
${\tt KeepHeavy}_{\cM}(D)$
for $D \in \cX^n$ where $\cM: [n] \times [d]_+ \rightarrow [n]$

\begin{enumerate}
\item
	Let $\{(x, \tc_x)\}_{x \in \cX} = {\tt BasicHistogram}_{\cM,\cX}(D)$.
\item
	Let $x_1, \ldots, x_{n+1}$ be the elements of $\cX$ with the largest counts in sorted order, i.e.
	\begin{align*}
	\tc_{x_1} \ge \tc_{x_2} \ge \ldots \ge \tc_{x_{n+1}} \ge \max_{x \in \cX\setminus\{x_1, \ldots, x_{n+1}\}} \tc_x
	\end{align*}
\item
	Release $h = \{(x, \tc_x): x \in \cX \text{ and } \tc_x > \tc_{x_{n+1}}\} \in \cH_{n,n}(\cX)$.\footnote{If instead $\cM$ had a continuous output distribution this last step is equivalent to releasing the $n$ heaviest bins.
	However, in the discrete case, where ties can occur, from the set $A\cup\{q_0, \ldots, q_n\}$ we cannot determine all bins with a count tied for the $n$-th heaviest as there may be many other noisy counts tied with $\tc_{x_n}$.
	As a result, we only output the bins with a strictly heavier count than $\tc_{x_{n+1}}$.}\vspace{-2ex}
\end{enumerate}
\end{algorithm}
\noindent\rule{\textwidth}{.5pt}

Observe that the output length has been improved to $O(n\cdot(\log |\cX| + \log n))$ bits compared to the $O(|\cX|\cdot(\log |\cX| + \log n))$ bits needed to represent the output of ${\tt BasicHistogram}_{\cM,\cX}$.

\begin{theorem}\label{thm:histogram heavy}
Let $\cM: [n] \times [d]_+ \rightarrow [n]$ be $(\varepsilon/2, 0)$-differentially private for counting queries such that ${\tt BasicHistogram}_{\cM,\cX}$ has $(a_1,\beta)$-per-query accuracy and $(a_2, \beta)$-simultaneous accuracy with $a_1 \le a_2$.
Then ${\tt KeepHeavy}_{\cM}: \cX^n \rightarrow \cH_{n,n}(\cX)$ has the following properties:

\begin{enumerate}[label=\roman*.]
\item
	${\tt KeepHeavy}_{\cM}$ is $(\varepsilon,0)$-differentially private.
\item
	${\tt KeepHeavy}_{\cM}$ has $(a_1, 2\beta)$-per-query accuracy on counts larger than $2a_2$.
\item
	${\tt KeepHeavy}_{\cM}$ has $(2a_2,\beta)$-simultaneous accuracy.
\end{enumerate}
\end{theorem}

Unlike ${\tt BasicHistogram}_{\cM,\cX}$, by taking $\cM = {\tt GeoSample}$, the algorithm ${\tt KeepHeavy}_\cM$ achieves $(O(\log(1/\beta)/\varepsilon), \beta)$-per-query accuracy only on counts larger than $O(\log(|\cX|/\beta)/\varepsilon)$.
This loss is necessary for any algorithm that outputs a sparse histogram
as we will show in Theorem \vpref{thm:lower pq}.

\begin{proof}[Proof of {\it i}]
Privacy follows from the $(\varepsilon,0)$-differential privacy of ${\tt BasicHistogram}_{\cM,\cX}$ (Theorem \vpref{thm:histogram basic} Part {\it i}) along with differential privacy's closure under post-processing (Lemma \vpref{lem:dp postprocess}).
\end{proof}

To prove the remaining parts, we start with the following lemma.

\begin{lemma}\label{lem:highlow}
Define the event $E = \{\forall x \in \cX~~|\tc_x - c_x(D)| \le a_2\}$.
Then $\Pr[E] \ge 1 - \beta$ and $E$ implies that for all $x \in \cX$ such that $c_x(D) > 2a_2$ we have $\tc_x > \tc_{x_{n+1}}$.
\end{lemma}

\begin{proof}
$\Pr[E] \ge 1-\beta$ follows from the $(a_2, \beta)$-simultaneous accuracy of ${\tt BasicHistogram}_{\cM,\cX}$.

Assume the event $E$.
Then for all $x \in \cX$ such that $c_x(D) > 2a_2$, we have $\tc_x > a_2$ and for all $x \in \cX$ such that $c_x(D) = 0$, we have $\tc_x \le a_2$.
Because there are at most $n$ distinct elements in $D$, we have $\tc_x > \tc_{x_{n+1}}$ for all $x \in \cX$ such that $c_x(D) > 2a_2$.
\end{proof}

\begin{proof}[Proof of Theorem \vpref{thm:histogram heavy} Part {\it ii}]
Let $x \in \cX$ such that $c_x(D) > 2a_2$.
We have
\begin{align*}
\Pr\left(|h_x- c_x(D)| > a_1\right)
&\le \Pr\left(\tc_x \le \tc_{x_{n+1}}\right)+\Pr\left(|\tc_x- c_x(D)| > a_1\right)
\le 2\beta
\end{align*}
by Lemma \vpref{lem:highlow} and the $(a_1,\beta)$-per-query accuracy of ${\tt BasicHistogram}_{\cM,\cX}$.
\end{proof}

\begin{proof}[Proof of Theorem \vpref{thm:histogram heavy} Part {\it iii}]
Assume the event $E$.
By Lemma \vpref{lem:highlow}, for all $x \in \cX$ such that $c_x(D) > 2a_2$ we have $\tc_x > \tc_{x_{n+1}}$ which implies $h_x = \tc_x$.
Thus, $|h_x - c_x(D)| \le a_2$.
For the remaining $x\in \cX$ we have $|h_x - c_x(D)| \le 2a_2$ as $h_x = 0$ or $h_x = \tc_x$.
\end{proof}

However, as described {\tt KeepHeavy} still requires adding noise to the count of every bin. The following algorithm ${\tt KH}'_{\cM}: \cX^n \rightarrow \cH_{n,n}(\cX)$ simulates ${\tt KeepHeavy}_{\cM}$ by generating a candidate set of heavy bins from which only the heaviest are released.
This candidate set is constructed from all bins with nonzero true count and a sample representing the bins with a true count of 0 that have the heaviest noisy counts.

\noindent\rule{\textwidth}{.5pt}
\begin{algorithm}\label{alg:histogram kh1}
${\tt KH}'_{\cM}(D)$
for $D \in \cX^n$ where $\cM: [n] \times [d]_+ \rightarrow [n]$ and $|\cX| \ge 2n+1$\footnote{$|\cX| \ge 2n+1$ ensures that $|\cX\setminus A| \ge n+1$. One can use ${\tt BasicHistogram}_{\cM,\cX}(D)$ when $|\cX| \le 2n$.}

\begin{enumerate}[leftmargin=.3in]
\item
	Let $A = \{x \in \cX : c_x(D) > 0\}$ and $m = |\cX\setminus A|$.
\item
	Let $\{(x, \tc_x)\}_{x \in A} = {\tt BasicHistogram}_{\cM,A}(D)$.
\item
	Pick a uniformly random sequence $(q_0,\ldots, q_n)$ of distinct elements from $\cX\setminus A$.
\item
	Sample $(\tc_{q_0}, \ldots, \tc_{q_n})$ from the joint distribution of the order statistics $(Z_{(m)}, \ldots,
Z_{(m - n)})$ where $Z_1, \ldots, Z_{m}$ are i.i.d. $\cM(0, U)$ random variables with $U \sim \unif([d]_+)$.
\item
 Sort the elements of $A \cup \{q_0, \ldots, q_n\}$ as $x_1, \ldots, x_{|A|+n+1}$ such that $\tc_{x_1} \ge \ldots \ge \tc_{x_{|A|+n+1}}$.
\item
	Release $h = \{(x, \tc_x): x \in \{x_1, \ldots, x_n\} \text{ and } \tc_x > \tc_{x_{n+1}}\} \in \cH_{n,n}(\cX)$.
\end{enumerate}
\end{algorithm}
\noindent\rule{\textwidth}{.5pt}

\begin{proposition}\label{prop:histogram kh1}
${\tt KH}'_{\cM}(D)$ is identically distributed to ${\tt KeepHeavy}_{\cM}(D)$.
\end{proposition}

\begin{proof}
Let $\{\hat{c}_x\}_{x\in\cX}$ be the noisy counts set by ${\tt KeepHeavy}_{\cM}(D)$ and let $\hat{x}_1, \ldots, \hat{x}_{n+1}$ be the sorted ordering of the $n+1$ heaviest bins defined by these counts.
We have $\tc_x \sim \hat{c}_x$ for all $x \in A$ and the $Z_i$'s are identically distributed to $\{\hat{c}_x\}_{x \in \cX\setminus A}$.

$\{(q_i, \tc_{q_i})\}_{i=0}^n$ is identically distributed to the $n+1$ bins with heaviest counts of $\{(x,\hat{c}_x)\}_{x \in \cX\setminus A}$ (breaking ties uniformly at random) due the noisy counts of the empty bins being independent and identically distributed.
Therefore,
\begin{align*}
h
&= \{(x, \tc_x): x \in \{x_1, \ldots, x_n\} \text{ and } \tc_x > \tc_{x_{n+1}}\}\\
&= \{(x, \tc_x) : x \in A\cup\{q_0, \ldots, q_n\} \text{ and }\tc_x > \tc_{x_{n+1}}\}\\
&\sim \{(x, \hat{c}_x) : x \in \cX \text{ and }\hat{c}_x > \hat{c}_{\hat{x}_{n+1}}\}
\end{align*}
which shows that ${\tt KH}'_{\cM}(D)$ is identically distributed to ${\tt KeepHeavy}_{\cM}(D)$.
\end{proof}

The original version of our paper \cite{BaVa18} used an incorrect method for sampling the order statistics based on a claim that their conditional CDF is given by:
\begin{align*}
F_{Z_{(m-k)} | Z_{(m-k+1)} = v_{m-k+1}, \ldots, Z_{(m)} = v_{m}}(z)
&= \begin{cases}
(F(z) / F(v_{m-k+1}))^{m-k}& \text{if $ z \le v_{m-k+1}$}\\
1& \text{otherwise}
\end{cases}
\end{align*}
Unfortunately, this only holds for continuous distributions.
Here we correct the error with a different method of sampling from the order statistics.
In order to sample from the order statistics used by ${\tt KH}'$ we construct an algorithm that samples from binomial distributions to construct a histogram of $Z_{(m)},\ldots, Z_{(m-n)}$.

\begin{proposition}\label{prop:ordsample}
Let $n, d \in \N_+$ and $F: [n] \rightarrow [d]_+$ such that $F$ is non-decreasing and $F(n) = d$.
Let $m \in \N_+$ such that $m \ge n+1$.
Let $Z_1, \ldots, Z_m$ be i.i.d. random variables over $[n]$ with cumulative distribution function $F(z)/d$ for all $z \in [n]$.
Then the following algorithm ${\tt OrdSample}_{F}(m)$ is identically distributed to the top $(n+1)$ order statistics $(Z_{(m)}, \ldots, Z_{(m-n)})$.
\end{proposition}

\begin{figure}[H]
\vspace{-.1in}
\rule{\textwidth}{.5pt}
\begin{algorithm}\label{alg:ordsample}
${\tt OrdSample}_F(m)$
for $m \in \N_+$ such that $m \ge n+1$ where $F: [n] \rightarrow [d]_+$ such that \hphantom{\hspace{16.25em}} $F$ is non-decreasing and $F(n) = d$

\begin{enumerate}[leftmargin=.3in]
\item
	Let $L_{n+1} = 0$.
\item
	For $v$ from $n$ to $1$, do the following:
	\begin{enumerate}[leftmargin=.3in]
	\item
		Sample $\ell_v \sim \min\{\mathrm{Bin}(m - L_{v+1}, p_v), n+1 - L_{v+1}\}$ where $p_v = 1 - \frac{F(v-1)}{F(v)}$.
	\item
		Let $L_{v} = L_{v+1} + \ell_v$.
	\end{enumerate}
\item
	Let $\ell_0 = n+1 - L_1$.
\item
	Return $(c_0, \ldots c_n)$ such that the first $\ell_n$ values are $n$, the next $\ell_{n-1}$ values are $n-1$ and so on until the last $\ell_0$ values are 0.\vspace{-.5ex}
\end{enumerate}
\end{algorithm}
\vspace{-.1in}
\rule{\textwidth}{.5pt}
\end{figure}

\begin{proof}
For $v \in [n]$, let $Y_v = |\{i : Z_i = v\}|$.
Using these random variables we can compute the order statistics $Z_{(m)}, \ldots, Z_{(1)}$ by taking the first $Y_n$ values to be $n$, the next $Y_{n-1}$ to be $n-1$ and so on.

Since we only need the top $n+1$ order statistics, it suffices for us to stop calculating the $Y_v$'s once their sum exceeds $n+1$.
That is, we can consider the random variables $Y'_v = \min\{Y_v, ~n+1 - \sum_{i=v+1}^n Y'_i\}$. Notice that
\begin{align*}
Y'_v
&= \begin{cases}
0 & \text{if $v < v^*$}\\
n+1 - \sum_{i=v^*+1}^n Y_i &\text{if $v = v^*$}\\
Y_v &\text{if $v > v^*$}
\end{cases}
&\text{ where } v^* = \max\{v \in [n] ~:~ \sum_{i=v}^{n} Y_i \ge n+1\}
\end{align*}
Therefore, using $(Y'_0, \ldots, Y'_n)$ we can compute the order statistics $(Z_{(m)}, \ldots, Z_{(m-n)})$ in the same manner as we did with the $Y_v$'s.

Observe that conditioned on $\{Y_n = y_n, Y_{n-1} = y_{n-1}, \ldots, Y_{v+1} = y_{v+1}\}$, we have $Y_v \sim \mathrm{Bin}(m - L_{v+1}, p_v)$ where $L_{v+1} = \sum_{i=v+1}^{n} Y_i$ is the number of ``assigned'' values of $Z_j$ and
\begin{align*}
p_v = \Pr[Z_j = v \,|\, Z_j \le v] = 1 - \Pr[Z_j \le v-1 \,|\, Z_j \le v] = 1 - \frac{F(v-1)}{F(v)}
\end{align*}

We now prove by induction from $v=n$ down to $v=1$ that
\begin{align*}
\ell_v \,|\, \ell_{v+1}, \ldots, \ell_n \sim Y_{v}' \,|\, Y_{v+1}',\ldots, Y_n'
\end{align*}
For the base case, observe that:
\begin{align*}
\ell_n
&\sim \min\{\mathrm{Bin}(m, p_n),\, n+1\}
\sim Y'_n
\end{align*}
For the induction step, we have that for $v \in [n-1]_+$,
\begin{align*}
\ell_v& ~|~ \ell_{v+1}, \ldots, \ell_n\\
&\sim \begin{cases}
\min\left\{\mathrm{Bin}(m - L_{v+1}, p_v), n+1 - L_{v+1}\right\}& \text{if $L_{v+1} < n+1$}\\
0 & \text{otherwise}
\end{cases}~|~ \ell_{v+1}, \ldots, \ell_n\\
&\sim \begin{cases}
\min\left\{\mathrm{Bin}(m - \sum_{i=v+1}^{n}Y'_i, p_v), n+1 - \sum_{i=v+1}^{n}Y'_i\right\}& \text{if $\sum_{i=v+1}^{n}Y'_i < n+1$}\\
0 & \text{otherwise}
\end{cases}~|~ Y'_{v+1}, \ldots, Y'_n\\
&\sim \begin{cases}
\min\left\{\mathrm{Bin}(m - \sum_{i=v+1}^{n}Y_i, p_v), n+1 - \sum_{i=v+1}^{n}Y'_i\right\}& \text{if $\sum_{i=v+1}^{n}Y'_i < n+1$}\\
0 & \text{otherwise}
\end{cases}~|~ Y'_{v+1}, \ldots, Y'_n\\
&\sim \min\left\{Y_v,~ n+1 - \sum_{i=v+1}^{n}Y'_i\right\} ~|~ Y'_{v+1}, \ldots, Y'_n\\
&\sim Y'_v ~|~ Y'_{v+1}, \ldots, Y'_n
\end{align*}
By construction, we have $\ell_0 ~|~ \ell_1, \ldots, \ell_n \sim Y'_0  ~|~ Y'_1, \ldots, Y'_n$.
Therefore, $(\ell_0, \ldots, \ell_n) \sim (Y'_0, \ldots, Y'_n)$ and so this algorithm returns values $(c_0, \ldots, c_n)$ distributed as $(Z_{(m)}, \ldots, Z_{(m-n)})$.
\end{proof}

Unfortunately, the running time of this algorithm is $\Omega(m)$ as it requires $\Omega(m)$ bits to represent the probabilities in $\mathrm{Bin}(m, p_v)$.
So while ${\tt KH}'_{\cM}$ only has an output of length $O(n\cdot(\log |\cX| + \log n))$, by using ${\tt OrdSample}_{F_\cM}$ to implement step 4, ${\tt KH}'_{\cM}$ has running time at least linear in $m \ge |\cX| - n$.
Indeed, this is necessary since the distribution of the order statistic $Z_{(m)}$ has probabilities that are exponentially small in
$m$.\footnote{Notice $\Pr[Z_{(m)} = 0] = \Pr[\mathcal{M}(0,U)]^m$.
Thus, we need $\Omega(m)$ random bits to sample from $Z_{(m)}$.}

\subsection{An Efficient Approximation}

To remedy the inefficiency of ${\tt KH}'$ we construct an efficient algorithm that approximates the output distribution of ${\tt KH}'$ by replacing {\tt OrdSample} with an efficient algorithm whose output distribution is close in statistical distance to that of {\tt OrdSample}, resulting in the following histogram algorithm:

\begin{theorem}\label{thm:histogram sparse}
Let $\varepsilon,\delta\in\N^{-1}$ and deterministic $\cM: [n] \times [d]_+ \rightarrow [n]$ be $(\varepsilon/2, 0)$-differentially private for counting queries.
Let $|\cX| \ge 2n + 1$.
Then there exists an algorithm ${\tt SparseHistogram}_{\cM,\delta}: \cX^n \rightarrow \cH_{n,n}(\cX)$ with the following properties:
\begin{enumerate}[label=\roman*.]
\item
	$\Delta\big({\tt KH}'_{\cM}(D), {\tt SparseHistogram}_{\cM,\delta}(D)\big) \le \delta$ for all $D \in \cX^n$.
\item
	${\tt SparseHistogram}_{\cM,\delta}$ is $\left(\varepsilon, (e^\varepsilon + 1)\cdot\delta\right)$-differentially private.
\item
	The running time of ${\tt SparseHistogram}_{\cM,\beta_1,\delta}$ is
	\begin{align*}
	\tilde{O}(n\cdot\log|\cX|& \cdot(\log |\cX| + \log (1/\delta)))
	+ O(n) \cdot \left(\tilde{O}(\log d) + \mathrm{Time}(\cM)+\mathrm{Time}(F_\cM)\right)
	\end{align*}
\end{enumerate}
\end{theorem}

Note that this algorithm only achieves $(\varepsilon, O(\delta))$-differential privacy.
By reducing $\delta$, the algorithm better approximates ${\tt KH}'$, at only the cost of increasing running time (polynomial in the bit length of $\delta$).
This is in contrast to most $(\varepsilon, \delta)$-differentially private algorithms such as the stability based algorithm of Section \vpref{ssec:general stability}, where one needs $n \ge \Omega(\log(1/\delta)/\varepsilon)$ to get any meaningful accuracy.

We will prove Theorem \vpref{thm:histogram sparse} in Section \vpref{ssec:construction}.
Before that we will convert {\tt SparseHistogram} to a pure differentially private algorithm by mixing it with random output, following Lemma \vpref{lem:mixture}.

\begin{figure}[H]
\rule{\textwidth}{.5pt}
\begin{algorithm}\label{alg:histogram puresparse}
${\tt PureSparseHistogram}_{\cM,\varepsilon,\beta_1}(D)$
for $D \in \cX^n$ where deterministic $\cM: [n] \times \hphantom{\hspace{24.25em}}[d]_+ \rightarrow [n]$, $\varepsilon,\beta_1 \in \N^{-1}$ s.t. $\beta_1 < 1$ and
\hphantom{\hspace{24.25em}}$|\cX| \ge 2n+1$
\begin{enumerate}[leftmargin=.3in]
\item
	With probability $1 - \beta_1$ release ${\tt SparseHistogram}_{\cM, \delta}(D)$ with
	\begin{align*}
	\delta
	= \frac{\varepsilon}{3}\cdot \beta_1 \cdot\left(\frac{1}{3\cdot|\cX|}\right)^n
	\end{align*}
\item
	Otherwise
	\begin{enumerate}
	\item
		Draw $(x_1, \ldots x_n)$ uniformly at random from $\cX^n$.
	\item
		Let $Q$ be the set of distinct elements from $(x_1, \ldots, x_n)$.
	\item
		For each $q \in Q$, sample $\tc_{q}$ uniformly at random from $[n]$.
	\item
		Release $h = \{(q, \tc_{q}) : q \in Q \text{ and } \tc_{q} > 0\} \in \cH_{n,n}(\cX)$.
	\end{enumerate}
\end{enumerate}
\end{algorithm}
\vspace{-.1in}
\rule{\textwidth}{.5pt}
\end{figure}

\begin{theorem}\label{thm:histogram puresparse}
Let $\varepsilon, \beta_1 \in \N^{-1}$ such that $\beta_1 < 1$
 and $|\cX| \ge 2n + 1$.
Let deterministic $\cM: [n] \times [d]_+ \rightarrow [n]$ be $(\varepsilon/2, 0)$-differentially private for counting queries such that ${\tt BasicHistogram}_{\cM,\cX}$ has $(a_1,\beta_2)$-per-query accuracy and $(a_2, \beta_2)$-simultaneous accuracy with $a_1 \le a_2$.

Then ${\tt PureSparseHistogram}_{\cM,\varepsilon,\beta_1}: \cX^n \rightarrow \cH_{n,n}(\cX)$ has the following properties:
\begin{enumerate}[label=\roman*.]
\item
	${\tt PureSparseHistogram}_{\cM,\varepsilon,\beta_1}$ is $(\varepsilon,0)$-differentially private.
\item
	${\tt PureSparseHistogram}_{\cM,\varepsilon,\beta_1}$ has $(a_1, 2\beta_1 + 2\beta_2)$-per-query accuracy on counts larger than $2a_2$.
\item
	${\tt PureSparseHistogram}_{\cM,\varepsilon,\beta_1}$ has $(2a_2, 2\beta_1 + \beta_2)$-simultaneous accuracy.
\item
	${\tt PureSparseHistogram}_{\cM,\varepsilon,\beta_1}$ has running time
	\begin{align*}
	\tilde{O}(n^2\log^2 |\cX| + n\cdot \log |\cX| \cdot \log(1/(\beta_1\eps))) + O(n) \cdot ( \tilde{O}(\log d) + \mathrm{Time}(\cM) + \mathrm{Time}(F_\cM))
	\end{align*}
\end{enumerate}
\end{theorem}

Notice that the running time of this algorithm depends nearly quadratically on $n$. It is an open problem as to whether or not one can improve the nearly quadratic dependence in running time on $n$ to nearly linear while maintaining the sparsity, privacy and accuracy guarantees achieved by this algorithm.
The nearly quadratic running time dependence on $n$ is due to approximating the distribution of each of the $O(n)$ order statistics in ${\tt KH}'$ to within a statistical distance that is exponentially small in $n$ (in order to apply Lemma \vpref{lem:mixture}).
See Section \vpref{ssec:construction} for more details.

\begin{proof}[Proof of {\it i}]
Define the distribution $\cD$ over $\cH_{n,n}(\cX)$ as the histogram returned by second step of {\tt PureSparseHistogram}. Notice that $\cD$ has full support over $\cH_{n,n}(\cX)$ and is determined by our sample from $\cX^n$ and $[n]^n$. Thus,
\begin{align*}
\min_{h' \in \cH_{n,n}(\cX)}\Pr_{h \sim \cD}[h = h'] \ge \frac{n!}{(n+1)^n \cdot {|\cX|}^{n}} \ge \left(\frac{1}{3 \cdot |\cX|}\right)^n
\end{align*}

Privacy follows from Theorem \vpref{thm:histogram sparse} and Lemma \vpref{lem:mixture} by taking $\cM = {\tt KH}'_{\cM}$, $\cR = \cH_{n,n}(\cX)$, $\cM' = {\tt SparseHistogram}_{\cM,\delta}$, $\gamma = \beta_1$ and $\cD = \cD$ as
\begin{align*}
\delta
&\le \frac{e^\varepsilon - 1}{e^\varepsilon + 1}\cdot\frac{\beta_1}{1-\beta_1}\cdot \min_{h' \in \cH_{n,n}(\cX)}\Pr_{h \sim \cD}[h = h'] \qedhere
\end{align*}
\end{proof}

\begin{proof}[Proof of {\it ii}-{\it iii}]
Let $D \in \cX^n$. For any $x \in \cX$ such that $c_x(D)>2a_2$ define the set $G = \{h \in \cH_{n,n}(\cX) : |h_x - c_x(D)| \le a_1 \}$. By construction,
\begin{align*}
\Pr[{\tt PureSparseHistoygram}_{\cM,\varepsilon,\beta_1}(D) \in G]
&\ge \Pr[{\tt SparseHistogram}_{\cM,\delta}(D) \in G] - \beta_1
\end{align*}
where $\delta$ is defined in Algorithm \vpref{alg:histogram puresparse}. Notice that $\delta \le \beta_1$. So by Theorem \vpref{thm:histogram sparse} Part {\it i},
\begin{align*}
\Pr[{\tt SparseHistogram}_{\cM,\delta}(D) \in G] - \beta_1
&\ge \Pr[{\tt KH}'_{\cM}(D) \in G] - 2\beta_1
\end{align*}
 And by Theorem \vpref{thm:histogram heavy} Part {\it ii} and Proposition \vpref{prop:histogram kh1}, we have
\begin{align*}
\Pr[{\tt KH}'_{\cM}(D) \in G] - 2\beta_1
&\ge 1 - 2\beta_1 - 2\beta_2
\end{align*}
Similarly, we can bound the simultaneous accuracy by using Theorem \vpref{thm:histogram heavy} Part {\it iii}.
\end{proof}

\begin{proof}[Proof of {\it iv}]
The second step of {\tt PureSparseHistogram} can be computed in $O(n\log n \cdot \log|\cX|)$ time.
The running time follows from Theorem \vpref{thm:histogram sparse} Part {\it iv} with
\begin{align*}
\log \frac{1}{\delta}
&= O\left(n \cdot\log |\cX| + \log\frac{1}{\beta_1 \cdot \varepsilon}\right)\qedhere
\end{align*}
\end{proof}

Now we can use {\tt PureSparseHistogram} with the counting query algorithms of Section \ref{sec:count}.

\begin{theorem}\label{thm:explicit pure}
Let $\varepsilon, \beta_0 \in \N^{-1}$ and $\cM = {\tt GeoSample}_{n,\varepsilon}$.
Then
\begin{enumerate}[label=\roman*.]
\item
	${\tt PureSparseHistogram}_{\cM,\varepsilon,\beta_0/4}$ is $(\varepsilon,0)$-differentially private.
\item
	For every $\beta \ge \beta_0$, ${\tt PureSparseHistogram}_{\cM,\varepsilon,\beta_0/4}$ has $(a,\beta)$-per-query accuracy on counts larger than $t$ for
	\begin{align*}
	a = \left\lceil\frac{9}{2\varepsilon}\ln\left(\frac{4}{\beta}\right)\right\rceil
	\hspace{1em}\text{and}\hspace{1.25em}
	t = 2\cdot\left\lceil\frac{9}{2\varepsilon}\ln\left(\frac{4\cdot|\cX|}{\beta}\right)\right\rceil
	\end{align*}
\item
	For every $\beta \ge \beta_0$, ${\tt PureSparseHistogram}_{\cM,\varepsilon,\beta_0/4}$ has $(a,\beta)$-simultaneous accuracy for
	\begin{align*}
	a = 2\cdot\left\lceil\frac{9}{2\varepsilon}\ln\left(\frac{2\cdot|\cX|}{\beta}\right)\right\rceil
	\end{align*}
\item
	${\tt PureSparseHistogram}_{\cM,\varepsilon,\beta_0/4}$ has running time
	\begin{align*}
	\tilde{O}\left(n^2 \cdot \log^2 |\cX| + n^2\cdot \log ({1}/{\eps}) + n \cdot \log |\cX|\cdot \log(1/(\beta_0\cdot\eps))\right)
	\end{align*}
\end{enumerate}
\end{theorem}

\begin{figure}[H]\label{fig:hist table pure}
\centering
\renewcommand{\arraystretch}{1.8}
\begin{tabular}{cc}
$\cM$ & Running Time\\
\hline
{\tt GeoSample} &
$\tilde{O}\left(n^2 \cdot \log^2 |\cX| + n^2\cdot \log ({1}/{\eps}) + n \cdot \log |\cX|\cdot \log(1/(\beta\cdot\eps))\right)$\\
{\tt FastSample} &
$\tilde{O}\left(n^2 \cdot \log^2|\cX| + ({n}/{\eps}) \cdot \log (|\cX|/\beta) + n\cdot \log |\cX|\cdot\log (1/(\beta \cdot\eps)) \right)$
\end{tabular}
\vspace{10pt}

\begin{tabular}{cccc}
\vspace{-.8em} & \multicolumn{2}{c}{$(a,\beta)$-Per-Query on $c_x(D) > t$} & \\
$\cM$ & $a$ & $t$ &$(a,\beta)$-Simultaneous\\
\hline
{\tt GeoSample} &
$\left\lceil\frac{9}{2\varepsilon}\ln\frac{4}{\beta}\right\rceil$ &
$2\cdot\left\lceil\frac{9}{2\varepsilon}\ln\frac{4|\cX|}{\beta}\right\rceil$ &
$2\cdot\left\lceil\frac{9}{2\varepsilon}\ln\frac{2|\cX|}{\beta}\right\rceil$\\
{\tt FastSample} &
$\left\lceil\frac{9}{2\varepsilon}\ln\frac{8}{\beta}\right\rceil$ &
$2\cdot\left\lceil\frac{9}{2\varepsilon}\ln\frac{8|\cX|}{\beta}\right\rceil$ &
$2\cdot\left\lceil\frac{9}{2\varepsilon}\ln\frac{4|\cX|}{\beta}\right\rceil$
\end{tabular}
\caption{The running time and errors of ${\tt PureSparseHistogram}_{\cM,\varepsilon,\beta/4}$ for the counting query algorithms of Section \vpref{sec:count} where $\varepsilon, \beta \in \N^{-1}$.
For per-query accuracy, the first value is the error $a$ and the second value is the threshold $t$.
Values shown are for a $(\varepsilon,0)$-differentially private release.
For ${\tt FastSample}$, we take $\gamma = \beta/(8|\cX|)$.}
\end{figure}

\subsection{Construction of {\tt SparseHistogram}}
\label{ssec:construction}

We finish this section with the construction of ${\tt SparseHistogram}$.
Notice that if we implement ${\tt KH}'$ by using ${\tt OrdSample}(m)$ to implement step 4, then {\tt OrdSample} will have to sample from $\mathrm{Bin}(i,p)$ for $i \ge |\cX| - n$ and rational $p$.
This in turn has probabilities as small as $p^i$, whose bit length is too large for us.

In order to keep our bit lengths manageable, we will only sample from a distribution close in statistical distance to the desired Binomial distribution.
We will leverage the fact that a Binomial random variable can be represented as the sum of independent Bernoulli random variables.
By representing the probability mass function of a Bernoulli as a vector, we can then compute the probability mass function of a Binomial by repeatedly convolving this vector.
\begin{definition}
Let $a, b \in \R^t$. We define the \textbf{convolution} $a * b \in \R^t$ such that $(a * b)_k = \sum_{i=0}^{k}a_ib_{k-i}$ for all $k \in [t-1]$.
And we define \textbf{$i$-fold convolution of $a$} denoted $*^{(i)}(a)$ to be $i-1$ convolutions of $a$ with itself.
\end{definition}

\begin{lemma}
Let $Y_1, Y_2 \in [t]$ be independent discrete random variables with probability mass functions $f_{Y_1}$ and $f_{Y_2}$ represented as vectors in $\R^{t+1}$.
Then $\Pr[Y_1 + Y_2 = k] = (f_{Y_1} * f_{Y_2})_k$ for all $k \in [t]$.
\end{lemma}

\begin{corollary}\label{corr:convolve binomial}
Let $Z \sim \mathrm{Bin}(m, p)$ and $Y \sim \mathrm{Bern}(p)$ with probability mass functions $f_{Z}$ and $f_{Y}$ represented as vectors in $\R^{m+1}$.
In particular, $f_{Y} = (1-p, p, 0, \ldots, 0)$.
Then we have $f_Z = *^{(m)}(f_Y)$.
\end{corollary}

The following algorithm will approximate $i$-fold convolution of $a$ by using an algorithm similar to exponentiation by repeated squaring while truncating each intermediate result to keep its bit length manageable.
The following lemma provides a bound on the error and on the running time.

\begin{proposition}\label{prop:ApproxConvExp}
There is an algorithm {\tt ApproxConvExp}$_{s,t}(a, i)$ such that for all $i, s, t \in \N_+$ and $a \in [s]^t$ such that $\|a\|_1 \le s$:

\begin{enumerate}[label=\roman*.]
\item
	${\tt ApproxConvExp}_{s,t}(a, i) \in [s]^t$ and $\|{\tt ApproxConvExp}_{s,t}(a, i)\|_1 \le s$.
\item
	The approximation does not change as $t$ varies. In particular, let $d = {\tt ApproxConvExp}_{s,t}(a, i)$. Then for $t' < t$
	\begin{align*}
	(d_0, \ldots, d_{t'-1}) = {\tt ApproxConvExp}_{s,t'}((a_0, \ldots, a_{t'-1}), i)
	\end{align*}
\item
	${\tt ApproxConvExp}_{s,t}(a, i)$ satisfies the accuracy bound
	\begin{align*}
	\left\|\frac{{\tt ApproxConvExp}_{s,t}(a, i)}{s} - \frac{*^{(i)}(a)}{s^i}\right\|_1
	&\le \frac{t}{s}\cdot (i-1)
	\end{align*}
\item     
	${\tt ApproxConvExp}_{s,t}(a, i)$ has running time $O(\log i) \cdot \tilde{O}(t\cdot\log s) + O(\log^2 i)$.
\end{enumerate}
\end{proposition}

\begin{proof}
The algorithm is defined as follows.
\begin{figure}[H]
\vspace{-.1in}
\rule{\textwidth}{.5pt}
\begin{algorithm}\label{alg:ApproxConvExp}{\tt ApproxConvExp}$_{s,t}(a, i)$ for $i,s,t \in \N_+$ and $a \in [s]^t$
\begin{enumerate}[leftmargin=.3in]
\item
	If $i = 1$, stop and return $a$.
\item
	Let $b = {\tt ApproxConvExp}_{s,t}\left(a, \lfloor i / 2\rfloor\right)$.
\item
	Compute $c = \begin{cases}
(b * b)/s& \text{if $i$ is even}\\
(a * b * b)/s^2& \text{if $i$ is odd}
\end{cases}
$
\item
	Return $d = (\lfloor c_0\rfloor, \lfloor c_1\rfloor, \ldots, \lfloor c_{t-1}\rfloor)$.
\end{enumerate}
\end{algorithm}\vspace{-.1in}
\rule{\textwidth}{.5pt}
\end{figure}

\begin{proof}[Proof of {\it i}]
We prove this part by induction on $i$.
The result is trivial for $i = 1$.
Let $i > 1$.
By the inductive hypothesis for $b = {\tt ApproxConvExp}_{s, t}(a, \lfloor i / 2\rfloor)$, we have $\|b\|_1 \le s$.
If $i$ is odd, then $c = (b * b) / s$ and we have $\|c\|_1 \le \|b\|_1^2/s \le s$.
If $i$ is even, then $c = (a * b * b) / s^2$ and we have $\|c\|_1 \le \|a\|_1 \cdot \|b\|_1^2/s^2 \le s$.
Finally, for the output $d = {\tt ApproxConvExp}_{s,t}(a, i) = (\lfloor c_0\rfloor, \ldots, \lfloor c_{t-1}\rfloor)$, we have $\|d\|_1 \le \|c\|_1 \le s$.
\end{proof}

\begin{proof}[Proof of {\it ii}]
Let $a' = (a_0, \ldots, a_{t'-1})$.
We prove this part by induction on $i$.
The result is trivial for $i = 1$.
Let $i > 1$.
By the inductive hypothesis, for $b = {\tt ApproxConvExp}_{s,t}(a, \lfloor i / 2\rfloor)$ and $b' = {\tt ApproxConvExp}_{s,t'}(a', \lfloor i / 2\rfloor)$ we have $b' = (b_0, \ldots, b_{t'-1})$.
Then for $k \in [t'-1]$ and $i$ even
\begin{align*}
{\tt ApproxConvExp}_{s,t}(a, i)_k
&= \left\lfloor (b * b)_k / s\right\rfloor\\
&= \left\lfloor (b' * b')_k / s\right\rfloor
= {\tt ApproxConvExp}_{s,t'}(a', i)_k
\end{align*}
where the second equality holds because $(b * b)_k$ only depends on the first $k+1$ terms of $b$ (i.e. $b_0, \ldots, b_k$) which are equal those of $b'$ since $k < t'$.
Similarly this induction holds for $i$ odd.
\end{proof}

\begin{proof}[Proof of {\it iii}]
We prove this part by induction on $i$.
The result is trivial for $i = 1$.
Let $i > 1$.
For $i$ odd, we can bound the error as
\begin{align*}
\left\|\frac{d}{s} - \frac{*^{(i)}(a)}{s^i}\right\|_1
&\le \left\|\frac{d - c}{s}\right\|_1 + \left\|\frac{c}{s} - \frac{*^{(i)}(a)}{s^i}\right\|_1\\
&\le \frac{t}{s} + \left\|\left(\frac{a}{s} * \frac{b}{s} * \frac{b}{s}\right) - \left(\frac{a}{s} * \frac{*^{(\lfloor i / 2 \rfloor)}(a)}{s^{\lfloor i / 2 \rfloor}} * \frac{*^{(\lfloor i / 2 \rfloor)}(a)}{s^{\lfloor i / 2 \rfloor}}\right)\right\|_1\\
&\le \frac{t}{s} + \left\|\frac{a}{s}\right\|_1 \cdot \left\|\frac{b}{s} + \frac{*^{(\lfloor i / 2 \rfloor)}(a)}{s^{\lfloor i/2 \rfloor}}\right\|_1 \cdot \left\|\frac{b}{s} - \frac{*^{(\lfloor i / 2 \rfloor)}(a)}{s^{\lfloor i/2 \rfloor}}\right\|_1\\
&\le \frac{t}{s} + 2 \cdot \left\|\frac{b}{s} - \frac{*^{(\lfloor i / 2 \rfloor)}(a)}{s^{\lfloor i/2 \rfloor}}\right\|_1\\
&\le \frac{t}{s}\cdot (i-1)
\end{align*}
with the last step by the inductive hypothesis.
Similarly this induction holds for $i$ even.
\end{proof}

\begin{proof}[Proof of {\it iv}]
The running time follows from the observation that a call to ${\tt ApproxConvExp}_{s,t}\left(a, i\right)$ makes at most $O(\log i)$ recursive calls with each dominated by computing $\lfloor i /2\rfloor$ and the convolution of up to three vectors in $[s]^t$ which can be done in time $\tilde{O}(t \cdot\log s)$ (\cite{vzGaGe13} Corollary 8.27).
\end{proof}
\vspace{-1.5em}
\renewcommand{\qedsymbol}{}
\end{proof}

Using ${\tt ApproxConvExp}$ we can approximately sample from a Binomial distribution truncated to a specified upper bound.

\begin{proposition}\label{prop:approxbinsample}
There exists an algorithm ${\tt ApproxBinSample}_{s, t}(m, p, q)$ such that for $m, q, s \in \N_+$ with $s \ge m$, $t \in \N$, $p \in [q]$ and $Z \sim \min\{\mathrm{Bin}(m, p/q),\,t\}$, we have
\begin{align*}
\Delta({\tt ApproxBinSample}_{s,t}(m, p, q), Z)
&\le \frac{m\cdot(t+1) - t}{s}
\end{align*}
In addition for every $\ell \in [t]$, every execution of ${\tt ApproxBinSample}_{s,t}(m, p, q)$ that produces an output of $\ell$ has running time at most
\begin{align*}
\tilde{O}(\log q) + O(\log m) \cdot \tilde{O}(\ell\cdot\log s)
\end{align*}
\end{proposition}

Let $a = (1-p/q, p/q, 0, \ldots, 0)$. By using ${\tt ApproxConvExp}_{s,t}$, we can approximate $*^{(m)}(a)$ and therefore the CDF of $Z$ for $k < t$.
Then to approximately sample from $Z$ we follow Lemma \vpref{lem:sampling discrete} by first generating a random uniform $u$ and then outputting the smallest $\ell$ such that the CDF at $\ell$ is at least $u$.
This would yield what is claimed in Proposition \vpref{prop:approxbinsample}, except that the running time would depend nearly linearly on $t$ instead of the specific output $\ell$.
To remedy this, we approximate the first $t'$ terms of $*^{(m)}(a)$ for $t'=1,2,4,\ldots$ until we find an $\ell < t'$ such that CDF of $Z$ at $\ell$ is at least $u$.

\begin{proof}
The algorithm is defined as follows.

\begin{figure}[H]
\vspace{-.1in}
\rule{\textwidth}{.5pt}
\begin{algorithm}\label{alg:approxbinsample}
{\tt ApproxBinSample}$_{s, t}(m, p, q)$ for $m, s, q \in \N_+$ such that $s \ge m$, $t \in \N$ and $p \in [q]$
\begin{enumerate}[leftmargin=.3in]
\item
	If $t = 0$, stop and return $0$.
\item
	Sample $u$ uniformly at random from $[s]_+$.
\item
	Let $t' = 1$ and $p' = \lfloor sp/q\rfloor$.
\item
	While $t' < 2t$, do the following:
	\begin{enumerate}
	\item
		Let $a \in [s]^{t'}$ such that $a_0 = s - p'$, $a_1 = p'$ and $a_k = 0$ for $k \ge 2$.
	\item
		Let $d^{(t')} = {\tt ApproxConvExp}_{s, t'}(a, m)$.
	\item
		Let $F^{(t')}_0 = d^{(t')}_0$ and $F^{(t')}_k = F^{(t')}_{k-1} + d^{(t')}_k$ for $k \in [t'-1]_+$.	
	\item
		For $\ell \in [\lfloor t' / 2\rfloor,\,\ldots,\, t'-1]$, do the following:
		\begin{enumerate}
		\item
			If $F^{(t')}_\ell \ge u$, stop and return $\min\{\ell, t\}$.
		\end{enumerate}
	\item
		Set $t'$ to $2t'$.
	\end{enumerate}
\item
	Return $t$.
\end{enumerate}
\end{algorithm}\vspace{-.1in}
\rule{\textwidth}{.5pt}
\end{figure}

Assume $t > 0$ (as accuracy is trivial otherwise).
Notice that ${\tt ApproxBinSample}_{s,t}(m,p) \in [t]$ and by Proposition \vpref{prop:ApproxConvExp} Parts {\it i}-{\it ii}, we have $F^{(t')}_{k} = F^{(t)}_{k}$ and $F^{(t)}_{k} \in [s]$ for $k < t' \le t$.
Then by construction for $\ell \in [t-1]$,
\begin{align*}
\Pr[{\tt ApproxBinSample}_{s,t}(m,p, q) = \ell]
&= \Pr\left[u \in \left(F^{(t)}_{\ell-1},\, F^{(t)}_{\ell}\right]\right]
= \frac{1}{s} \cdot d^{(t)}_\ell
\end{align*}
Now, let $Z' \sim \min\{\mathrm{Bin}(m, p'/s),\,t\}$.
Let $Y \sim \mathrm{Bern}(p / q)$ and $Y' \sim \mathrm{Bern}(p' / s)$.
Notice that 
\begin{align*}
\Delta(Z,\, Z')
&\le m\cdot\Delta(Y,\, Y')
\le \frac{m}{s}
\end{align*}
and, by Corollary \ref{corr:convolve binomial}, $\Pr[Z' = k] = *^{(m)}(a)_k/s^m$ for $k \in [t-1]$ where $a = (s - p', p', 0, \ldots, 0) \in [s]^t$.
Therefore, by Proposition \vpref{prop:ApproxConvExp} Part {\it iii}
\begin{align*}
\Delta({\tt ApproxBinSample}_{s,t}(m,p, q), Z)
&\le \Delta(Z, Z') + \Delta({\tt ApproxBinSample}_{s,t}(m,p, q), Z')\\
&\le \frac{m}{s} + \left\|\frac{d^{(t)}}{s} - \frac{*^{(m)}(a)}{s^m}\right\|_1\\
&\le \frac{m\cdot(t+1) - t}{s}
\end{align*}

We consider each step to calculate the running time.
Step 3 takes times $\tilde{O}(\log s + \log q)$.
The $i$-th iteration of step 4 has $t' = 2^{i-1}$.
Therefore, for the $i$-th iteration step 4b takes $O(\log m) \cdot \tilde{O}(2^i\cdot\log s)$ time and steps (4c-d) take $O(2^i \cdot\log s)$ time.
Now, for the algorithm to output $\ell$, it must halt in the $1 + \lceil\log_2(\ell+1)\rceil$-th iteration of step 4. Therefore, the overall running time is
\begin{align*}
\tilde{O}(\log s + \log q) + \sum_{i=1}^{1 + \lceil\log_2(\ell+1)\rceil} O(\log m) \cdot \tilde{O}(2^i\cdot\log s)
&= \tilde{O}(\log q) + O(\log m) \cdot \tilde{O}(\ell\cdot\log s)\qedhere
\end{align*}
\end{proof}

Now we can modify {\tt OrdSample} by replacing sampling from a binomial distribution with a call to {\tt ApproxBinSample} to keep the bit lengths of its numbers from becoming too large, yielding an efficient algorithm whose output distribution is close to that of {\tt OrdSample}.

\begin{proposition}\label{prop:approxordsample}
Let $n, d \in \N_+$ and $F: [n] \rightarrow [d]_+$ such that $F$ is non-decreasing and $F(n) = d$.
Let $m \in \N_+$ such that $m \ge n + 1$ and $s \in \N_+$ such that $s \ge m$.
Then the following algorithm ${\tt ApproxOrdSample}_F(m, s)$ satisfies
\begin{align*}
\Delta\left({\tt OrdSample}_F(m), {\tt ApproxOrdSample}_F(m, s)\right) \le \frac{m \cdot (n^2+2n)}{s}
\end{align*}
In addition, ${\tt ApproxOrdSample}_F(m, s)$ has running time
\begin{align*}
O(n) \cdot \left(\mathrm{Time}(F) + \tilde{O}(\log d)\right) + O(\log m) \cdot \tilde{O}(n\cdot\log s)
\end{align*}
where $\mathrm{Time}(F)$ is the worst-case time to evaluate $F$.
\end{proposition}

The running time of {\tt ApproxOrdSample} is exponentially faster than {\tt OrdSample} as a function of $m$, which we take to be close to $|\cX|$.
However, {\tt ApproxOrdSample} will still be the bottleneck of {\tt PureSparseHistogram} as it has a nearly quadratic running time dependence on $n$ (since we will take $s \ge 2^n$) as opposed to the other steps which only have a nearly linear dependence on $n$ (see Theorem \vpref{thm:histogram puresparse} Part {\it iv}).

\begin{figure}[H]
\vspace{-.1in}
\rule{\textwidth}{.5pt}
\begin{algorithm}\label{alg:approxordsample}
${\tt ApproxOrdSample}_F(m, s)$
for $m,s \in \N_+$ such that $m \ge n + 1$ and $ s \ge m$ \hphantom{\hspace{21.75em}}where $F: [n] \rightarrow [d]_+$ such that $F$ is non-\hphantom{\hspace{21.75em}}decreasing and $F(n) = d$
\begin{enumerate}[leftmargin=.3in]
\item
	Let $L'_{n+1} = 0$.
\item
	For $v$ from $n$ to $1$, do the following:
	\begin{enumerate}[leftmargin=.3in]
	\item
		Let $\ell'_v = {\tt ApproxBinSample}_{s,\,n+1 - L'_{v+1}}(m - L'_{v+1}, F(v) - F(v-1), F(v))$.
	\item
		Let $L'_{v} = L'_{v+1} + \ell'_v$.
	\end{enumerate}
\item
	Let $\ell'_0 = n+1 - L'_1$.
\item
	Return $(c'_0, \ldots c'_n)$ such that the first $\ell'_n$ values are $n$, the next $\ell'_{n-1}$ values are $n-1$ and so on until the last $\ell'_0$ values are 0.\vspace{-.5ex}
\end{enumerate}
\end{algorithm}
\vspace{-.1in}
\rule{\textwidth}{.5pt}
\end{figure}

\begin{proof}
Let $(\ell_0, \ldots, \ell_n)$ and $(L_1, \ldots, L_{n+1})$ and be defined as in ${\tt OrdSample}_F(m)$. And let $(\ell'_0, \ldots, \ell'_n)$ and $(L'_1, \ldots, L'_{n+1})$ be defined as in ${\tt ApproxOrdSample}_F(m, s)$.
Let $\mathcal{K}_v = \{(k_{v+1}, \ldots, k_{n}) \in \N^{n-v} \,:\, \sum_{i=v+1}^{n} k_i \le n+1\}$.
By Lemma \vpref{lem:tvd} Part {\it iii} and induction,
\begin{align*}
&\Delta({\tt ApproxOrdSample}_F(m,s), {\tt OrdSample}_F(m))\\
&\hspace{40pt}\le \Delta((\ell_0, \ldots, \ell_n), (\ell'_0, \ldots, \ell'_n))\\
&\hspace{40pt}\le \sum_{v = 0}^{n}\max_{(k_{v+1}, \ldots, k_n) \in \mathcal{K}_v}\Delta(\ell_v \,|\, \{\forall i > v ~~ \ell_i = k_i\},\, \ell_v' \,|\, \{\forall i > v ~~ \ell'_i = k_i\})
\end{align*}
Let $(k_{v+1}, \dots, k_n) \in \mathcal{K}_v$ and $K_{v+1} = \sum_{i=v+1}^{n}k_i$.
Then for $v \in [n]_+$ we have
\begin{align*}
&\ell'_v ~|~ \{\forall i > v ~~ \ell'_i = k_i\} \sim {\tt ApproxBinSample}_{s,\,n+1 - K_{v+1}}(m - K_{v+1}, F(v) - F(v-1), F(v))
\end{align*}
and recall that
\begin{align*}
&\ell_v ~|~ \{\forall i > v ~~ \ell_i = k_i\} \sim \min\left(\mathrm{Bin}\left(m - K_{v+1},\, p_v\right),\, n+1 - K_{v+1}\right)&\text{where $p_v = \frac{F(v) - F(v-1)}{F(v)}$}
\end{align*}
So by Proposition \vpref{prop:approxbinsample}
\begin{align*}
\Delta(\ell_v \,|\, \{\forall i > v ~~ \ell_i = k_i\},\, \ell_v' \,|\, \{\forall i > v ~~ \ell'_i = k_i\}) \le \frac{m \cdot(n + 2)}{s}
\end{align*}
In addition, $\ell_0 ~|~ \{\forall i > 0 ~~ \ell_i = k_i\} \sim \ell'_0 ~|~ \{\forall i > 0 ~~ \ell'_i = k_i\}$.
Therefore,
\begin{align*}
\Delta({\tt ApproxOrdSample}_F(m,s), {\tt OrdSample}_F(m))
&\le \frac{m \cdot(n^2 + 2n)}{s}
\end{align*}

The running time of this algorithm is dominated by the $n$ calls to ${\tt ApproxBinSample}$ and the evaluation of $F$ on all $n+1$ points.
The other steps are simple arithmetic on $O(\log n)$ bit numbers and the final step can be done in $O(n \log n)$ time.

Because $\sum_{v=1}^{n} \ell'_v \le n+1$, the running time to all calls of ${\tt ApproxBinSample}$ takes time
\begin{align*}
O(n) \cdot \tilde{O}(\log d) + \sum_{v=1}^{n}O(\log m) \cdot \tilde{O}(\ell'_v\cdot\log s)
&\le O(n) \cdot \tilde{O}(\log d) + O(\log m) \cdot \tilde{O}(n\cdot\log s)\qedhere
\end{align*}
\end{proof}

We are ready to state the algorithm ${\tt SparseHistogram}$ and show it satisfies Theorem \vpref{thm:histogram sparse}.
It is identical to ${\tt KH}'$ except we replace sampling of order statistics with a call to ${\tt ApproxOrdSample}$.

\begin{figure}[H]
\rule{\textwidth}{.5pt}
\begin{algorithm}\label{alg:sparse}
${\tt SparseHistogram}_{\cM, \delta}(D)$
for $D \in \cX^n$ where deterministic $\cM: [n] \times [d]_+ \rightarrow [n]$,
\hphantom{\hspace{21.25em}}$\delta \in \N^{-1}$ and $|\cX| \ge 2n + 1$
\begin{enumerate}[leftmargin=.3in]
\item
	Let $A = \{x \in \cX : c_x(D) > 0\}$ and $m = |\cX\setminus A|$.
\item
	Let $\{(x, \tc_x)\}_{x \in A} = {\tt BasicHistogram}_{\cM,A}(D)$.
\item
	Pick a uniformly random sequence $(q_0,\ldots, q_n)$ of distinct elements from $\cX\setminus A$.
\item
	Let $(\tc_{q_0}, \ldots, \tc_{q_n}) = {\tt ApproxOrdSample}_{F_{\cM}}(m, s)$ where $s = {(n^2+2n)} \cdot {|\cX|}/{\delta}$.
\item
 Sort the elements of $A \cup \{q_0, \ldots, q_n\}$ as $x_1, \ldots, x_{|A|+n+1}$ such that $\tc_{x_1} \ge \ldots \ge \tc_{x_{|A|+n+1}}$.
\item
	Release $h = \{(x, \tc_x): x \in \{x_1, \ldots, x_n\} \text{ and } \tc_x > \tc_{x_{n+1}}\} \in \cH_{n,n}(\cX)$.
\end{enumerate}
\end{algorithm}
\vspace{-.1in}
\rule{\textwidth}{.5pt}
\end{figure}

\begingroup
\def\thetheorem{\ref*{thm:histogram sparse}}
\addtocounter{theorem}{-1}
\begin{theorem}[restated]
Let deterministic $\cM: [n] \times [d]_+ \rightarrow [n]$ be $(\varepsilon/2, 0)$-differentially private for counting queries.
And let $\delta \in \N^{-1}$ and $|\cX| \ge 2n + 1$.
Then ${\tt SparseHistogram}_{\cM,\beta_1,\delta}: \cX^n \rightarrow \cH_{n, n}(\cX)$ has the following properties:
\begin{enumerate}[label=\roman*.]
\item
	$\Delta\left({\tt KH}'_{\cM}(D), {\tt SparseHistogram}_{\cM, \delta}(D)\right) \le \delta$ for all $D \in \cX^n$.
\item
	${\tt SparseHistogram}_{\cM, \delta}$ is $\left(\varepsilon, (e^\varepsilon + 1)\cdot\delta\right)$-differentially private.
\item
	The running time of ${\tt SparseHistogram}_{\cM,\delta}$ is
	\begin{align*}
	\tilde{O}(n\cdot\log|\cX|& \cdot(\log |\cX| + \log (1/\delta)))
	+ O(n) \cdot \left(\tilde{O}(\log d) + \mathrm{Time}(\cM)+\mathrm{Time}(F_\cM)\right)
	\end{align*}
\end{enumerate}
\end{theorem}
\endgroup

\begin{proof}[Proof of {\it i}]
Let ${\tt KH}'^{*}: \cX^n \rightarrow \cH_{n,2n+1}(\cX)$ be the algorithm ${\tt KH}'$ except, (if it passes the first step) instead of releasing the heaviest bins, ${\tt KH}'^{*}$ releases the bins for all elements of $A \cup \{q_0, \ldots, q_n\}$ (i.e. ${\tt KH}'^{*}$ releases $(x, \tc_x)$ for all $x \in A$ and $(q_i, \tc_{q_i})$ for all $i \in [n]$).
Similarly, we define ${\tt SparseHistogram}^*$ with respect to {\tt SparseHistogram}.

Notice that ${\tt KH}'^*$ and ${\tt SparseHistogram}^*$ have the same distribution overs the bins with nonzero true count. Only on the bins with counts sampled using ${\tt OrdSample}$ and ${\tt ApproxOrdSample}$ respectively do their output distributions differ. As a result, we can apply Proposition \vpref{prop:approxordsample} to the output distributions of ${\tt KH}'^*$ and ${\tt SparseHistogram}^*$. So for all $D \in \cX^n$
\begin{align*}
\Delta\left({\tt KH}_{\cM}'^*(D), {\tt SparseHistogram}_{\cM, \delta}^*(D)\right) \le {(n^2+2n)} \cdot \frac{|\cX|}{s}
&= \delta
\end{align*}

Now we consider the effect of keeping the heaviest counts.
Define $T:\cH_{n,2n+1}(\cX) \rightarrow \cH_{n,n}(\cX)$ to be the function that sets counts not strictly larger than the $(n+1)$-heaviest count of its input to 0.
Notice that $T\circ {\tt KH}'^*\sim {\tt KH}'$ and $T\circ {\tt SparseHistogram}^*\sim{\tt SparseHistogram}$.
So for all $D \in \cX^n$, by Lemma \vpref{lem:tvd} Part {\it 2},
\begin{align*}
\Delta\left({\tt KH}_{\cM}'(D), {\tt SparseHistogram}_{\cM, \delta}(D)\right)
&= \Delta\left(T\left({\tt KH}_{\cM}'^*(D)\right), T\left({\tt SparseHistogram}_{\cM, \delta}^*(D)\right)\right)\\
&\le \Delta({\tt KH}_{\cM}'^*(D), {\tt SparseHistogram}_{\cM, \delta}^*(D))\\
&\le \delta\qedhere
\end{align*}
\end{proof}

\begin{proof}[Proof of {\it ii}]
Let $D$ and $D'$ be neighboring datasets. Let $S \subseteq \cH_{n,n}(\cX)$.
By the previous part, Theorem \vpref{thm:histogram heavy} Part {\it i} and Proposition \vpref{prop:histogram kh1},
\begin{align*}
\Pr[{\tt SparseHistogram}_{\cM, \delta}(D) \in S]
&\le \Pr[{\tt KH}_{\cM}'(D) \in S] + \delta\\
&\le e^\varepsilon\cdot\Pr[{\tt KH}_{\cM}'(D') \in S] + \delta\\
&\le e^\varepsilon\cdot\left(\Pr[{\tt SparseHistogram}_{\cM, \delta}(D) \in S] + \delta\right) + \delta
\end{align*}
Therefore, ${\tt SparseHistogram}_{\cM, \delta}$ is $\left(\varepsilon, (e^\varepsilon + 1)\cdot\delta\right)$-differentially private.
\end{proof}

\begin{proof}[Proof of {\it iii}]
We consider the running time at each step.
Steps 1, 5 and 6 take time $O(n\log n \cdot \log |\cX|)$.
Step 3 can be done in time $O(n\log n \cdot \log^2 |\cX|)$ (see Appendix \ref{sec:distinct sample}).
For step 2, by Theorem \vpref{thm:histogram basic} Part {\it i}, the call to ${\tt BasicHistogram}_{\cM, A}(D)$ takes time
\begin{align*}
O\left(n\log n \cdot \log |\cX| + n\cdot \log d + n\cdot\ctime(\cM)\right)
\end{align*}

Notice that $s$ can be computed in $\tilde{O}(\log n + \log |\cX| + \log (1/\delta))$ time and has bit length $O(\log n + \log |\cX| + \log (1/\delta))$.
Thus, by Proposition \vpref{prop:approxordsample} the call to ${\tt ApproxOrdSample}_{F_\cM}(m,s)$ in step 4 can be computed in time
\begin{align*}
\tilde{O}(n\cdot\log|\cX| \cdot(\log |\cX| + \log (1/\delta))) + O(n) \cdot \left(\tilde{O}(\log d) + \mathrm{Time}(F_\cM)\right)
\end{align*}
Therefore, overall ${\tt SparseHistogram}_{\cM, \delta}$ has the desired running time.
\end{proof}

\section{Lower Bounds}\label{sec:lower}
In this section, we prove a lower bound on the per-query accuracy of histogram algorithms whose outputs are restricted to $\cH_{\infty,n'}(\cX)$ (i.e. sparse histograms) using a packing argument \cite{HaTa10,BeBrKaNi14}.
First, for completeness we state and reprove existing lower bounds for per-query accuracy and simultaneous accuracy as well as generalize them to the case of $\delta > 0$.

\begin{theorem}[following \cite{HaTa10,BeBrKaNi14}]\label{thm:lower old}
Let $\cM: \cX^n \rightarrow \cH_{\infty,|\cX|}(\cX)$ be $(\varepsilon, \delta)$-differentially private and $\beta \in (0,1/2]$.
\begin{enumerate}[label=\roman*.]
\item
	If $\cM$ has $(a,\beta)$-per-query accuracy, then
	\begin{align*}
	a
	\ge \frac{1}{2}\cdot\min\left\{\frac{1}{\varepsilon}\ln\left(\frac{1}{4\beta}\right) - 1,~ \frac{1}{\varepsilon}\ln\left(\frac{\varepsilon}{4\delta}\right) - 1,~n\right\}
	\end{align*}
\item
	If $\cM$ has $(a,\beta)$-simultaneous accuracy, then
	\begin{align*}
	a
	\ge \frac{1}{2}\cdot\min\left\{\frac{1}{\varepsilon}\ln\left(\frac{|\cX|-1}{4\beta}\right)-1,~\frac{1}{\varepsilon}\ln\left(\frac{\varepsilon}{4\delta}\right)-1,~n\right\}
	\end{align*}
\end{enumerate}
\end{theorem}

\begin{proof}[Proof of {\it i}]
Assume $a < n/2$.
Let $x,x_0 \in \cX$ such that $x \neq x_0$.
Define the dataset $D' \in \cX^n$ such that all rows are $x_0$.
And define the dataset $D$ such that the first $m = \lfloor 2a \rfloor + 1$ rows are $x$ and the remaining $n-m$ rows are $x_0$.
Notice that $\Pr[|\cM(D)_x - c_x(D)| > a] \le \beta$ by the $(a,\beta)$-per-query accuracy of $\cM$.
By Lemma \vpref{lem:dp group} and the fact that $c_x(D) > 2a$ while $c_{x}(D') = 0$,
\begin{align*}
\Pr[|\cM(D)_x - c_x(D)| > a]
&\ge e^{-m \varepsilon}\cdot\Pr[|\cM(D')_x - c_x(D)| > a] - \delta/\varepsilon\\
&\ge e^{-m \varepsilon}\cdot\Pr[|\cM(D')_x - c_x(D')| \le a] - \delta/\varepsilon\\
&\ge e^{-m \varepsilon}\cdot(1-\beta) - \delta/\varepsilon
\end{align*}
Therefore,
\begin{align*}
e^{-(2a+1)\cdot\varepsilon}
&\le \frac{1}{1-\beta}\cdot\left(\beta + \frac{\delta}{\epsilon}\right) \le 4\cdot\max\left\{\beta,~ \frac{\delta}{\varepsilon}\right\}\qedhere
\end{align*}
\end{proof}

\begin{proof}[Proof of {\it ii}]
Assume $a < n/2$.
Let $x_0 \in \cX$.
For each $x \in \cX$ define the dataset $D^{(x)} \in \cX^n$ such that the first $m = \lfloor 2a \rfloor + 1$ rows are $x$ and the remaining $n-m$ rows are $x_0$.
For all $x \in \cX$, let
\begin{align*}
G_x = \{h \in \cH_{\infty, |\cX|}(\cX) : \forall x' \in \cX ~~ |h_{x'} - c_{x'}(D^{(x)})| \le a\}
\end{align*}
By Lemma \vpref{lem:dp group}, for all $x \in \cX$
\begin{align*}
\Pr[\cM(D^{(x_0)}) \in G_{x}]
&\ge e^{-m\varepsilon}\cdot \Pr[\cM(D^{(x)}) \in G_{x}] - \delta/\varepsilon\\
&\ge e^{-m\varepsilon}\cdot (1-\beta) - \delta/\varepsilon
\end{align*}
Notice that $\Pr[\cM(D^{(x_0)}) \notin G_{x_0}] \le \beta$ and $\{G_x\}_{x\in\cX}$ is a collection of disjoint sets.
Then
\begin{align*}
\Pr[\cM(D^{(x_0)}) \notin G_{x_0}]
&\ge \sum_{x \in \cX : x \neq x_0}\Pr[\cM(D^{(x_0)}) \in G_{x}]\\
&\ge \left(|\cX| - 1\right)\cdot\left(e^{-m\varepsilon}\cdot (1-\beta) - \delta/\varepsilon\right)
\end{align*}
Therefore,
\begin{align*}
e^{-(2a+1)\cdot \varepsilon}
\le \frac{1}{1-\beta}\cdot\left(\frac{\beta}{|\cX|-1} + \frac{\delta}{\varepsilon}\right)
\end{align*}
which implies the desired lower bound.
\end{proof}

We now state and prove our lower bound for privately releasing sparse histograms.

\begin{theorem}\label{thm:lower pq}
Let $\cM: \cX^n \rightarrow \cH_{\infty,n'}(\cX)$ be $(\varepsilon, \delta)$-differentially private with $(a, \beta)$-per-query accuracy with $\beta \in (0,1/2]$.
Then
\begin{align*}
a
&\ge \frac{1}{2}\cdot\min\left\{\frac{1}{2\varepsilon}\ln\left(\frac{|\cX|}{16\beta n'}\right) - 1,~\frac{1}{\varepsilon}\ln\left(\frac{\varepsilon}{4\delta}\right)-1, n\right\}
\end{align*}
\end{theorem}

The histogram algorithms of Sections \vpref{ssec:general stability} and \vpref{sec:speed} achieve $(O(\log(1/\beta)/\varepsilon),\beta)$-per-query accuracy on large enough counts.
However, on smaller counts we can only guarantee $(a,\beta)$-per-query accuracy with $a= O(\log(1/(\beta\delta))/\varepsilon)$ and $a = O(\log(|\cX|/\beta)/\varepsilon)$ for algorithms from Sections \vpref{ssec:general stability} and \vpref{sec:speed} respectively  (taking threshold $b=O(\log(1/\delta)/\varepsilon)$ in Section \vpref{ssec:general stability}) .
Theorem \vpref{thm:lower pq} shows these bounds are the best possible, up to constant factors, when $|\cX| \ge (n')^2$, $\varepsilon^2 \ge \delta$ and $\beta \ge \delta$.

\begin{proof}
Assume $a < n/2$.
Let $x_0 \in \cX$.
For each $x \in \cX$ define the dataset $D^{(x)} \in \cX^n$ such that the first $m = \lceil2a\rceil$ rows are $x$ and the remaining $n-m$ rows are $x_0$.
By definition of $(a,\beta)$-per-query accuracy and the fact that $c_x(D^{(x)}) \ge 2a$, we have
\begin{align*}
\Pr\left[\cM(D^{(x)})_x \ge a\right]
\ge \Pr\left[\left|\cM(D^{(x)})_x - c_x(D^{(x)})\right|\le a\right]
\ge 1-\beta
\end{align*}
Then, by Lemma \vpref{lem:dp group} and that $D^{(x)}$ is at distance at most $m$ from $D^{(x_0)}$, we have
\begin{align*}
\Pr\left[\cM(D^{(x_0)})_{x} \ge a\right]
\ge (1-\beta)e^{-m\varepsilon} - \delta/\varepsilon
\end{align*}
Thus, by linearity of expectations
\begin{align*}
\mathrm{E}\left[\left|\left\{x \in \cX : \cM(D^{(x_0)})_{x} \ge a\right\}\right|\right]
\ge |\cX|\cdot\left((1-\beta)e^{-m\varepsilon} - \delta/\varepsilon\right)
\end{align*}
On the other hand, as $\cM(D^{(x_0)}) \in \cH_{\infty,n'}(\cX)$ we have
\begin{align*}
\mathrm{E}\left[\left|\left\{x \in \cX : \cM(D^{(x_0)})_{x} \ge a\right\}\right|\right] \le n'
\end{align*}
Therefore,
\begin{align*}
e^{-\lceil2a\rceil\cdot\varepsilon} \le \frac{1}{1-\beta}\cdot\left(\frac{n'}{|\cX|} + \frac{\delta}{\varepsilon}\right)
\end{align*}
which along with $\lceil2a\rceil \le 2a + 1$ implies the lower bound of
\begin{align*}
a &\ge \frac{1}{2}\cdot\min\left\{\frac{1}{\varepsilon}\ln\left(\frac{|\cX|}{4n'}\right) - 1,~ \frac{1}{\varepsilon}\ln\left(\frac{\varepsilon}{4\delta}\right)-1, ~ n\right\}
\end{align*}
Therefore, along with Theorem \vpref{thm:lower old} Part {\it i}, we have
\begin{align*}
a
&\ge \frac{1}{2}\cdot\min\left\{\max\left\{\frac{1}{\varepsilon}\ln\left(\frac{|\cX|}{4n'}\right)-1,~\frac{1}{\varepsilon}\ln\left(\frac{1}{4\beta}\right) - 1\right\},~ \frac{1}{\varepsilon}\ln\left(\frac{\varepsilon}{4\delta}\right) - 1,~n\right\}\\
&\ge \frac{1}{2}\cdot\min\left\{\frac{1}{2\varepsilon}\ln\left(\frac{|\cX|}{16\beta n'}\right) - 1,~\frac{1}{\varepsilon}\ln\left(\frac{\varepsilon}{4\delta}\right)-1, n\right\}\qedhere
\end{align*}
\end{proof}

\section{Better Per-Query Accuracy via Compact, Non-Sparse Representations}\label{sec:compact}

In this section, we present a histogram algorithm whose running time is poly-logarithmic in $|\cX|$, but, unlike Algorithm \vpref{alg:histogram puresparse}, is able to achieve $(O(\log(1/\beta)/\varepsilon), \beta)$-per query accuracy.
It will output a histogram from a properly chosen family of succinctly representable histograms.
This family necessarily contains histograms that have many nonzero counts to avoid the lower bound of Theorem \vpref{thm:lower pq}.

\subsection{The Family of Histograms}

We start by defining this family of histograms.

\begin{lemma}\label{lem:hash}
Let deterministic $\cM_0: [d_0]_+ \rightarrow [n]$, $d_0 = 2^{2\cdot 3^\ell}$ for some $\ell \in \N$, $d_0 \ge |\cX|$ and $U \sim \unif([d_0]_+)$.
There exists a multiset of histograms $\cG_{\cM_0}(\cX)$ satisfying:

\begin{enumerate}[label=\roman*.]
\item
	Let $g \sim \unif(\cG_{\cM_0}(\cX))$.
	For all $x \in \cX$, the marginal distribution $g_x$ is distributed according to $\cM_0(U)$.
\item
	Let $g \sim \unif(\cG_{\cM_0}(\cX))$.
	For all $B \subseteq \cX$ such that $|B| \le n+1$ and for all $c \in [n]^{B}$
	$$\Pr[\forall x \in B ~~ g_x = c_x] = \prod_{x \in B}	\Pr[g_x = c_x]$$
\item
	For all $g \in \cG_{\cM_0}(\cX)$, the histogram $g$ can be represented by a string of length $O(n\cdot \log d_0)$ and given this representation for all $x \in \cX$ the count $g_x$ can be evaluated in time
\begin{align*}
O(n)\cdot\tilde{O}(\log d_0) + \ctime(\cM_0)
\end{align*}
\item
	For all $A \subseteq \cX$ such that $|A| \le n$ and $c \in [n]^{A}$ sampling a histogram $h$ uniformly at random from $\{g \in \cG_{\cM_0}(\cX) : \forall x \in A ~~ g_x = c_x\}$ can be done in time
\begin{align*}
O(n) \cdot \ctime(\cS) + \tilde{O}(n \cdot \log d_0)
\end{align*}
where $\ctime(\cS)$ is the maximum time over $v \in [n]$ to sample from the distribution $\cS_{v} \sim \unif(\{u_0 \in [d_0]_+ : \cM_0(u_0) = v\})$.
\end{enumerate}
\end{lemma}

\begin{proof}
(Construction) Let $\cG_{\cM_0}'$ be the set of all degree at most $n$ polynomials over the finite field $\F_{d_0}$.
Now, $\cG_{\cM_0}'$ is a $(n+1)$-wise independent hash family mapping $\F_{d_0}$ to $\F_{d_0}$. That is, $p \sim \unif(\cG_{\cM_0}')$ has the following properties:
\begin{itemize}
\item
	Let $x \in \F_{d_0}$.
	Then $p(x) \sim \unif(\F_{d_0})$.
\item
	Let $x_0, \ldots, x_n \in \F_{d_0}$ be distinct.
	Then the random variables $p(x_0), \ldots, p(x_n)$ are independent.
\end{itemize}
And given any function $p_g \in \cG_{\cM_0}'$ we construct a histogram $g \in \cG_{\cM_0}(\cX)$ by using $p_g(x)$ as the randomness for $\cM_0$.
More specifically, let $T: \F_{d_0} \rightarrow [d_0]_+$ be a bijection and for all $x \in \cX$, define
\begin{align*}
g_x = \cM_0(T(p_g(x)))
\end{align*}
By construction, $g \sim \unif(\cG_{\cM_0}(\cX))$ if and only if $p_g \sim \unif(\cG_{\cM_0}')$.

\begin{proof}[Proof of {\it i}]
Let $g \sim \unif(\cG_{\cM_0}(\cX))$ and $U \sim \unif([d_0]_+)$.
Then $p_g \sim \unif(\cG_{\cM_0}')$ and because $\cG_{\cM_0}'$ is a $(n+1)$-wise independent hash family, for all $x \in \cX \subseteq [d_0]_+$, $T(p_g(x)) \sim U$ which implies $g_x = \cM_0(T(p_g(x))) \sim \cM_0(U)$.
\end{proof}
\begin{proof}[Proof of {\it ii}]
Let $g \sim \unif(\cG_{\cM_0}(\cX))$. Because $\cG_{\cM_0}'$ is a $(n+1)$-wise independent hash family, for all $B \subseteq \cX$ such that $|B| \le n+1$ and for all $c \in [n]^{B}$
\begin{align*}
\Pr[\forall x \in B ~~ g_x = c_x]
&= \Pr\left[\forall x \in B ~~ \cM_0(T(p_g(x))) = c_x\right]\\
&= \prod_{x \in B}\Pr\left[\cM_0(T(p_g(x))) = c_x\right]\\
&= \prod_{x \in B}\Pr[g_x = c_x]\qedhere
\end{align*}
\end{proof}

\begin{proof}[Proof of {\it iii}]
By choice of $d_0$, we have $\F_{d_0} \simeq \F_2[x] / (x^{2\cdot3^\ell} + x^{3^\ell} + 1)$ \cite{Lint99}.
Thus, elements of $\F_{d_0}$ can be represented by a polynomial of degree at most $\log d_0 - 1$ over $\F_2$ which requires $\log d_0$ bits.
Arithmetic operations (addition, multiplication and inverse) of elements in $\F_{d_0}$ can be done in time $\tilde{O}(\log d_0)$ (\cite{vzGaGe13} Corollary 11.11).
Also, this encoding defines an efficient bijection $T$ between $\F_{d_0}$ and $[d_0]_+$ by also interpreting the string as the binary representation of an element in $[d_0]_+$ offset by 1.

For all $g \in \cG_{\cM_0}(\cX)$, $g$ can be represented by the coefficients of $p_g$.
This representation can be encoded in $O(n \cdot \log{d_0})$ bits.
And given this encoding, the time to compute $g_x = \cM_0(T(p_g(x)))$ follows from its construction.
\end{proof}

\begin{proof}[Proof of {\it iv}] Let $A \subseteq \cX$ such that $|A| \le n$ and $c \in [n]^{A}$. Given $A$ and $c$, we can sample $h \in \cG_{\cM_0}(\cX)$ given by the coefficients $a_0,\ldots,a_n \in \F_{d_0}$ uniformly at random from $\{g \in \cG_{\cM_0}(\cX) : \forall x \in A ~~ g_x = c_x\}$ with the following steps:
\begin{enumerate}
\item
	For each $x \in A$, sample $u_x$ from the distribution $\cS_x$.
\item
	Let $B \subseteq \cX$ such that $A \subseteq B$ and $|B| = n+1$. For all $x \in B \setminus A$, sample $u_x$ uniformly at random from $\F_{d_0}$.
\item
	Take the coefficients $a_0, \ldots, a_{n} \in \F_{d_0}$ to be the coefficients of the interpolating polynomial over $\F_{d_0}$ given the set of points $\left(x, u_x\right)$ for all $x \in B$.
\end{enumerate}
We first prove correctness. Notice this procedure can only return a histogram $h \in \cG_{\cM_0}(\cX)$ such that $h_x = c_x$ for all $x \in A$ as the interpolating polynomial always exists. Now, let $h$ be any such histogram. Then
\begin{align*}
\Pr[\text{Sampling $h$}]
&= \Pr[\text{$(a_0, \ldots, a_n)$ are the coefficients of $p_h$}]\\
&= \Pr[\forall x \in B ~~u_x = p_h(x)]\\
&= \prod_{x\in B}\Pr[u_x = p_h(x)]\\
&= \left(\prod_{x\in A}\frac{1}{|\supp(\cS_x)|}\right)\cdot\left(\frac{1}{d_0}\right)^{|B \setminus A|}
\end{align*}
Therefore, these steps output $h \in \cG_{\cM_0}(\cX)$ uniformly at random such that $h_x = c_x$ for all $x \in A$.

Because $|A| \le n$, the first step takes time $O(n) \cdot \ctime(S)$.
Because $|B| = n+1$, the second step takes time $O(n\log n \cdot \log d_0)$.
Polynomial interpolation of $n+1$ points over $\F_{d_0}$ takes time $\tilde{O}(n \cdot \log d_0)$ (\cite{vzGaGe13} Corollary 10.12).
Thus, we have the desired running time overall.
\end{proof}
\vspace{-1.5em}
\renewcommand{\qedsymbol}{}
\end{proof}

\subsection{The Algorithm}

For our algorithm to have the correct marginal distributions over all bins we first compute the noisy counts for the nonzero bins using an algorithm $\cM$ that is differentially private for counting queries and then randomly pick a histogram from our family that is consistent with these computed counts.
However, for technical reasons (e.g. requiring $d_0 \ge |\cX|$) we allow our family to be defined in terms of an algorithm $\cM_0$ that approximates $\cM$; we refer to $\cM_0$ as the {\em empty-bin sampler}.

\begin{figure}[H]
\rule{\textwidth}{.5pt}
\begin{algorithm}\label{alg:histogram compact}
${\tt CompactHistogram}_{\cM,\cM_0}(D)$
for $D \in \cX^n$ where deterministic $\cM: [n] \times [d]_+ \rightarrow [n]$ \hphantom{\hspace{21.5em}} and deterministic $\cM_0: [d_0]_+ \rightarrow [n]$ such that \hphantom{\hspace{21.65em}} $d_0 = 2^{2\cdot 3^\ell}$ for some $\ell \in \N$ and $d_0 \ge |\cX|$
\begin{enumerate}[leftmargin=.3in]
\item
	Let $A = \{x \in \cX : c_x(D) > 0\}$.
\item
	Let $\{(x, \tc_x)\}_{x \in A} = {\tt BasicHistogram}_{\cM, A}(D)$.
\item
	Release $h$ drawn uniformly at random from $\{g \in \cG_{\cM_0}(\cX) : \forall x \in A ~~ g_x = \tc_x\}$.
\end{enumerate}
\end{algorithm}
\vspace{-.1in}
\rule{\textwidth}{.5pt}
\end{figure}

\begin{theorem}\label{thm:histogram compact}
Let deterministic $\cM: [n] \times [d]_+ \rightarrow [n]$ be $(\varepsilon_1/2, 0)$-differentially private for counting queries and have $(a,\beta)$-accuracy.
Let $\cM_0: [d_0]_+ \rightarrow [n]$ be deterministic, $d_0 = 2^{2\cdot 3^\ell}$ for some $\ell \in \N$ and $d_0 \ge |\cX|$.
Assume $\Pr[\cM_0(U_0) \le a] \ge 1 - \beta$ and for all $c \in [n]$
\begin{align*}
e^{-\varepsilon_2}\cdot\Pr[\cM_0(U_0) = c]
\le \Pr[\cM(0, U) = c]
\le e^{\varepsilon_3}\cdot\Pr[\cM_0(U_0) = c]
\end{align*}
where $U \sim \unif([d]_+)$ and $U_0 \sim \unif([d_0]_+)$.
Then ${\tt CompactHistogram}_{\cM,\cM_0}: \cX^n \rightarrow \cG_{\cM_0}(\cX)$ has the following properties:
\begin{enumerate}[label=\roman*.]
\item
	${\tt CompactHistogram}_{\cM,\cM_0}$ is $\left(\varepsilon_1 + \varepsilon_2 + \varepsilon_3,0\right)$-differentially private.

\item
	${\tt CompactHistogram}_{\cM,\cM_0}$ has $(a, \beta)$-per-query accuracy.
\item
	${\tt CompactHistogram}_{\cM,\cM_0}$ has $(a, \beta \cdot |\cX|)$-simultaneous accuracy.
\item
	${\tt CompactHistogram}_{\cM,\cM_0}$ has running time
	\begin{align*}
	\tilde{O}(n \cdot \log d_{0}) + O(n) \cdot (\log d + \ctime(\cM) + \ctime(\cS)) + O(n\log{n} \cdot \log|\cX|)
	\end{align*}
	where $\ctime(\cS)$ is the maximum time over $v \in [n]$ to sample from the distribution $\cS_{v} \sim \unif(\{u_0 \in [d_0]_+ : \cM_0(u_0) = v\})$.
\item
	Given $h = {\tt CompactHistogram}_{\cM, \cM_0}(D)$, for all $x \in \cX$ the count $h_x$ can be evaluated in time
	\begin{align*}
	O(n) \cdot \tilde{O}(\log d_0) + \ctime(\cM_0)
	\end{align*}
\end{enumerate}
\end{theorem}

As discussed earlier, a natural choice for $\cM_0$ is to take $\cM_0(u) = \cM(0,u)$ for all $u \in [d]$. However, $d$ may not satisfy the required constraints. In the next section (see Lemma \vpref{lem:count conversion}), we will show how to construct $\cM_0$ for the counting query algorithms of Section \vpref{sec:count} at only a constant factor loss in privacy (i.e. $\varepsilon_2 = O(\varepsilon)$ and $\varepsilon_3 = O(\varepsilon)$).

\begin{proof}[Proof of {\it i}]
Let $D,D' \in \cX^n$ be neighboring datasets.
Let $A = \{ x \in \cX : c_x(D) > 0\}$.
Similarly, define $A' = \{ x \in \cX : c_x(D') > 0\}$.
Let $B = A \cup A'$.
Notice that $|B| \le n + 1$.
Let $h \sim {\tt CompactHistogram}_{\cM,\cM_0}(D)$ and $h' \sim {\tt CompactHistogram}_{\cM, \cM_0}(D')$.
Let $g \sim \unif(\cG_{\cM_0}(\cX))$ and $r \in \cG_{\cM_0}(\cX)$.
Then by construction
\begin{align*}
\Pr[h = r]
&= \Pr[\forall x \in B ~~ h_x = r_x] \cdot \frac{\Pr[h = r ~|~ \forall x \in A ~~ h_x = r_x]}{\Pr[\forall x \in B ~~ h_x = r_x ~|~ \forall x \in A ~~ h_x = r_x]}\\
&= \frac{\Pr[\forall x \in B ~~ h_x = r_x]}{|\{g' \in \cG_{\cM_0}(\cX) : \forall x \in A ~~ g'_x = r_x\}|} \cdot\frac{|\{g' \in \cG_{\cM_0}(\cX) : \forall x \in A ~~ g'_x = r_x\}|}{|\{g' \in \cG_{\cM_0}(\cX) : \forall x \in B ~~ g'_x = r_x\}|}\\
&=  \Pr[\forall x \in B ~~ h_x = r_x] \cdot \Pr[g = r ~|~ \forall x \in B ~~ g_x = r_x]
\end{align*}
Now, because $|B\setminus A| \le 1$ and $|B\setminus A'| \le 1$ along with Lemma \vpref{lem:hash} Parts {\it i}-{\it ii}
\begin{align*}
&\frac{\Pr[h = r]}{\Pr[g = r ~|~ \forall x \in B ~~ g_x = r_x]}\\
&\hspace{4.5em}=\left(\prod_{x\in A}\Pr[\cM(c_x(D),U) = r_x]\right) \left(\prod_{x \in B \setminus A} \Pr[h_x = r_x ~|~ \forall x \in A ~~ h_x = r_x]\right)\\
&\hspace{4.5em}=\left(\prod_{x\in A}\Pr[\cM(c_x(D),U) = r_x]\right) \left(\prod_{x \in B \setminus A} \Pr[g_x = r_x ~|~ \forall x \in A ~~ g_x = r_x]\right)\\
&\hspace{4.5em}= \left(\prod_{x\in A}\Pr[\cM(c_x(D),U) = r_x]\right) \left(\prod_{x \in B \setminus A} \Pr[\cM_0(U_0) = r_x]\right)\\
&\hspace{4.5em}\le e^{\varepsilon_2}\cdot\prod_{x\in B}\Pr[\cM(c_x(D),U) = r_x]\\
&\hspace{4.5em}\le e^{\varepsilon_1 + \varepsilon_2}\cdot\prod_{x\in B}\Pr[\cM(c_x(D'),U) = r_x]\\
&\hspace{4.5em}\le \left(e^{\varepsilon_1 +\varepsilon_2}\cdot\prod_{x\in A'}\Pr[\cM(c_x(D'),U) = r_x]\right)\left( e^{\varepsilon_3}\cdot \prod_{x \in B \setminus A'} \Pr[\cM_0(U_0) = r_x]\right)\\
&\hspace{4.5em}= e^{\varepsilon_1+\varepsilon_2+\varepsilon_3}\cdot\frac{\Pr[h' = r]}{\Pr[g = r ~|~ \forall x \in B ~~ g_x = r_x]}
\end{align*}
Therefore, ${\tt CompactHistogram}_{\cM, \cM_0}$ is $(\varepsilon_1+\varepsilon_2+\varepsilon_3, 0)$-differentially private.
\end{proof}

\begin{proof}[Proof of {\it ii}]
Let $h \sim {\tt CompactHistogram}_{\cM, \cM_0}(D)$.
Let $A = \{x \in \cX : c_{x}(D) > 0\}$.
If $x \in A$, then $h_x \sim \cM(c_x(D),U)$ with accuracy following from $\cM$ having $(a,\beta)$-accuracy.

Otherwise, for $x \in \cX\setminus A$, let $g \sim \unif(\cG_{\cM_0}(\cX))$.
Notice that $c_x(D) = 0$ and $|A| \le n$.
By construction and Lemma \vpref{lem:hash} Parts {\it i}-{\it ii}
\begin{align*}
\Pr[|h_x| \le a]
&= \sum_{c \in [n]^A}\Pr[\forall x' \in A ~~ h_{x'} = c_{x'}] \cdot \Pr[h_x \le a ~|~ \forall x' \in A ~~ h_{x'} = c_{x'}]\\
&= \sum_{c \in [n]^A}\Pr[\forall x' \in A ~~ h_{x'} = c_{x'}] \cdot \Pr[g_x \le a ~|~ \forall x' \in A ~~ g_{x'} = c_{x'}]\\
&= \sum_{c \in [n]^A}\Pr[\forall x' \in A ~~ h_{x'} = c_{x'}] \cdot \Pr[g_x \le a]\\
&= \Pr[g_x \le a]\\
&= \Pr[\cM_0(U_0) \le a]\\
&\ge 1 -\beta
\end{align*}
with the last inequality by assumption.
Therefore, ${\tt CompactHistogram}_{\cM, \cM_0}$ has $(a,\beta)$-per-query accuracy.
\end{proof}

\begin{proof}[Proof of {\it iii}]
${\tt CompactHistogram}_{\cM, \cM_0}$ has $(a, \beta \cdot |\cX|)$-simultaneous accuracy by a union\\bound over each $x \in \cX$ along with the previous part.
\end{proof}

\begin{proof}[Proof of {\it iv}-{\it v}]
By Theorem \vpref{thm:histogram basic} Part {\it iv} and Lemma \vpref{lem:hash} Parts {\it iii}-{\it iv}, we get the desired bounds on running time.
\end{proof}

\subsection{Constructing the Empty-Bin Sampler}\label{ssec:conversion}

The following lemma allows us to construct the empty-bin sampler satisfying the constraints of Theorem \vpref{thm:histogram compact} for the counting query algorithms of Section \vpref{sec:count}.

\begin{lemma}\label{lem:count conversion}
Let deterministic $\cM: [n] \times [d]_+ \rightarrow [n]$ such that $\cM(0,u)$ is non-decreasing function in $u$.
Let $d_0 \in \N_+$ such that $d_0 \ge (4/3) \cdot d$.
Let $U \sim \unif([d]_+)$ and $U_0 \sim \unif([d_0]_+)$.
Then there is a deterministic algorithm $\cM_0: [d_0]_+ \rightarrow [n]$ with the following properties:
\begin{enumerate}[label=\roman*.]
\item
	$\cM_0$ is non-decreasing.
\item
	For all $c \in [n]$
	\begin{align*}
	e^{-2\cdot d/d_0}\cdot\Pr[\cM(0, U) = c]
	\le \Pr[\cM_0(U_0) = c]
	\le e^{d/d_0}\cdot \Pr[\cM(0, U) = c]
	\end{align*}

	Moreover, if $d_0$ is a multiple of $d$, then $\cM_0(U_0) \sim \cM(0,U)$.
\item
	For all $a \in [n]$
	\begin{align*}
	\Pr[\cM_0(U_0) \le a]
	\ge \Pr[\cM(0, U) \le a]
	\end{align*}

	In particular, if $\cM$ has $(a,\beta)$-accuracy, then $\Pr[\cM_0(U_0) \le a] \ge 1 - \beta$.
\item
	$\cM_0$ has running time
	\begin{align*}
	\tilde{O}(\log d_0) + \ctime(\cM)
	\end{align*}
\item
	For any $v \in [n]$, define the distribution $\cS_v \sim \unif(\{u_0 \in [d_0]_+ : \cM_0(u_0) = v\})$.
	Then $\cS_v$ can be sampled in time
	\begin{align*}
	\tilde{O}(\log d_0) + O(\ctime(F_\cM))
	\end{align*}
\end{enumerate}
\end{lemma}

\begin{proof}[Proof]
(Construction) We start by considering a sample uniformly at random from $[d_0]_+$ and then map it to $[d]_+$ such that the resulting number is almost uniform while preserving monotonicity.
The algorithm is as follows:

\begin{figure}[H]
\rule{\textwidth}{.5pt}
\begin{algorithm}\label{alg:m0}
$\cM_0(u_0)$
for $u_0 \in [d_0]_+$
\begin{enumerate}[leftmargin=.3in]
\item
	Pick $q \in \N$ and $r \in [d-1]$ such that $d_0 = q \cdot d + r$.
\item
	Define the function $f: [d_0] _+ \rightarrow [d]_+$
	\begin{align*}
	f(u_0)
	&= \begin{cases}
	\left\lceil u_0 / (q+1) \right\rceil
	& \text{if $r\neq0$ and $u_0 \le r\cdot(q+1)$}\\
	\left\lceil (u_0-r) / q \right\rceil
	& \text{if $r=0$ or $u_0 > r\cdot(q+1)$}
	\end{cases}
	\end{align*}
\item
	Return $\cM(0,f(u_0))$.
\end{enumerate}
\end{algorithm}
\vspace{-.1in}
\rule{\textwidth}{.5pt}
\end{figure}

\begin{proof}[Proof of {\it i}]
Notice that $f$ as defined in Algorithm \vpref{alg:m0} is non-decreasing.
Therefore, $\cM_0(u_0) = \cM(0,f(u_0))$ is non-decreasing in $u_0$ as $\cM(0,u)$ is non-decreasing in $u$.
\end{proof}
\begin{proof}[Proof of {\it ii}]
Notice that $\left\lceil (u_0-r) / q \right\rceil = \left\lceil (u_0-r\cdot(q+1)) / q \right\rceil + r$.
So $|\{u_0 \in [d_0]_+ : f(u_0) = u\}| = q \text{ or } q+1$ for all $u \in [d]_+$.
Thus,
\begin{align*}
\Pr[\cM_0(U_0) = c]
&\le \frac{q+1}{d_0}\cdot |\{u \in [d]_+ : \cM(0,u) = c\}|\\
&= \frac{(q+1)\cdot d}{d_0}\cdot\Pr[\cM(0, U) = c]\\
&\le \left(1 + \frac{d}{d_0}\right)\cdot\Pr[\cM(0, U) = c]\\
&\le e^{d/d_0}\cdot \Pr[\cM(0, U) = c]
\end{align*}
and because $d/d_0 \le 3/4$
\begin{align*}
\Pr[\cM_0(U_0) = c]
&\ge \frac{q\cdot d}{d_0}\cdot\Pr[\cM(0,U) = c]\\
&\ge \left(1 - \frac{d}{d_0}\right)\cdot\Pr[\cM(0, U) = c]\\
&\ge e^{-2\cdot d/d_0}\cdot\Pr[\cM(0, U) = c]
\end{align*}

If $d_0$ is a multiple of $d$, then $r=0$ and $|\{u_0 \in [d_0]_+ : f(u_0) = u\}| = q$ for all $u \in [d]_+$.
So $\Pr[\cM_0(U_0) = c] = \Pr[\cM(0,U) = c]$ for all $c \in [n]$.
\end{proof}

\begin{proof}[Proof of {\it iii}]
Let $u \in [d]_+$.
Notice that
\begin{align*}
|\{u_0 \in [d_0]_+ : f(u_0) = u\}| = \begin{cases}
q+1 & \text{if $u \le r$}\\
q & \text{otherwise}
\end{cases}
\end{align*}
Then for all $a \in [n]$
\begin{align*}
\Pr[\cM_0(U_0) \le a]
&= \frac{q\cdot d}{d_0}\cdot\Pr[\cM(0, U) \le a] + \frac{r}{d_0}\cdot\Pr[\cM(0,U) \le a ~|~ U \le r]\\
&\ge \frac{q\cdot d}{d_0}\cdot\Pr[\cM(0, U) \le a] + \frac{r}{d_0}\cdot\Pr[\cM(0,U) \le a]
\end{align*}
as $\cM(0,u)$ is non-decreasing in $u$.
Therefore, $\Pr[\cM_0(U_0) \le a] \ge \Pr[\cM(0, U) \le a]$ for all $a \in [n]$.
\end{proof}

\begin{proof}[Proof of {\it iv}]
Evaluating $f$ takes time $\tilde{O}(\log d_0)$.
So we have the desired running time overall.
\end{proof}

\begin{proof}[Proof of {\it v}]
Because $\cM(0,u)$ is non-decreasing in $u$ and $f(u_0)$ is non-decreasing in $u_0$, we have
\begin{align*}
\supp(\cS_v)
&= \{u_0 \in [d_0]_+ : \cM(0,f(u_0)) = v\}\\
&= \{u_0 \in [d_0]_+ : F_{\cM}(v-1) < f(u_0) \le F_{\cM}(v)\}\\
&= \left\{\min\{u_0 \in [d_0]_+ : f(u_0) \ge F_{\cM}(v-1) + 1\}, \ldots,  \max\{u_0 \in [d_0]_+ : f(u_0) \le F_{\cM}(v)\}\right\}
\end{align*}
One can show for $u \in [d]_+$ that
\begin{align*}
\max\{u_0 \in [d_0]_+ : f(u_0) \le u\}
&= \begin{cases}
(q+1) \cdot u &\text{if $r \neq 0$ and $u \le r$}\\
q \cdot u + r &\text{if $r = 0$ or $u > r$}
\end{cases}
\end{align*}
and $\min\{u_0 \in [d_0]_+ : f(u_0) \ge u\} = 1 + \max\{u_0 \in [d_0]_+ : f(u_0) \le u-1\}$.
So both endpoints can be computed in time $\tilde{O}(\log d_0)$.
\end{proof}
\vspace{-1.5em}
\renewcommand{\qedsymbol}{}
\end{proof}

Now, we are ready to obtain a private histogram algorithm that achieves pure differential privacy, has the same accuracy guarantees (up to constant factors) as the Laplace mechanism and has running polynomial in $n$ and $\log |\cX|$.

\begin{theorem}\label{thm:explicit compact}
Let $\varepsilon,\beta_0 \in \N^{-1}$ and $\cM = {\tt FastSample}_{n,\varepsilon',\gamma}$ where $\varepsilon' = 1/\lceil 10/(9\varepsilon)\rceil$ and $\gamma = \beta_0/(2|\cX|)$.
Then there exists deterministic $\cM_0: [d_0]_+ \rightarrow [n]$ with $\log d_0 = \tilde{O}(1/\varepsilon) \cdot (\log{n} + \log|\cX| + \log (1/\beta_0))$ such that ${\tt CompactHistogram}_{\cM, \cM_0}: \cX^n \rightarrow \cG_{\cM_0}(\cX)$ has the following properties:

\begin{enumerate}[label=\roman*.]
\item
	${\tt CompactHistogram}_{\cM, \cM_0}$ is $(\varepsilon,0)$-differentially private.
\item
	For every $\beta \ge 2\gamma$, ${\tt CompactHistogram}_{\cM, \cM_0}$ has $(a, \beta)$-per-query accuracy for
	\begin{align*}
	a
	= \left\lceil\frac{5}{\varepsilon}\ln\left(\frac{2}{\beta}\right)\right\rceil
	\end{align*}
\item
	For every $\beta \ge \beta_0$, ${\tt CompactHistogram}_{\cM, \cM_0}$ has $(a, \beta)$-per-query accuracy for
	\begin{align*}
	a
	= \left\lceil\frac{5}{\varepsilon}\ln\left(\frac{2\cdot |\cX|}{\beta}\right)\right\rceil
	\end{align*}
\item
	${\tt CompactHistogram}_{\cM, \cM_0}$ has running time
	\begin{align*}
	\tilde{O}\left(\frac{n}{\eps}\cdot\log\frac{|\cX|}{\beta_0}\right)
	\end{align*}
\item
	Given $h = {\tt CompactHistogram}_{\cM, \cM_0}(D)$, for all $x \in \cX$ the count $h_x$ can be evaluated in time
	\begin{align*}
	\tilde{O}\left(\frac{n}{\eps}\cdot\log\frac{|\cX|}{\beta_0}\right)
	\end{align*}
\end{enumerate}
\end{theorem}

\begin{proof}
Let $d_0 = 2^{2 \cdot 3^{\ell}}$ where $\ell = \lceil{\log_{3}(\lceil{\log_{2}{\max\{|\cX|,\, 30 \cdot d / \varepsilon\}}}\rceil / 2)}\rceil$ where $\log d = \tilde{O}(1/\varepsilon)\cdot (\log n + \log |\cX| + \log(1/\beta_0))$ as defined in Algorithm \vpref{alg:count fast}.
Notice that $d_0 \ge |\cX|$ and $d/d_0 \le \varepsilon/30$.
In addition, $d_0 = O(|\cX|^3 + (d/\varepsilon)^3)$.
So $\log d_0 = \tilde{O}(1/\varepsilon) \cdot (\log{n} + \log|\cX| + \log(1/\beta_0))$.
Now, let $\cM_0: [d_0]_+ \rightarrow [n]$ be defined as in Lemma \vpref{lem:count conversion} for $\cM$.
The proof follows from Theorem \vpref{thm:count fast} and Theorem \vpref{thm:histogram compact}.
\end{proof}

In the following figure, we show using {\tt CompactHistogram} with {\tt FastSample} is asymptotically faster than when using it with {\tt GeoSample}, particularly in the case when $\eps \gg 1/n$, with only a small constant loss in accuracy. 

\begin{figure}[H]
\centering
\renewcommand{\arraystretch}{1.8}
\begin{tabular}{ccc}
$\cM$ & Running Time & Evaluation Time\\
\hline
{\tt GeoSample} &
$\tilde{O}(n^2 \cdot \log(1/\varepsilon) + n \cdot \log|\cX|)$ &
$\tilde{O}(n^2 \cdot \log(1/\varepsilon) + n \cdot \log|\cX|)$ \\
{\tt FastSample} &
$\tilde{O}((n/\varepsilon) \cdot \log (|\cX|/\beta))$ &
$\tilde{O}((n/\varepsilon) \cdot \log (|\cX|/\beta))$
\vspace{3ex}
\end{tabular}

\begin{tabular}{ccc}
$\cM$ & $(a,\beta)$-Per-Query & $(a,\beta)$-Simultaneous \\
\hline
{\tt GeoSample} &
$\left\lceil\frac{5}{\varepsilon}\ln\frac{1}{\beta}\right\rceil$ &
$\left\lceil\frac{5}{\varepsilon}\ln\frac{|\cX|}{\beta}\right\rceil$\\
{\tt FastSample} &
$\left\lceil\frac{5}{\varepsilon}\ln\frac{2}{\beta}\right\rceil$ &
$\left\lceil\frac{5}{\varepsilon}\ln\frac{2|\cX|}{\beta}\right\rceil$
\end{tabular}
\caption{The running time and errors of ${\tt CompactHistogram}_{\cM, \cM_0}$ for the counting query algorithms of Section \vpref{sec:count} using the empty-bin sampler $\cM_0$ defined in Lemma \vpref{lem:count conversion} with $\log d_0 = O(\log(1/\varepsilon) + \log {d} + \log |\cX|)$.
For more details, see Theorem \vpref{thm:explicit compact}.}
\end{figure}

\section*{Acknowledgments}
We thank the Harvard Privacy Tools differential privacy research group, particularly Mark Bun and Kobbi Nissim, for informative discussions and feedback.
And we thank Ashwin Machanavajjhala, Frank McSherry, Uri Stemmer, and the anonymous TPDP and ITCS reviewers for their helpful comments.

\bibliography{hist}

\appendix

\section{Efficient Sampling of Distinct Elements}\label{sec:distinct sample}
In this section we show how to efficiently sample distinct elements from a subset of $\cX$.

\begin{lemma}\label{lem:distinct sample}
There exists an algorithm that given an integer $m$ specifying the set $\cX = [m]_+$ and a subset $A \subseteq \cX$, samples a uniformly random sequence of $r \le |\cX \setminus A|$ distinct elements from $\cX \setminus A$ with running time
\begin{align*}
O(|A| \log |A|\cdot \log |\cX|) + O(r \cdot \log^2 |\cX| \cdot \log (|A| + r))
\end{align*}
\end{lemma}

To prove this lemma, we will use a data structure that supports efficiently computing the numbers of elements in the tree less than a given value.
We will use this data structure to store the elements in $\cX$ which we do not want to sample.

\begin{proposition}[Order-Statistic Tree \cite{CoLeRiSt09} Chapter 14.1]\label{prop:redblack}
There exists a data structure $T$ maintaining a set over $\cX$ (let $|T|$ denote the size of the set maintained by $T$) with the following properties:
{\renewcommand{\theenumi}{\roman{enumi}}
\begin{enumerate}
\item
	Inserting an element $x \in \cX$ into $T$ takes $O(\log |T| \cdot \log |\cX|)$ time.
\item
	$T$ can be represented in $O(|T| \cdot \log |\cX|)$ bits.
\item
	For all $x \in \cX$, the quantity $|\{x' \in T \,:\, x' \le x\}|$ can be computed in $O(\log |T|\cdot \log |\cX|)$ time.
\end{enumerate}}
\end{proposition}

\begin{proof}[Proof of Lemma \vpref{lem:distinct sample}]
We define the sampling algorithm as follows:
\begin{figure}[H]
\rule{\textwidth}{.5pt}
\begin{algorithm}\label{alg:distinctsample}
${\tt DistinctSample}(A, r)$
for $A \subseteq \cX$ and $r \in \N_+$ such that $r \le |\cX \setminus A|$
\begin{enumerate}
\item
	Let $S = \emptyset$ and $T$ be the data structure defined in Proposition \vpref{prop:redblack}.
\item
	For $x \in A$, insert $x$ into $T$.
\item
	For $i \in [r]_+$:
	\begin{enumerate}
	\item
		Let $m' = |\cX|-|T|$.
	\item
		Sample $z$ uniformly at random form $[m']_+$.
	\item
		Perform binary search over $\cX$ to find $s = \min \{x \in \cX \,:\, x - |\{x' \in T \,:\, x' \le x\}|  = z\}$.
	\item
		Insert $s$ into $T$ and let $S = S \cup \{s\}$.
	\end{enumerate}
\item
	Return $S$
\end{enumerate}
\vspace{-.1in}
\rule{\textwidth}{.5pt}
\end{algorithm}
\end{figure}

We prove correctness by induction on $r$.
We start with the base case $r = 1$.
Notice that $x-|\{x' \in T \,:\, x' \le x\}| = |\{x' \in \cX\setminus A \,:\, x' \le x\}|$.
Thus for all $z \in [m']_+$ there exists $x \in \cX \setminus A$ such that $|\{x' \in \cX\setminus A \,:\, x' \le x\}| = z$.

Now, we show $s \in \cX \setminus A$.
Let $x \in A$.
If $|\{x' \in \cX\setminus A \,:\, x' \le x\}| = 0$, then $s \neq x$ as $z \ge 1$.
Otherwise $|\{x' \in \cX\setminus A \,:\, x' \le x\}| = |\{x' \in \cX\setminus A \,:\, x' \le (x-1)\}|$ which also implies $s \neq x$ by definition of $s$.
Therefore,  $s \in \cX \setminus A$.

Notice that two different values of $z$ cannot output the same $s$ and $m' = |\cX \setminus A|$.
Therefore, ${\tt DistinctSample}(A, 1)$ is uniformly distributed over $\cX\setminus A$.

For the induction step, let $S \sim {\tt DistinctSample}(A, 1)$ and assume ${\tt DistinctSample}(A', r-1)$ is uniformly distributed over random sequences of $r-1$ elements from $\cX \setminus A'$ for any $A' \subseteq \cX$.
Then
\begin{align*}
{\tt DistinctSample}(A, r) \sim S \cup {\tt DistinctSample}(A \cup S, r-1)
\end{align*}
Therefore ${\tt DistinctSample}(A, r)$ is uniformly distributed over random sequences of $r$ elements from $\cX \setminus A$.

Now, we analyze the running time of ${\tt DistinctSample}$.
Step 2 can be done in $O(|A| \log |A| \cdot \log |\cX|)$ time by Proposition \vpref{prop:redblack} Part {\it i}.
Each of the $r$ iterations of step 3 is dominated by step 3c
which takes $O(\log^2 |\cX| \cdot \log (|A| + r))$ time by Proposition \vpref{prop:redblack} Part {\it iii} as $|T| \le |A| + r$.
\end{proof}
\end{document}